\documentclass[a4paper,11pt]{article}
\usepackage{jheppub}     

\usepackage{amsfonts,amsthm}
\usepackage[utf8]{inputenc}
\usepackage{physics}
\usepackage{graphicx}
\usepackage{upgreek}  
\usepackage{mdframed}
\usepackage[usenames,dvipsnames]{xcolor}
\hypersetup{
      breaklinks=true,
      linkcolor=BlueGreen,
      urlcolor=NavyBlue,
      citecolor=Purple,
      filecolor=black,
      menucolor=black,
      pdfstartview=FitV,
      bookmarksopen=true
      }

\newcommand\kwar{\mathrm{war}}
\newcommand\kcon{\mathrm{con}}
\newcommand\krad{\mathrm{rad}}

\newcommand{\U}{\mathcal{U}}
\newcommand{\M}{\mathcal{M}}
\newcommand{\N}{\mathcal{N}}
\newcommand{\ucon}{\U^\kcon}
\newcommand{\uwar}{\U^\kwar}
\newcommand{\urad}{\U^\krad}

\newcommand{\teng}{\boldsymbol{g}}
\newcommand{\tenq}{\boldsymbol{q}}
\newcommand{\gw}{\teng^\kwar}
\newcommand{\gc}{\teng^\kcon}
\newcommand{\gr}{\teng^\krad}
\newcommand{\qw}{\tenq^\kwar}
\newcommand{\qr}{\tenq^\krad}
\newcommand{\HH}{\mathcal{H}}

\newcommand\vecu{\boldsymbol{u}}
\newcommand\vecW{\boldsymbol{W}}
\newcommand\vecX{\boldsymbol{X}}
\newcommand\vecZ{\boldsymbol{Z}}
\newcommand\vecxi{\boldsymbol{\xi}}
\newcommand\xiB{\xi_{\parallel}}
\newcommand\xiF{\xi_{\perp}}
\newcommand\vecxiB{\vecxi_{\parallel}}
\newcommand\vecxiF{\vecxi_{\perp}}
\newcommand\vecze{\boldsymbol{\zeta}}
\newcommand\dfom{\boldsymbol{\omega}}
\newcommand\dfeta{\boldsymbol{\eta}}

\newcommand\iEi{\mathtt{I}}
\newcommand\iEj{\mathtt{J}}
\newcommand\iEk{\mathtt{K}}
\newcommand\iLa{\mathfrak{a}}
\newcommand\iLb{\mathfrak{b}}

\newcommand\R{\mathcal{R}}

\newcommand\tengB{\boldsymbol{h}}
\newcommand\gB{h}
\newcommand\RB{\R^{(\tengB)}}

\newcommand\tengS{\boldsymbol{\upgamma}}
\newcommand\gS{\gamma}

\newcommand\dimF{m}

\newcommand\gradB{\mathrm{grad}_{\tengB}}
\newcommand\gradS{\mathrm{grad}_{\gS}}
\newcommand\Lie{\pounds}
\newcommand\ds{\tilde{\dd}}

\newtheorem{thm}{Theorem}[section]
\newtheorem{proposition}[thm]{Proposition}
\newtheorem{lemma}[thm]{Lemma}
\newtheorem{corollary}[thm]{Corollary}
\theoremstyle{definition}
\newtheorem{remark}[thm]{Remark}

\begin{document}

\title{Cosmological Spacetimes with Sign-Changing Spatial Curvature and Topological Transitions}

\author[1,2]{Gerardo Garc\'ia-Moreno,}
\affiliation[1]{Dipartimento di Fisica, Sapienza Università di Roma, Piazzale Aldo Moro 5, 00185, Roma, Italy}
\affiliation[2]{INFN, Sezione di Roma, Piazzale Aldo Moro 2, 00185, Roma, Italy}

\author[3]{Bert Janssen,}
\affiliation[3]{Departamento de F\'isica Te\'orica y del Cosmos and CAFPE, Facultad de Ciencias, Universidad de Granada, 18071 Granada, Spain}

\author[4]{Alejandro Jim\'enez Cano,}
\affiliation[4]{Escuela T\'ecnica Superior de Ingenier\'ia de Montes, Forestal y del Medio Natural, Universidad Polit\'ecnica de Madrid, 28040 Madrid, Spain}

\author[5]{Marc Mars,}
\affiliation[5]{Instituto de F\'isica Fundamental y Matem\'aticas,
Universidad de Salamanca
Plaza de la Merced s/n 37008, Salamanca, Spain}

\author[6]{Miguel S\'anchez,}
\affiliation[6]{Departamento de Geometr\'ia y Topolog\'ia \& IMAG, Universidad de Granada, 18071 Granada, Spain}

\author[7]{Ra\"ul Vera}
\affiliation[7]{Department of Physics and EHU Quantum Center, Euskal Herriko Unibertsitatea UPV/EHU, 48080 Bilbao, Basque Country, Spain}

\emailAdd{gerardo.garciamoreno@uniroma1.it}
\emailAdd{bjanssen@ugr.es}
\emailAdd{alejandro.jimenez.cano@upm.es}
\emailAdd{marc@usal.es}
\emailAdd{sanchezm@ugr.es}
\emailAdd{raul.vera@ehu.eus}

\newpage
\abstract{

Observational evidence, together with practical computations and modeling, supports a Euclidean spatial sector in the current cosmological model based on the FLRW metric. This, however, would imply that the total amount of matter and energy immediately after the Big Bang must have been infinite, an implication that could only be avoided through a transition from a closed to an open universe, a process forbidden in standard FLRW models. In this article, we investigate the spacetimes resulting from promoting the spatial curvature $k$ in FLRW spacetimes to a time-dependent function, $k \to k(t)$, allowing it to change sign and thereby allowing changes in the topology of the constant-$t$ slices. Although previously dismissed due to a classical theorem by Geroch, such transitions are shown to be consistent with global hyperbolicity when the comoving time is distinct from a Cauchy time, as recent work by one of the authors demonstrates. We construct three distinct geometries exhibiting this behavior using different representations of constant-curvature spaces. We analyze their global properties and identify mild conditions under which they remain globally hyperbolic. Furthermore, we characterize their Killing vectors, proving a general result for spherically symmetric spacetimes and compare them with known geometries in the literature.
}

\maketitle

\section{Introduction}
\label{Sec:Introduction}

Friedman-Lema\^itre-Robertson-Walker (FLRW) spacetimes offer three fundamental background  geometries, spherical, flat, and hyperbolic, arising from general symmetry considerations prior to specifying the energy-momentum tensor. These assumptions can be formulated with mathematical precision leading to a rigidity theorem characterizing these geometries (see, for instance, \cite[Ch. 12, Prop. 6]{ONeill1983}). However, many commonly cited physical notions of homogeneity and isotropy do not necessarily imply one of these standard cases. In fact, a recent example constructed by one of the authors \cite{Sanchez2023} and further examined by \'Avalos \cite{Avalos2022}, demonstrates the possibility of foliating those spacetimes with spatial hypersurfaces of constant curvature $k(t)$ (and hence, maximally symmetric), where the sign of $k(t)$ changes over time. 

The goal of this article is to show that such a scenario is not merely a mathematical curiosity, but a potentially valuable framework for modern cosmology. To support this, we construct two additional classes of examples, derived from conventional models of spacetimes with constant spatial curvature, and study their properties, highlighting certain features that may prove advantageous. We analyze both global, namely, under which conditions they are globally hyperbolic and nonsingular, as well as local properties, namely their isometries and curvature properties. Furthermore, we show that the three of them are not isometric among them, nor to some geometries previously introduced in the literature, specifically the Stephani and the Lema\^itre-Tolman-Bondi (LTB) geometries, since their energy-momentum tensor is that of a perfect fluid and in the $k(t)$ this implies FLRW, see Prop.~\ref{perfect-fluid is RW}. 

The first of the geometries, the {\em $k(t)$-warped} metric, is such that the point $r = 0$ is {\em extrinsically} privileged (indeed, its expansion is vanishing in the simplest case, as the second fundamental form vanishes therein) and, when $k(t)>0$, the antipodal point $r=\pi/\sqrt{k(t)}$ is also privileged (the expansion diverges towards this point). The metric is automatically smooth at $r=0$ but not at the $t$-antipodal points. In \cite{Sanchez2023}, the metric was deformed along a small region around $t$-antipodal points (maintaining the constant curvature of the $t$-slices) so that it becomes smooth everywhere. Here, we have not used the smoothening and we have analyzed the origin of non-smoothness: the scalar curvature diverges and the metric cannot even  be $C^1$-extended  to the $t$-antipodal points (Prop.~\ref{p_3metrics1}, Rem.~\ref{r_3metrics1}). The properties of this spacetime are interesting because, on the one hand, they do not depend on smoothening and, on the other, the smoothed spacetimes coincide with the $k(t)$-warped up to a region which can be chosen arbitrarily small.

Our second metric is called $k(t)$-\emph{conformal}, because its spatial metric is conformally isometric to the Euclidean one with conformal factor $4/(1+k(t) r^2)^2$. In addition, $t$-slices are totally umbilical. When $k(t)>0$ the metric can be smoothly extended to the point at infinity (which corresponds to stereographic compactification). So, a completely classic smooth  cosmological spacetime is obtained  (Prop.~\ref{p_3metrics2}).

For our third metric, which we called the $k(t)$-\emph{radial}, the spatial metric is constructed by changing the natural radial component $\dd r^2$ of  the Euclidean metric by $\dd r^2/(1-k(t))r^2$. When $k(t)>0$, the constant $t$-slices are not a whole round sphere, but an open subset of it, indeed, a half sphere  (Prop.~\ref{p_3metrics3}, Rem. \ref{r_3metrics3}). So, they admit a change in the sign of the spatial curvature, even if no topological change occurs. 

Our results reveal that these geometries can be considered as viable geometrical basis for cosmological models. Some of these metrics have previously been considered in models where the curvature of spatial slices evolves with time~\cite{Bergmann1981,Mersini2007,Stichel2018,Wang2025}, primarily to alleviate certain cosmological tensions. Furthermore, from a theoretical perspective, they have also been studied in connection with the backreaction problem~\cite{Larena2008,Clifton2024,Raffai2025}, which concerns how inhomogeneities and anisotropies, although vanishing when averaged over large scales, do not generally yield the same equations as those obtained by assuming an FLRW spacetime \textit{ab initio} and deriving the evolution of the scale factor from Einstein equations. Here, however, we perform the first exhaustive analysis of their local and global properties. In particular, we identify the conditions under which they are globally hyperbolic, a fundamental requisite for describing a cosmological model. Additionally, for certain choices of the defining functions, we find an explicit example that admits an extra Killing vector.

 \paragraph*{\textbf{Structure of the article.}}
The paper is structured as follows. In Sec.~\ref{Sec:CosmologicalModel} we introduce the notion of cosmological model and discuss its properties, making a special emphasis on the necessary distinction between the global and local ingredients. In Sec.~\ref{Sec:kmetrics} we introduce the spacetimes that we study in the article and present some of their local geometrical properties. In Sec.~\ref{Sec:extendibility} we discuss the possibility to extend the geometries to points not covered by the original coordinates. In Sec.~\ref{Sec:globalhyp} we analyze the global properties of these geometries with a special emphasis on characterizing under which conditions global hyperbolicity holds. In Sec.~\ref{Sec:isometries} we characterize the isometries of the geometries, exhaustively studying how many Killing vectors exhibit depending on the functions entering the metrics. In Sec.~\ref{sec:StephaniLTB} we compare these geometries with some previously studied models, the LTB and the Stephani metrics, and show that they are not isometric. Finally we summarize the main results of the article in Sec.~\ref{Sec:results}, and we conclude in Sec.~\ref{Sec:conclusions} by explaining the opportunities that these metrics present for current Cosmology and discussing follow-ups to this work. 

 \paragraph*{\textbf{Notation and conventions.}}
In this article, we use the signature $(-,+,...,+)$  for the spacetime metric and units in which $c = G_{\text{N}} = 1$. For the curvature tensors we use the conventions in \cite{Wald1984},  i.e., $[\nabla_a, \nabla_b] V^d =: - \R_{abc}{}^d V^c$, $\R_{ab}:=\R_{acb}{}^c$. We call $\R$ the Ricci scalar, $G_{ab}:=\R_{ab}-g_{ab}\R/2$ the Einstein tensor and $C_{abcd}$ the Weyl tensor.  

Dots represent partial derivatives with respect to the time coordinate (or total derivative in case the function only depends on time) and primes correspond to the partial derivatives with respect to the radial coordinates: $\dot{X}:=\partial_t X$, $X':= \partial_r X$.

In the tensor expressions of the different metrics, we abbreviate $\dd X^2 := \dd X \otimes \dd X$ and $\dd X \dd Y:= {\frac12}  (\dd X\otimes  \dd Y +  \dd Y\otimes  \dd X)$.

Boldface symbols correspond to tensors, vector fields and differential forms in index-free notation.

We work in $n+1$ spacetime dimensions. We use $\tengB$ for the metric of the 2-dimensional base manifold of a warped product and $\tengS$ for the standard $\mathbb{S}^m$ metric, where $m:=n-1$ is the dimension of the fiber $\mathbb{S}^{m}$.

The symbol $\mathbb{L}^{n+1}$ will represent the Minkowski spacetime ($\mathbb{R}^{1,n}$ equipped with the flat metric).

Finally we summarize our index notation:
\begin{itemize}
    \item $a, b, c, \ldots$ are \textit{abstract indices} in $n+1$ dimensions (used essentially to write any tensor identity without involving partial derivatives, connections, etc.). These do not refer to any specific basis and can apply to coordinate bases, tetrads, etc.

    \item $\mu, \nu, \rho, \ldots$ are \textit{indices in a specific coordinate basis} in $n+1$ dimensions. In our case, we use them to refer to components in the chart $\{x^\mu\} = \{t, r, \theta^A\}$ (see below for the definition of $\theta^A$).

    \item $i, j, k, \ldots$ are \textit{indices in a specific coordinate basis}  $\{y^i\} = \{t, r\}$ on the \textit{2-dimensional base space}.

    \item $A, B, C, \ldots$ are \textit{indices on the sphere} $\mathbb{S}^m$, and correspond to the hyperspherical coordinates (denoted as $\theta^A$).

    \item $\iEi, \iEj, \iEk$ are \textit{indices in the Euclidean $(m+1)$-dimensional space} where $\mathbb{S}^m$ is canonically embedded.

    \item $\iLa, \iLb$ are \textit{indices in the Killing algebra} of the base space and run from $1$ to $\mathfrak{n}$.
\end{itemize}

\section{Global and local roles of time in cosmological spacetimes}
\label{Sec:CosmologicalModel}

%
\subsection{Defining a cosmological spacetime}

To start with, it is convenient to recall that a cosmological model should not simply be thought as a spacetime, i.e., a manifold with a Lorentzian metric, but it also needs two additional key ingredients: a description of matter and radiation, and a uniquely defined set of fundamental (also called comoving) observers $\vecu$ that are to describe the average motion of matter in the universe, in terms of which observational relations are set between the geometry and the matter content \cite{Ellis2012}. Keeping only the geometrical ingredients, and in a slightly simplified fashion enough for the purposes of this work, following~\cite{Choquet-Bruhat2009} (see also~\cite{Avalos2022}) we say a Lorentzian manifold $(\mathcal{V}^{n+1},\teng)$, $n \geq 3$, is a \emph{cosmological spacetime} if
\begin{itemize}
    \item $\mathcal{V}^{n+1} = I \times \Sigma^n$, where $I$ is a connected open interval of $\mathbb{R}$, which we parametrize by some coordinate $t$,
    \item $\vecu = \partial_t$ is orthogonal to the leaves $\{t\}\times \Sigma^n$ and $\teng(\vecu,\vecu) = -1$. The vector field $\vecu$ represents the comoving observers, which usually are assumed to be  freely-falling (this will hold in our models too).
\end{itemize}
In fact, this means that a single Lorentzian manifold can potentially give rise to different cosmological spacetimes, each one associated with different choices of comoving observers. 

\paragraph*{\textbf{de Sitter spacetime as a playground for different cosmological spacetimes.}}
We can illustrate this with a simple example. We define the $(n+1)$-dimensional de Sitter spacetime as the submanifold embedded in $\mathbb{L}^{n+2}$ described by 
\begin{align}
    - \left( X^{0} \right)^2 + \sum_{i=1}^{n+1} \left( X^{i} \right)^2  = \ell^2,
\end{align}
with the induced metric inherited from the flat metric
\begin{align}
        \teng_{\mathbb{L}^{n+2}} = - (\dd X^0)^2 + \sum_{i=1}^{n+1}(\dd X^i) ^2.
\end{align}
The de Sitter spacetime is a solution of the Einstein equations in vacuum with positive cosmological constant, where $\Lambda = n(n-1)/(2 \ell^2)$. 

There are different coordinate systems that bring de Sitter to an FLRW form. In particular, we can take the closed slicing, which foliates de Sitter with $n$-dimensional spheres, by setting
\begin{align}
    X^0 & = \ell \sinh \left( \frac{T}{\ell} \right), \\
    X^1 & = \ell \cosh \left( \frac{T}{\ell} \right) \cos(\vartheta), \\
    X^i & = \ell \cosh \left( \frac{T}{\ell} \right) \sin(\vartheta)\, \hat{n}^i (\theta^A) , \qquad i = 2, ..., n+1,
\end{align}
where $\hat{n}^i$ represents the outward unit vector orthogonal to $\mathbb{S}^{n-1}\subset \mathbb{R}^{n}$ in terms of the coordinates $\{\theta^A\}$ of the sphere. The range of the time coordinate is $T \in \mathbb{R}$, $\vartheta \in (0, \pi)$ is a colatitude angle on the $n$-sphere and, as already mentioned, $\{\theta^A\}$ are standard coordinates in the $(n-1)$-spheres. 

The metric in these coordinates reads
\begin{equation}
    \teng_{\text{dS}} 
    = - \dd T^2 + \ell ^2 \cosh^2 \left ( \frac{T}{\ell} \right) \left [  \dd \vartheta^2 + \sin^2 (\vartheta) \tengS\right ],
\end{equation}
where we notice that the part in square brackets is the standard metric on the $n$-dimensional sphere $\mathbb{S}^n$. This coordinate representation of the de Sitter spacetime exhibits a foliation in which the area of the $(n-1)$-dimensional spheres tends to infinity for $T \to - \infty$, monotonically decreases to a minimum at $T =0$, and then monotonically increases all the way to infinity as $T \to + \infty$.

On the other hand, we can take the flat slicing in terms of a set of coordinates $(t, x^\iEi) \equiv (t, r, \theta^A) $, where $\{x^\iEi\}$ represent the Cartesian coordinates in Euclidean space and $(r, \theta^A)$ are the corresponding radial and hyperspherical coordinates, by fixing
\begin{align}
       X^0 & = \ell \sinh \left( \frac{t}{\ell} \right) + \frac{r^2}{2 \ell} e^{t/\ell} ,\\
       X^1 & = \ell \cosh \left( \frac{t}{\ell} \right) - \frac{r^2}{2 \ell} e^{t/\ell}, \\
       X^i & = e^{t/\ell} r \, \hat{n}^i (\theta^A) , \qquad i = 2, ..., n+1 .
\end{align}
The metric in these coordinates is given by
\begin{align}
    \teng_{\text{dS}} = - \dd t^2 +  e^{2t/\ell} (\dd r^2 +  r^2 \tengS) . 
\end{align}
This covers the region $X^0 + X^1 > 0$ of the hyperboloid with planes $t=\text{constant}$ intersecting it  at $45^\circ$ (see \cite[Sec. 5.2]{HawkingEllis}). To cover the remaining part, we can choose another similar patch of coordinates, resulting simply in an exchange of the time coordinate $t \to -t$ in the metric. 

\begin{figure}
    \centering
    \includegraphics[scale=1.4]{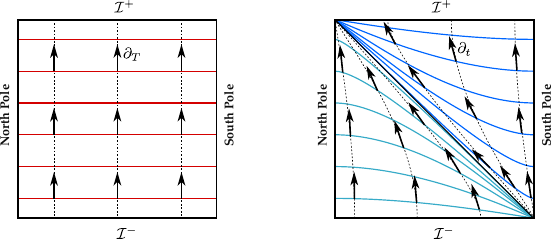}
    \caption{
    Foliations of the de Sitter spacetime ($n \geq 2$) represented in a Penrose diagram: the spherical foliation by constant-$T$ slices (left), which covers the entire spacetime, and the foliation by flat constant-$t$ slices (right), which covers only half of it. Note that, in the latter case, there exists an analogous foliation of the other half of the spacetime, but the two cannot be smoothly joined to form a global one, since the surface $X^0 + X^1 = 0$ is null (solid black line in the right panel). The arrows represent the corresponding comoving observers. Dashed lines represent constant-$\vartheta$ ($r$) in the left (right) panel.} 
    \label{fig:deSitter_foliations}
\end{figure}

From these two foliations, we see that within the same manifold we can take two (actually more) vector fields $\vecu$ that represent the comoving time, namely $\partial_T$ and $\partial_t$. Clearly, both vectors are different, as the two sets of coordinates are related through
\begin{align}
    e^{ t /\ell} & = \sinh \left( \frac{T}{\ell} \right) + \cos (\vartheta)  \cosh \left( \frac{T}{\ell} \right),  \\
    e^{t/\ell} r & = \ell \cosh \left( \frac{T}{\ell} \right) \sin (\vartheta),
\end{align}
while $\theta^A$ remain unchanged. In other words, different sets of observers can be taken as the fundamental, or comoving, ones. Even though this can be done in de Sitter spacetime due to its high level of symmetry, it highlights the crucial role that the choice of comoving time plays in defining a cosmology.

Another interesting example is the Milne spacetime which also stresses the importance of the comoving observers. In fact, it shows that even an open subset of Minkowski spacetime (the chronological future of a point) can be regarded as a non-trivial cosmological model.

 \paragraph*{\textbf{Predictability.}}
In addition to this, a cosmological spacetime should also fulfill a global fundamental property: \emph{predictability}. This feature is adequately captured by the notion of \emph{global hyperbolicity}, which can be defined in different ways such as unique evolution from initial data posed on  a Cauchy hypersurface, absence of naked singularities and others (see \cite{Sgeroch,Spenrose})\footnote{
    When global hyperbolicity fails, one should add ``information at conformal infinity'', as occurs in the case of anti-de Sitter spacetime. This situation can be handled  with the notion of globally hyperbolic spacetime-with-timelike-boundary developed in \cite{AFS, Solis}. Here we will focus in the case without boundary for simplicity (but see footnote~\ref{fnote:kradboundary} in Rem.~\ref{r_radial} regarding the $k(t)$-radial case).}.
A celebrated theorem by Geroch \cite{Ge} refined by Bernal and S\'anchez \cite{BS03, BS05}, proves that a globally hyperbolic spacetime admits a smooth Cauchy temporal  function $\tau$,  so that the spacetime manifold splits globally as an orthogonal product $(I\times \mathcal{S}, \teng)$. Here, $I\subset \mathbb{R}$ is an interval and  the metric takes the form
\begin{equation}
    \label{e_0intro}
    \teng=- \lambda(\tau,x) \dd \tau^2 + \teng_\tau[x],
\end{equation}
where each slice $\{\tau_0\}\times \mathcal{S}$, $ \tau_0\in I$, inherits a positive-definite metric $\teng_{\tau_0}$ and turns out to be a smooth spacelike Cauchy hypersurface. 

From a purely mathematical viewpoint, there is a huge freedom to choose such  a $\tau$.\footnote{\label{foot1}
    This was stressed in \cite{BS06} where, for example, any compact acausal spacelike submanifold with boundary is shown to lie in a $\tau$-slice for some Cauchy temporal $\tau$. The freedom in $\tau$ had already led Andersson et al. to propose a way to construct a ``standard cosmological time'' \cite{Aetal} for quite general spacetimes.}   
Notice that the splitting in Eq.~\eqref{e_0intro} and the one appearing in the definition of cosmological spacetime do not necessarily need to agree. They do agree in the standard FLRW models, since the comoving time plays not only a global role (constant $t$-slices are Cauchy and provide a global splitting), but also a local role as the slices are maximally symmetric and hence have constant curvature.  As mentioned above, the freedom to choose a Cauchy temporal function is huge and, thus, this local symmetry condition which is sometimes linked to physical magnitudes (as it happens in the FLRW case, see for example \cite{RZ}) selects a specially appropriate time, typically unique. 

\subsection{Global and local time functions are independent}

Coming to our example in de Sitter spacetime, the closed slicing gives rise to a globally defined foliation in which the constant comoving time surfaces are Cauchy hypersurfaces (they are compact). The flat slicing though, is such that the comoving time gives rise to noncompact hypersurfaces which clearly cannot be Cauchy surfaces of the de Sitter spacetime as defined above (since all Cauchy surfaces are homeomorphic). In that sense, the example shows that both times do not necessarily need to agree. However, this is actually pathological from the global point of view, as the $\partial_t$ vector field cannot be extended continuously at $t \to - \infty$ and hence it is not possible to smoothly connect it with another vector field defined in the other half-part of de Sitter. 

In fact, time functions that attempt to smoothly interpolate from the closed to the open slicing in de Sitter spacetime have been previously reported by Krasi\'nski taking coordinates in which the metric belongs to the Stephani family~\cite{Krasinski1981,Krasinski1983}, see also~\cite{Krasinski1997, Dabrowski1991}. Furthermore, other solutions describing less symmetric spacetimes that in principle allow for this dissociation of roles have been reported~\cite{Cook1975, Sussman1988}. 

This already suggests that the twofold role played by time in FLRW spacetimes should not be taken as mandatory, as it is actually not needed. In fact, we want to highlight the inherent independence between two different properties implicitly assumed for the ``spatial part'' of  cosmological spacetimes so far:
\begin{itemize}
    \item[(a)] Existence of a temporal function with a constant curvature slicing (motivated by  considerations on  intrinsic isotropy and homogeneity).
    \item[(b)] Existence of a Cauchy temporal function $t$  (motivated by considerations on predictability), that is, inextensible causal curves, which  represent particles, must cross each spatial $t$-slice exactly once. 
\end{itemize}
In fact, whereas the topology of Cauchy slices (b) is fixed, there is no requirement that the topology of the constant curvature slices (a) needs to be fixed. This means that models in which the topology of the constant curvature slices changes (our $k(t)$ metrics allow this possibility) will admit  some temporal functions fulfilling  (a) and others fulfilling (b) but none fulfilling both.

\subsection{Current Physics might require the  disassociation of roles}
The different roles  for the spatial slicing pointed out in (a) and (b) above are clearly   distinguished from a purely physical viewpoint.  
   
The property (a) is a strong assumption with local nature. In principle, the time which produces such a constant curvature slicing should be then tied to some locally measurable physical process or magnitude.

However, the property (b) is  a mathematical consequence of predictability, which has a  philosophical rather than  physical nature: we may believe that  the present determines the future (in a mathematically precise  way), but  only  omniscent observers could check it.  Indeed, very general  acausal subsets become a part of a  slice for some Cauchy temporal function $\tau$ (recall footnote~\ref{foot1}) but no local measurement could determine the whole slice. 

More down-to-earth, the dissociation of both roles in some of our cosmological spacetimes with changing $k(t)$ permits to model appealing cosmological possibilities, such as a Big-Bang in a closed model which may seem an open model for comoving observers later.

\section{Construction of the \texorpdfstring{$k(t)$}{k(t)} metrics and local curvature properties}
\label{Sec:kmetrics}

We begin this section by reviewing the FLRW metric, as the various possible coordinate choices on the spatial slices serve as the starting point for constructing the metrics analyzed in this article. 

\subsection{FLRW metrics in several coordinate systems}

The metric of FLRW spacetimes in comoving coordinates takes the form
\begin{align}
   \teng_{\text{FLRW}} = - \dd t^2 + a(t)^2 \tilde{\teng}_{t}\,,
\end{align}
where $t$ is the comoving time, $a(t)$ the scale factor, and $\tilde{\teng}_{t}$ the metric of constant (time independent) curvature spatial sections. It is well known that there are only three possibilities,\footnote{
    Up to taking the quotient by a discrete subgroup of the symmetry group, effectively changing the topology of the spatial slices~\cite{Lachieze-Rey1995}.}
depending on the sign of the curvature: in three dimensions positive constant curvature spaces are $3$-spheres $\mathbb{S}^3$, negative constant curvature spaces are hyperboloids $\mathbb{H}^3$ and the zero curvature space is flat Euclidean space $\mathbb{R}^3$. Concretely, for $k>0$ the metric $\tilde{\teng}_{t}$ describes a $3$-sphere with curvature radius $1/\sqrt{k}$, for $k<0$ a hyperbolic space with curvature radius $1/\sqrt{-k}$ and $k=0$ is Euclidean space. 

Different coordinate systems can be used to describe the spatial hypersurfaces, for instance we have the three following possibilities:
\begin{align}
        \teng_{\text{FLRW}} &= -\dd t^2 + a(t)^2 \left(\frac{1}{1-k r^2}\dd r^2 + r^2 \tengS \right) \label{Eq:FLRW1} \\
     &= -\dd t^2 + a(t)^2 \left(\dd \chi^2 + F(\chi)^2 \tengS \right) \label{Eq:FLRW2} \\
     &= -\dd t^2 + \frac{4a(t)^2}{(1+k R^2)^2} \Big( \dd R^2 + R^2 \tengS \Big) \,, \label{Eq:FLRW3}
\end{align}
where the different radial coordinates are related by:
\begin{align}
    r=S_k(\chi) = \frac{2R}{1+kR^2}\,,
\end{align}
and with the function $S_k(\chi)$ defined as
\begin{equation}\label{eq:defSk}
    S_k(\chi):= \begin{cases}
        \dfrac{\sin(\sqrt{k}\, \chi)}{\sqrt{k}}    & \text{if} \; k>0 \\
        \chi                                       & \text{if} \; k=0 \\
        \dfrac{\sinh(\sqrt{-k}\, \chi)}{\sqrt{-k}} & \text{if} \; k<0 
    \end{cases}\,.
\end{equation} 
A rigidity theorem for FLRW is presented in~\cite{ONeill1983}, where it is shown that FLRW can be characterized as the set of cosmological spacetimes that are isotropic, where isotropy should be understood in the sense of spacetime isotropy as defined in~\cite{Avalos2022}. 

As carefully analyzed in~\cite{Avalos2022}, it has sometimes been incorrectly stated in the literature that FLRW spacetimes are simply those that are spatially homogeneous and isotropic (observe that now ``isotropy'' is understood in the sense of spatial isotropy), that is, that each constant-time slice inherits an induced metric that is maximally symmetric. In fact, we can see that the existence of a foliation with maximally symmetric slices does not lead to the rigidity of FLRW spacetimes by providing three straightforward counterexamples. Any promotion of the intrinsic curvature of the spatial slices in Eqs.~\eqref{Eq:FLRW1}-\eqref{Eq:FLRW3} to a time function $k \to k(t)$ leads to a counterexample. Indeed, among other local and global properties, we will see that none of the three resulting spacetimes is isometric to a patch of FLRW or each other (previous examples appear in \cite{Sanchez2023}, see also  \cite{MarsVera2024}).

\subsection{Introducing the \texorpdfstring{$k(t)$}{k(t)} metrics}
\label{ssec:def_kt_metrics}

Let $I\subset \mathbb{R}$ be an open interval and $a, k: I \rightarrow \mathbb{R}$ any smooth functions with $a >0$. Consider the following metrics on some open subsets $\U$ of $I\times\mathbb{R}^n$, which are defined by taking spherical coordinates $(r,\theta^A)$ in $\mathbb{R}^n$ (the convention $1/\sqrt{C}=\infty$ if $C\leq 0$ is used):

\begin{enumerate}

\item  The $k(t)$-warped metric
\begin{equation}
    \gw=-\dd t^2+a(t)^2 \Big(\dd r^2+ S_{k(t)}^2(r) \tengS \Big) \,, 
\end{equation}
where $\tengS$ is the standard metric on $\mathbb{S}^{n-1}$, the domain of the coordinates is 
  \begin{equation}
    \uwar=\left\{(t,r,\theta^A): t \in I, 0<r<\frac{\pi}{\sqrt{k(t)}}\right\}
\end{equation}
and the function $S_k(r)$ has been introduced in Eq.~\eqref{eq:defSk}. 

\item  The $k(t)$-conformal metric:
\begin{equation}\label{e_ucon}
    \gc=-\dd t^2 + \dfrac{4a(t)^2}{\left(1+k(t) r^2\right)^2} \left(\dd r^2+  r^2 \tengS\right)\,,
\end{equation}
where the $\{ t, r \}$ take values in
\begin{equation}
    \ucon= \left\{(t,r,\theta^A): t \in I, 0<r < \frac{1}{\sqrt{-k(t)}}\right\}\,.
\end{equation}

\item  The $k(t)$-radial metric:
\begin{equation}
    \gr=-\dd t^2 + a(t)^2 \left(\frac{1}{1-k(t) r^2} \dd r^2+  r^2 \tengS\right)\,,   
\end{equation}
where the $\{ t, r \}$ take values in
\begin{equation}
    \urad=\left\{(t,r,\theta^A): t \in I, 0<r<\frac{1}{\sqrt{k(t)}}\right\}\,.
\end{equation}
\end{enumerate}
In the three cases we will collectively denote by $x = (r,\theta^A)$ the spacelike coordinates. 

We now turn our attention to their local geometric properties, focusing on the curvature tensors and the behavior of the congruence of integral curves of the geodesic vector field $\partial_t$. The discussion of whether these geometries can be extended to the boundary of the interval over which $r$ is defined on each case will be postponed until Sec.~\ref{Sec:extendibility}.

\subsection{Curvature decomposition for the \texorpdfstring{$k(t)$}{k(t)} metrics}

The geometries given by the $k(t)$ metrics have the structure of a warped product with fiber $\mathbb{S}^{n-1}$ and base manifold $\N$ an open neighborhood of $\mathbb{R}^2$. To be precise, the total manifold admits a local product decomposition $\N\times \mathbb{S}^{n-1}$, and the metric tensor has the form
\begin{equation}
    \teng = \tengB + f^2 \tengS \,,
    \label{eq:warpedprod}
\end{equation}
with
\begin{align}
    \tengB &= \gB_{ij}(y) \dd y^i \dd y^j = -\dd t^2 + W^2 \dd r^2 \,, \label{eq:basespace}\\
    \tengS &= \gS_{AB}(\theta) \dd \theta^A \dd \theta^B\,.
\end{align}
The warping functions $f$ are given by
\begin{align}
        f(t,r) = \begin{cases}
          a(t) S_{k(t)} (r) &  \text{$k(t)$-warped metric} \\
          a(t) \dfrac{2r}{1+k(t)r^2} &  \text{$k(t)$-conformal metric}   \\
          a(t) r  &  \text{$k(t)$-radial metric}
        \end{cases} \quad
    \label{Eq:Warp_Function}
\end{align}
and the functions $W$ are
\begin{align}
        W(t,r) =  \begin{cases}
            a(t) &  \text{$k(t)$-warped metric}\\
            a(t) \dfrac{2}{1+k(t)r^2} &  \text{$k(t)$-conformal metric} \\[5mm]
            a(t)\dfrac{1}{\sqrt{1-k(t)r^2}} &  \text{$k(t)$-radial metric}   
        \end{cases} \quad .
    \label{Eq:Base_warped}
\end{align}
Given that the fiber is an $(n-1)$-sphere of unit radius (hence maximally symmetric with constant curvature equal to one), the Ricci tensor of the fiber is
\begin{align}
      \R^{(\tengS)}_{AB} =(n-2) \gS_{AB}\,.
\end{align}
On the other hand, the components of the Ricci tensor and the Ricci scalar of the base space can be easily computed:
\begin{equation}
    \RB_{tt} = - \frac{\ddot{W}}{W}\,, \qquad
    \RB_{tr} = 0\,,\qquad
    \RB_{rr} =  W \ddot{W}\,,
\end{equation}
\begin{equation}
    \RB = 2 \frac{\ddot{W}}{W}\,.\label{eq:CurvScal_B} 
\end{equation}

Now we can make use of well-known expressions for the Ricci tensor of a warped product (see e.g.~\cite[Cor. 7.43]{ONeill1983}) to find the components of the Ricci tensor of the total space:
\begingroup
\allowdisplaybreaks
\begin{align}
    \R_{tt} &=  - \frac{\ddot{W}}{W} - (n-1) \frac{\ddot{f}}{f} \,, \label{eq:R_WF1} \\
    \R_{tr} &= (n-1)\left[  \frac{\dot{W} f'}{Wf} -  \frac{\dot{f}'}{f} \right]\,, \\
    \R_{rr} &= W\ddot{W}  + (n-1) \left[W\dot{W}\frac{ \dot{f}}{f} + \frac{W' f'}{Wf} - \frac{f''}{f} \right]\,, \\
    \R_{tA} &= \R_{rA} =0\,, \\
    \R_{AB}  &=  \left[ f\ddot{f} + \frac{\dot{W}}{W}f\dot{f}  - \frac{ff''}{W^2} + \frac{W'f f'}{W^3} + (n-2) \left( 1+ \dot{f}^2  - \frac{f'^2}{W^2}\right)\right] \gS_{AB}\,.
\end{align}
\endgroup
The Ricci scalar is then given by
\begin{align}\label{eq:R_WF6}
    \R = 2\frac{\ddot{W}}{W} + 2(n-1)\left(\frac{\ddot{f}}{f}- \frac{f''}{W^2f} + \frac{\dot{W}\dot{f}}{ Wf}   + \frac{W' f'}{W^3f}\right) + \frac{(n-1)(n-2)}{f^2} \left(1 + \dot{f}^2 - \frac{f'^2}{W^2}\right) \,.
\end{align}
The explicit expressions for the specific metrics $\gc$, $\gw$ and $\gr$ can be found in App.~\ref{app:curvatures}.

\subsection{Congruences}
\label{SSec:congruences}

We can take the vector $u^a = (\partial_t)^a$ and analyze the behavior of the congruence of its integral curves. We start by constructing the tensor:
\begin{align}
    B_{ab} := \nabla_b u_a,
\end{align}
Since the one form $u_a$ is closed ($\dd\boldsymbol{u}=0$) it follows that $B_{ab}$ is symmetric. From $u$ being unit we get $B_{ab} u^a =0$, so $B_{ab}$ is purely spatial,
i.e., $u^a B_{ab} = u^a B_{ba} = 0$, and $u^a$ is geodesic. The information about the infinitesimal behavior of the congruence is encoded in the irreducible parts of $B_{ab}$ which are called vorticity, shear and expansion and are respectively given by
\begin{align}
    \omega_{ab} &:= \frac{1}{2}(B_{ab}-B_{ba})\,,\\
    \sigma_{ab} &:= \frac{1}{2}(B_{ab}+B_{ba}) - \frac{1}{n} \Theta \mathfrak{s}_{ab}\,,\\
    \Theta &:= g^{ab} B_{ab}\,,
\end{align}
where $\mathfrak{s}_{ab}:= g_{ab} + (\dd t)_a (\dd t)_b$ is the spatial part of the metric. In our case, the vorticity vanishes identically and the components $B_{\mu\nu}$ in the chart $\{x^\mu\}=\{t, r,\theta^A\}$ are purely spatial and read:
\begin{align}
    B_{\mu\nu} = \partial_\nu u_\mu - \Gamma^{\rho}{}_{\nu\mu} u_\rho = \Gamma^{t}{}_{\nu\mu} = \frac{1}{2} \partial_t \mathfrak{s}_{\mu\nu}\,.
\end{align}
Using $\mathfrak{s}_{\mu\nu} \dd x^\mu \dd x^\nu = W^2(t,r) \dd r^2 + f^2 \gS_{AB} \dd \theta^A \dd \theta^B \,,$ the only nontrivial components $B_{\mu\nu}$ are
\begin{equation}
    B_{rr} = W \dot{W}\,, \qquad
    B_{AB} = f\dot{f}  \gS_{AB}\,.
\end{equation}
Thus, the expansion, the non-vanishing components of the shear tensor and the shear scalar $\sigma^2:= \sigma^{\mu\nu} \sigma_{\mu\nu}$ read
\begin{align}
    \Theta & = \frac{\dot{W}}{W} + (n-1) \frac{\dot{f}}{f}\,, \quad \\
    \sigma_{rr} &= \frac{n-1}{n} W^2  \left(\frac{\dot{W}}{W} - \frac{\dot{f}}{f}  \right) , \qquad    \sigma_{AB} = -\frac{1}{n} f^2\left( \frac{\dot{W}}{W}-\frac{\dot{f}}{f} \right) \gS_{AB}\,, \\
    \sigma^2 &  = \frac{n-1}{n} \left ( \frac{\dot{W}}{W} - \frac{\dot{f}}{f} \right )^2\,.
\end{align}
Notice that the shear vanishes  in the FLRW case, as it should. Inserting the specific form of $f$, $W$ one finds the expansion and shear for each of the metrics. We collect the results in App.~\ref{app:congruence}.

Finally, we present another relevant quantity that will be used in the following to analyze the spacetime close to singular regions. As a measure of the tidal forces, we consider the relative acceleration of two nearby test particles that follow geodesics of the comoving observers $\partial_t$. This can be computed via the geodesic deviation equation given by
\begin{equation}
    \frac{\nabla^2 V^a}{\dd t^2} = \R_{bcd}{}^a V^b (\partial_t)^c (\partial_t)^d\,,
\end{equation}
where $V^a$ is a deviation vector, satisfying $[\boldsymbol{V}, \partial_t]=0$. If we focus on particles separated along the $r$ direction, at the initial time of deviation $t_0$ we find 
\begin{equation}
    \label{eq:geodesicdev}
    \left .     \frac{\nabla^2 V^\mu}{\dd t^2}  \right |_{t=t_0} =
    \left . \R_{rtt}{}^\mu \right |_{t=t_0},
\end{equation}
since $V^\mu |_{t_0} =(\partial_r)^\mu= \delta^\mu{}_r$ in the coordinates $\{t,r,\theta^A\}$. All these expressions are valid for the three $k(t)$ metrics.

\subsection{Singularities}

In the three cases, the metrics were defined on an open subset of $I\times \mathbb{R}^n$ in Sec.~\ref{ssec:def_kt_metrics}, and they are smoothly extensible to $r=0$. However curvature singularities may appear at the maximum value of $r$ making the metric inextensible therein (at least as a $C^2$ metric). In the following we analyze case by case and study in detail the differences between them.

\subsubsection{The \texorpdfstring{$k(t)$}{k(t)}-warped case}

\begin{proposition}\label{prop:kwar_sing}
    The $k(t)$-warped metric with $k(t)>0$ and $\dot{k}(t)\neq 0$ exhibits a curvature singularity at the antipodal point $r \to \pi/ \sqrt{k(t)}$.
\end{proposition}
\begin{proof}
  Consider a time interval in which $k(t)>0$ and $\dot{k}(t)\neq 0$. Using \eqref{eq:RicciScalarkwar}, we find the following expression:
\begin{equation}
     \R= \alpha_0(t) + \alpha_1(t) \hat{r}^2 + \alpha_2(t) \hat{r} \cot(\hat{r}) + \alpha_3(t) \hat{r}^2 \cot^2(\hat{r}),\qquad \hat{r} := \sqrt{k(t)} r\,,\label{eq:Rkwar_kpositive}
\end{equation}
for certain functions $\alpha_0$, $\alpha_1$, $\alpha_2$ and $\alpha_3$. The antipodal point corresponds to the limit $\hat{r}\nearrow \pi$ (at a fixed time). As we approach such a limit for a fixed time, we find the leading asymptotic behavior is
\begin{equation}
    \R \ \sim\  \alpha_3  \hat{r}^2 \cot^2(\hat{r}) \qquad (\hat{r}\nearrow \pi)\,,
\end{equation}
where
\begin{equation}
    \alpha_3 = \frac{(n-1)(n-2)}{4}\frac{\dot{k}^2}{k^2}\neq 0
\end{equation}
due to the hypothesis $\dot{k}\neq 0$. Consequently, the Ricci scalar diverges as we approach the antipodal point.
\end{proof}
However, following the approach in \cite[Sec. 4]{Sanchez2023}, it is possible to introduce a nonsingular geometry (a \emph{smoothening}) that agrees with the $k(t)$-warped line-element everywhere but at a small open neighborhood near the antipodal point and preserves the global properties of the spacetime (more on this in Prop.~\ref{p_3metrics1} below).

The geodesic deviation equation \eqref{eq:geodesicdev} for the $k(t)$-warped metric reads, at any initial time $t_0$,
\begin{align}
  \left .     \frac{\nabla^2 V^\mu}{\dd t^2} \right |_{t=t_0} =
  \left . \frac{\ddot{a}}{a}   \right |_{t=t_0} \delta^\mu{}_r \, .
\end{align}
This tidal acceleration is finite everywhere. This shows the mildness of this singularity and the possibility to smoothen it as in~\cite{Sanchez2023} in a simple way (in that reference $a(t)$ was constant so that the  geodesic deviation identically vanished). After the smoothening, there still remains an unremovable  singularity that appears along a comoving observer lying at the antipodal point as $k(t) \searrow 0$. Indeed, even when, say, $k(t_0)=0$ and $\dot k(t_0)<0$, the antipodal point must be removed at $t= t_0$, which underlies the topological change in the spatial part.

The singularity is similar to the one that appears in the $k(t)$-conformal case, but in a somehow more symmetric way: the space around the comoving observer would be growing so fast that they would need to disappear in finite proper time.

\subsubsection{The \texorpdfstring{$k(t)$}{k(t)}-conformal case}

\begin{figure}
    \centering
    \includegraphics[scale=1.4]{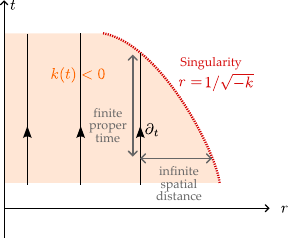}
    \caption{Spacetime diagram that shows the situation of the cosmological comoving observer of the $k(t)$-conformal metric hitting the curvature singularity at $r\to 1/\sqrt{-k(t)}$ whenever $|k(t)|$ increases in a certain interval. As indicated in the picture, such a singularity is always at an infinite spatial distance (along the $\partial_r$ direction) from any of these observers but is reached in a finite proper time for any of them at sufficiently large distance from the origin $r=0$. 
    }
    \label{fig:kcon_obser_sing}
\end{figure}

\begin{proposition}\label{prop:kconf_Weylzero}
    The $k(t)$-conformal metric is locally conformally flat.
\end{proposition}
\begin{proof}
    This can be shown by directly computing the Weyl tensor which turns out to be identically vanishing (independently of the spatial dimension $n\geq 3$).
\end{proof}

As a consequence of the previous result, all the curvature properties are hence encoded in the Ricci tensor. 
\begin{proposition}\label{prop:kconf_sing}
    Consider a $k(t)$-conformal metric with $k(t)<0$ and $\dot{k}(t)\neq 0$.
    \begin{itemize}
        \item[(1)] The spacetime exhibits a curvature singularity at the antipodal point $r \to 1/ \sqrt{-k(t)}$. 

        \item[(2)] If $\dot{k} <0$, the comoving observers $\partial_t$ at sufficiently high $r$ reach the singularity in finite proper time.

        \item[(3)] The spatial distance along $\partial_r$ curves between an event at $t=t_0$ and \mbox{$r=r_0< 1/\sqrt{-k(t_0)}$} and the singularity is infinite.
    \end{itemize}
\end{proposition}
\begin{proof} ~
\begin{itemize}
    \item[(1)] Immediate from the expression of the Ricci scalar \eqref{Eq:Ricci_kconformal}, which blows up at that point as long as $\dot{k} \neq 0$ and $k<0$. 
    \item[(2)] Fix any time $t_1$.  Since $\dot{k}<0$, the minimum of $|k|$ in the domain $t< t_1$ is reached at $t=t_1$. Let $r_0>0$ be defined by $r_0^2 = - 1/k(t_1)$. The range of values of the coordinate $r$ for which the metric is defined shrinks as $t$ increases. An observer along $\partial_t$ at a value $r> r_0$ will reach the singularity at the antipodal point before $t=t_1$, hence in finite proper time. 
    \item[(3)] A direct computation leads to:
    \begin{equation}
        \int^{1/\sqrt{-k(t_0)}}_{r_0} \sqrt{g_{rr}} \ \dd r = 2a(t_0)\int^{1/\sqrt{-k(t_0)}}_{r_0}\frac{\dd r}{1 +k(t_0)r^2} = \infty\,.
    \end{equation}
\end{itemize}
\end{proof}

The singular behavior can also be seen from the following fact: 
\begin{proposition}
  The singularity of the $k(t)$-conformal metric with $k(t)<0$ and $\dot{k}(t)\neq 0$ is such that the tidal forces and the expansion of the congruence of comoving observers defined by $\partial_t$ diverge as $r \to 1/ \sqrt{-k(t)}$ at constant $t$.
\end{proposition}
\begin{proof}
  The first term in the expression of the expansion  \eqref{Eq:kconf_expansion} is finite as in FLRW, but the second term blows up as $r \rightarrow 1/ \sqrt{-k(t)}$.
  
For the $k(t)$-conformal metric, the geodesic deviation \eqref{eq:geodesicdev} reads, at initial deviation,
\begin{equation}
 \left .    \frac{\nabla^2 V^\mu}{\dd t^2}  \right |_{t=t_0} = \left .  \left (\frac{\ddot{a}}{a}- \frac{r^2(2\dot{a} \dot{k} +a\ddot{k})}{a(1+kr^2)} + \frac{2r^4 \dot{k}^2}{(1+kr^2)^2}\right) \right |_{t=t_0} \delta^\mu{}_r\,,
    \label{Eq:tidal_conformal}
\end{equation}
so the tidal forces are indeed divergent as $r\to 1/\sqrt{-k}$ unless $\dot{k}=0$. 
\end{proof}

This singularity can be interpreted in the following way. If $k$ is negative and decreasing, the spatial sections become more and more hyperbolic, in the sense that the same displacement in the radial coordinate corresponds to an increasingly larger proper distance. This happens in such a way that the separation between two close freely-falling observers grows until the one at bigger $r$ disappears from the spacetime (it being pushed to infinity), so the spacetime becomes timelike geodesically incomplete. See Fig.~\ref{fig:kcon_obser_sing}, for a pictorial representation of the singularity.

\subsubsection{The \texorpdfstring{$k(t)$}{k(t)}-radial case}

\begin{figure}
    \centering
    \includegraphics[scale=1.4]{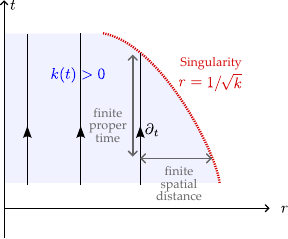}
    \caption{Spacetime diagram that shows the situation of the cosmological comoving observer of the $k(t)$-radial metric hitting the curvature singularity at $r\to 1/\sqrt{k(t)}$ whenever $k(t)$ increases in a certain interval. As indicated in the picture, the singularity is hit in a finite proper time by the observers $\partial_t$ at sufficiently large distance from the origin $r=0$. 
    }
    \label{fig:krad_obser_sing}
\end{figure}

\begin{proposition}
    \label{prop:krad_sing}
    Consider a $k(t)$-radial metric with $k(t)>0$ and $\dot{k}(t)\neq 0$.
    \begin{itemize}
        \item[(1)] The spacetime exhibits a curvature singularity at $r \to 1/ \sqrt{k(t)}$, where both the Ricci scalar and the scalar $C_{abcd}C^{abcd}$ become infinite. 

        \item[(2)] If $\dot{k}(t)>0$, the comoving observers $\partial_t$ at sufficiently high $r$  reach the singularity in finite proper time.

        \item[(3)] The spatial distance along $\partial_r$ curves between the singularity at $t=t_0$ and an event at $t=t_0$ and $r=r_0< 1/\sqrt{k(t_0)}$ is finite.
    \end{itemize}

\end{proposition}
\begin{proof} ~
\begin{itemize}
    \item[(1)] Immediate from the expression of the Ricci scalar Eq.~\eqref{Eq:RicciScal_krad}, which blows up at that point provided $\dot{k} \neq 0$ and $k>0$.
    
      We can compute the square of the Weyl tensor to obtain
\begin{equation}
    C_{abcd}C^{abcd} = \frac{n-2}{n}\left[\frac{r^2}{1-kr^2}\left(\frac{\dot{a}}{a}\dot{k}+\ddot{k}\right)+\frac{3r^4}{2(1-kr^2)^2}\dot{k}^2\right]^2\,, 
    \label{eq:CC_kradial}
\end{equation}
which also becomes infinite at the limit $r \to 1/\sqrt{k}$.

\item[(2)] The comoving time $t$ measures proper time for the geodesic observers defined by $\partial_t$. Fix any time $t=t_1$. When $\dot{k}>0$, the value of $k$ increases and the radius at which the singular behavior appears decreases with time. An observer in an orbit $r = r_0$ with $r_0 > 1/\sqrt{k(t_1)}$ will hit the singularity before $t= t_1$ and hence at finite proper time. 

\item[(3)] In this case, we find: 
\begin{equation}
\int^{r_\text{max}}_{r_0} \sqrt{g_{rr}} \ \dd r  = a(t_0) r_\text{max} \left (\frac{\pi}{2} - \arcsin(\frac{r_0}{r_\text{max}}) \right )
\end{equation}
with $r_\text{max}:= 1/\sqrt{k(t_0)}$.
\end{itemize}
\end{proof}

Notice that the divergence of the square of the Weyl tensor at $r = r_\text{max}$,  which only encapsulates information of the conformal structure of the geometry, points to the fact that the lightcone structure itself is suffering a weird phenomenon there. In fact, we will later see that the lightcones degenerate at the singularity. 

\begin{proposition}
    The singularity of the $k(t)$-radial metric with $k(t)>0$ and $\dot{k}(t)\neq 0$ is such that the tidal forces, the expansion and the shear of the congruence of comoving observers defined by $\partial_t$ diverge as $r \to 1/ \sqrt{k(t)}$.    
\end{proposition}
\begin{proof}
  The expansion and the shear scalar $\sigma^2$ of the congruence associated to $\partial_t$ become infinite as $r \rightarrow 1/\sqrt{k}$, cf.  Eqs.~\eqref{Eq:Expansion_krad}-\eqref{Eq:Shear}. In this case, the expression for the tidal forces~\eqref{eq:geodesicdev} reads
\begin{align}
  \left. \frac{\nabla^2 V^\mu}{\dd t^2} \right |_{t=t_0} = \left .  \left[\frac{\ddot{a}}{a}+ \frac{r^2(2\dot{a} \dot{k} +a\ddot{k})}{2a(1-kr^2)} + \frac{3r^4 \dot{k}^2}{4(1-kr^2)^2}\right] \right |_{t=t_0}\delta^\mu{}_r\,,
    \label{Eq:tidal_radial}
\end{align}
which clearly diverges as $r \to 1/\sqrt{k}$.
\end{proof}

Notice that even though the expression~\eqref{Eq:tidal_radial} is qualitatively similar to the one obtained in the $k(t)$-conformal case~\eqref{Eq:tidal_conformal}, in this case the situation is qualitatively different from a causal point of view (see e.g. Lem.~\ref{l_3metrics3}).

It is worth pointing that Stichel~\cite{Stichel2018} made an observational and local study of the $k(t)$ conformal and radial models which is complementary to ours. Specifically, these are his models V1 and V2 which are able to account for an accelerated expansion without a dark energy component. Even though Weak Energy Condition violations are found, it is suggested that they can be a consequence of the backreaction effects from the averaging of a more fundamental cosmological model (something consistent with previous literature~\cite{Larena2008}).

\section{Smooth extensions of the \texorpdfstring{$k(t)$}{k(t)} metrics}
\label{Sec:extendibility}

In this section we study the possibility of extending the metrics to $r=0$ and to the upper limit of $r$.  For the latter, we focus on the case $k(t)>0$ (including $r=\infty$ for $\gc$), since the constant-$t$ slices correspond to a domain of a sphere  that might a priori be extendible across its boundary. On the other hand, for $k(t)\leq 0$ the topology of these spatial slices is $\mathbb{R}^n$ and the upper limit of $r$ is already at infinity.

Decompose $I=I_{+}\cup I_{\leq 0}$, where
\begin{equation}
    \label{e_Ik0}
    I_{+}:= \{t\in I: k(t)> 0\}, \qquad I_{\leq 0} := \{t\in I: k(t)\leq  0\} \quad (=I\setminus I_+) ,
\end{equation}
and let us consider separately the three cases. Notice that in the three cases we have 
\begin{align}
    \teng = - \dd t^2 + a(t)^2 \tilde{\teng}_{t} = a(\tau)^2 \left( - \dd \tau^2 + \tilde{\teng}_{t(\tau)} \right),
    \label{Eq:a_conformalfactor}
\end{align}
where we have introduced the coordinate $\tau$ by 
\begin{equation}
    \dd \tau = \frac{\dd t}{a(t)}.
\end{equation}
In this section and the following one we study extendibility and global properties of the spacetimes. Those results that hold for the metric $\teng$ will also hold for the metric $- \dd \tau^2 + \tilde{\teng}_{t(\tau)}$ as long as the function $a(t)$ is nonvanishing.

In fact, for the rest of this section and the following one we set $a(t)=1$. The only place where this may have a potential effect is regarding the properties of the conformal boundary appearing in the $k(t)$-radial spacetime, where it can change the causal character of the boundary depending on whether the range of the coordinate $\tau$ is bounded (leading to a natural conformal boundary that is spacelike). See Rem.~\ref{rem:a_effect} below. 

\subsection{The \texorpdfstring{$k(t)$}{k(t)}-warped metric}

\begin{proposition}\label{p_3metrics1} 

The metric $\gw$ is smooth on $\uwar\, \cup  \{r=0\}$, but it cannot be smoothly extended to $r=\pi/\sqrt{k(t)}$ whenever $k(t)>0$, $\dot{k}(t)\neq 0$.

However, $\gw$ can be deformed in an arbitrarily small region $\mathcal{U}_ \pi$ around $r=\pi/\sqrt{k(t)}$, so that:  
\begin{itemize}

    \item[(a)] the deformed metric becomes smooth on $I_+\times \mathbb{S}^n$ and, so, on the manifold obtained as the union $(I_+\times \mathbb{S}^n) \cup (I_{\leq 0}\times \mathbb{R}^n)$, and  

    \item[(b)] each spherical  slice $\{t_0\}\times \mathbb{S}^n, t_0\in I_+,$ is endowed with a Riemannian metric of constant curvature $k(t_0)$. 

\end{itemize}
This deformation will be called a \emph{smoothening of $\gw$} on $\mathcal{U}_ \pi$. 
 
\end{proposition}

\begin{proof} 
    The required smoothness of $\gw$ at $r=0$ was  proven in \cite[Thm.~3.1]{Sanchez2023}. The singularity at $r=\pi/\sqrt{k(t)}$ (see Prop.~\ref{prop:kwar_sing}) makes $\gw$ $C^2$-inextensible therein.

    When $I_+$ is an interval, the smoothenings with the stated properties (a) and (b) were carefully constructed in \cite[Sec. 4]{Sanchez2023}. Otherwise, as $I_+$ is open it can be written as the union of open disjoint intervals the procedure can be applied to each of them. 
\end{proof}
\noindent Notice that $\mathcal{U}_ \pi$ can be regarded as ``small'' and irrelevant for our purposes.\footnote{\label{foot_simplification_causality}
    Notice that, in a globally hyperbolic spacetime with a temporal function $t$, any Lorentzian perturbation of the metric with compact support $K$ which preserves $t$ as a temporal function is still globally hyperbolic; moreover, if $t$ was Cauchy for the original metric then so it is for the perturbed one. This can be proved by using the properties explained at the beginning of Sec.~\ref{Sec:globalhyp}. Namely, as the cones of both metrics agree outside $K$, $J(K,K)$ is the same set for both the original and the perturbed metrics and, thus, it is compact.
    }

\begin{remark}\label{r_3metrics1}
The previous proof shows $C^2$-inextendibility, but the following computations deepen in the model by proving even $C^1$-inextendibility.
\end{remark}
Replace $r$ by the coordinate $\hat{r}=\sqrt{k(t)} \, r$ ($<\pi$) in $\uwar\cap (I_+\times \mathbb{R}^n)$, which was introduced in \eqref{eq:Rkwar_kpositive}. Then $\dd\hat{r}=\big(\partial_t\sqrt{k(t)}\big) r \, \dd t + \sqrt{k(t)}\, \dd r$ and
\begin{equation}
    \dd r^2= \frac{1}{k(t)} \dd\hat{r}^2  +  \frac{\dot{k}(t)^2}{4k(t)^3} \hat{r}^2 \dd t^2  - \frac{\dot{k}(t)}{k(t)^2} \hat{r} \dd t \dd\hat{r} , 
    \quad S_{k(t)}(r)=\frac{1}{\sqrt{k(t)}} \sin (\hat{r}) . 
\end{equation}
In these coordinates, 
\begin{align}
    \gw & =-\dd t^2+ \frac{1}{k(t)}\left(\dd\hat{r}^2+ \sin^2(\hat{r}) \tengS\right) \nonumber \\
    &\quad + \frac{\dot{k}(t)^2}{4k(t)^3} \hat{r}^2 \dd t^2  - \frac{\dot{k}(t)}{2 k(t)^2} \dd(\hat{r}^2) \dd t ,
    \qquad \qquad 0<\hat{r} <\pi \,.\label{e_sutil}
\end{align}
The first line clearly extends smoothly to $\hat{r}=\pi$ (for the same reason it does at $\hat{r}=0$) and, so, on $\mathbb{S}^n$. For the second line,  the function $\hat{r}^2$ corresponds to the square of the intrinsic distance in $\mathbb{S}^n$ to a selected point $p_0 \in \mathbb{S}^n$, which is known to be smooth at $\hat{r}=0$ (as it corresponds with the standard function $\sum_i (x^i)^2$ in normal coordinates at $p_0$), but it is not smooth in any slice $t=t_0$ at the antipodal point $-p_0$ (i.e., when  $\hat{r}= \pi$). In general, the distance to a point fails to be smooth at zero and at the cut or conjugate points along radial geodesics; for the squared distance, smoothness is recovered at zero but not at the other points. In our specific case, given a unit sphere $\mathbb{S}^n(1) \subset \mathbb{R}^{n+1}$, the distance between two points  $p,q\in \mathbb{S}^n(1)$ is $d(p,q)=\arccos(p\cdot q)$ where $ \cdot$ denotes the usual scalar product of $\mathbb{R}^{n+1}$. The derivative of $\arccos^2(x)$  is $-2\arccos(x)/\sqrt{1-x^2}$, which diverges at $-1$ (i.e., when $d^2=\pi^2$); notice that this derivative is smooth at 1 (use L'H{\^o}pital's rule) and, thus, so is $d^2$  when $d=0$.

\subsection{The \texorpdfstring{$k(t)$}{k(t)}-conformal metric}

\begin{proposition}\label{p_3metrics2}
    The  metric $\gc$ extends smoothly to $r=0$ independently of the sign of $k(t)$, and to $r=\infty$ when $k(t)>0$, that is, in $I_+\times \mathbb{S}^n$. So, the domain of the metric $\gc$ will be topologically equivalent to
    \[
        (I\times \mathbb{S}^n) \setminus (I_{\leq 0}\times \{x_\infty\}) \quad \cong \quad (I_+\times \mathbb{S}^n) \cup (I_{\leq 0}\times \mathbb{R}^n)
    \]
    where $x_\infty$ is a point in $\mathbb{S}^n$ at infinity (i.e, at each $t$-slice with  $k(t)>0$,  $x_\infty$ becomes the point  $r=\infty$ which compactifies $\mathbb{R}^n$ by sterographic projection).
\end{proposition}

\begin{proof}
The smoothness of the metric at $r  = 0$ can be immediately checked by changing to Cartesian coordinates in Eq.~\eqref{e_ucon} as it was done in~\cite[Proof of Thm. 3.1]{Sanchez2023}.  

To check that for $k(t)>0$  the metric extends smoothly to $r=\infty$, 
consider  $\bar{r}=1/ r$ as a new coordinate whenever $ r>0$. Then
\begin{align}
\label{e_g3}
    \gc & =  -\dd t^2 + 4\left(1+\frac{k(t)}{\bar{r}^2}\right)^{-2} \left(\frac{1}{\bar{r}^4}\dd\bar{r}^2+ \frac{1}{\bar{r}^2} \tengS\right) \nonumber \\ 
    & =  -\dd t^2 + \frac{4}{\left(\bar{r}^2+k(t)\right)^2} \left(\dd\bar{r}^2+ \bar{r}^2 \tengS\right),
\end{align}
which is clearly smooth at $\bar{r}=0$.
\end{proof}

These properties make the $k(t)$-conformal case simpler from the technical viewpoint.

\subsection{The \texorpdfstring{$k(t)$}{k(t)}-radial metric}

\begin{proposition}\label{p_3metrics3}
    The  metric $\gr$ can be extended smoothly to $r=0$, but it cannot be extended smoothly to $r=1/\sqrt{k(t)}$ whenever  $k(t)>0, \dot{k}(t)\neq 0$. 
\end{proposition}
\begin{proof}
  We only need to prove extendibility at $r=0$, since inextendibility at $r = 1/\sqrt{k(t)}$ has already been established in  Prop.~\ref{prop:krad_sing}(1). Consider the change $\hat{r}=\sqrt{k(t)} \; r$ in $\urad\cap (I_+\times \mathbb{R}^n)$ with $\hat{r}\in [0,1]$. Reasoning as in Rem.~\ref{r_3metrics1},
\begin{align}
    \gr & =-\dd t^2+ \frac{1}{k(t)} \left(\frac{\dd\hat{r}^2}{1-\hat{r}^2} + \hat{r}^2 \tengS\right) \nonumber \\
    & \quad +  \frac{\dot{k}(t)}{4 k(t)^2} \frac{1}{1-\hat{r}^2} \left( \frac{\dot{k}(t)}{k(t)} \hat{r}^2 \dd t^2  -  4\hat{r} \dd t \dd\hat{r} \right),
    \qquad\qquad 0<\hat{r} <1 .
\end{align}
Taking now $\bar{r}=\arcsin (\hat{r})$, so that $\dd\bar{r}=\dd\hat{r}/\sqrt{1-\hat{r}^2}$ and $\hat{r}^2/(1-\hat{r}^2)=\tan^2 (\bar{r})$ we obtain
\begin{align}
    \gr & =-\dd t^2+ \frac{1}{k(t)} \left(\dd\bar{r}^2+ \sin^2(\bar{r}) \tengS\right) \nonumber\\
        & \quad +  \frac{\dot{k}(t)}{4 k(t)^2} \left( \frac{\dot{k}(t)}{k(t)} \tan^2(\bar{r}) \dd t^2  -  4\tan(\bar{r}) \dd t \dd\bar{r}\right),\qquad 0<\bar{r} <\pi/2. \label{e_g_r}
\end{align}
To check that $\gr$ is smoothly extensible to $\bar{r}=0$,  the first line of \eqref{e_g_r} is clearly smooth  and, in the second one, $2\tan(\bar{r}) \dd\bar{r}=(\tan(\bar{r})/\bar{r}) \dd(\bar{r}^2)$, where both $\tan(\bar{r})/\bar{r}$ (whose power series only has even powers of $\bar{r}$) and $\bar{r}^2$ are smooth at $\bar{r}=0$, independently of the sign of $k(t)$.
\end{proof}
 
\begin{remark}\label{r_3metrics3} 
    The metric $\gr$ exhibits a striking difference compared with the previous warped and conformal cases. The range of the coordinate $r \in [0, 1 /\sqrt{k(t)})$ covers only half of a sphere. In the case of constant $k$, the metric can be extended across the equator $r = 1/\sqrt{k}$ to the whole sphere yielding the usual radial representation of a FLRW. However, when $\dot{k} \neq 0$, the smooth extension of the metric across the equator fails. Notice that this obstruction occurs along the entire equator of the sphere, rather than at a single point as in the previous cases.

    Therefore, the case $\dot{k}(t)\neq 0$ affects the smooth extension of the metric to an equator, rather than a point, in striking difference with the previous  warped and conformal cases.
     
    What is more, when $\dot{k}(t)\neq 0$, in the coordinates $\bar t= t, \bar{r}= \arcsin(\sqrt{k(t)}\, r)$,  $\partial_{\bar t}$  becomes spacelike close to the boundary $\bar{r}=\pi/2$ of the half sphere and, so, $\partial_{\bar t}$ cannot represent observers. Indeed, $r=1/\sqrt{k(t)}$ will behave as a spacelike hypersurface from the viewpoint of the causal and conformal boundaries (Lem.~\ref{l_3metrics3}). This has a dramatic impact for the global properties of the metric, see Prop.~\ref{p_radial} below.   
\end{remark}

\begin{remark}\label{rem:krad_topology}
    A direct consequence of the previous remark is that the spatial $t$-slices of the $k(t)$-radial spacetime with $k(t)>0$ are homeomorphic to $\mathbb{R}^{n}$. This means that a sign change in $k(t)$ does not correspond to a change in the topology of the constant-$t$ hypersurfaces for this spacetime. Notice that the same is not true for the $k(t)$-warped and the $k(t)$-conformal metrics. Hence, there will be less restrictions on the number of sign changes for $k(t)$ that are compatible with global hyperbolicity in the $k(t)$-radial case (see Rem.~\ref{rem:krad_otherGHcases}).
\end{remark}

Specifically, the expression in  parentheses at the first line of \eqref{e_g_r} is just the metric of a round sphere of radius one, and the range $0\leq \bar{r} \leq \pi/2$ gives a closed  half sphere. The terms in $\dd t^2$ and the cross terms $\dd t \dd\bar{r}$ do not affect the intrinsic geometry of the slice and, so, it becomes a half sphere of $t$-dependent radius. 
 
The second line, however, provides terms in $\tan(\bar{r})$ which diverge at $\bar{r}=\pi/2$ (if $\dot{k}\neq 0$). The coefficient of $\dd t^2$ in this line diverges positively, making $\partial_{\bar t}$ spacelike for $\bar{r}$ bigger than some $\bar{r}_\star(t)>0$ (while the cross terms preserve the Lorentzian character of the metric\footnote{
    Notice that, in these coordinates, the determinant of the part of the metric corresponding to  $(t,\bar{r})$ is always $-1/k(t)$.
    }).
More precisely, the vector $\partial_{\bar t}$ becomes lightlike at 
\begin{equation}
    \bar{r}_\star(t)= \arctan\left\vert 2\frac{k(t)^{3/2}}{\dot{k}(t)}\right\vert\,,
    \label{g_epsilon}
\end{equation}
and spacelike when $\bar{r}>  \bar{r}_\star(t)$. Let $\bar{r}_\epsilon :=\frac{\pi}{2}-\epsilon$, then each hypersurface $H_{\epsilon}:=\{\bar{r} =\bar{r}_\epsilon \}$
inherits the metric
\begin{equation}
    \gr_\epsilon= \left(-1+ \frac{\dot{k}(t)^2}{4 k(t)^3} \tan^2(\bar{r}_{\epsilon})\right) \dd t^2+ \frac{1}{k(t)} \sin^2(\bar{r}_\epsilon) \tengS \,,    
    \label{Eq:grepsilon}
\end{equation}
which is positive definite whenever $\bar{r}_\epsilon > \bar{r}_\star(t)$.  Notice that the whole region $\bar{r}<\pi/2$ can be regarded as the limit of the regions $\bar{r}_\epsilon <\pi/2$ when $\epsilon\searrow 0$ (see Fig.~\ref{fig:krad_Hepsilon}) and that the function $\tan(\bar{r}_{\epsilon})$ becomes unbounded when $\bar{r}_{\epsilon} \rightarrow \pi/2$. Summing up, we have proved the following result to be used later\footnote{
    Recall two Riemannian metrics $\teng_1$ and $\teng_2$ satisfy $\teng_1 > \teng_2$ if the norm of any non-zero vector with respect to $\teng_1$ is strictly larger than its norm with respect to $\teng_2$.}
\begin{figure}
    \centering
    \includegraphics[scale=1.4]{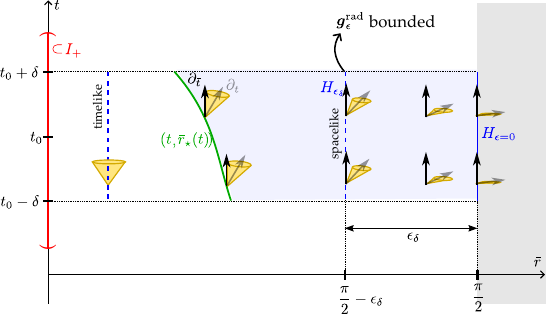}
    \caption{Causal structure of the $k(t)$-radial spacetime with $k(t)>0$.}
    \label{fig:krad_Hepsilon}
\end{figure}

\begin{lemma}\label{l_3metrics3}
    For the $k(t)$-radial spacetime, the boundary $\bar{r}=\pi/2$ for the region $I_+\times \{\bar{r}<\pi/2\}$ (included in $\urad \subset I\times \mathbb{R}^n$) is spacelike in the following sense: for each $t_0\in I_+$ and $\delta>0$ such that $[t_0-\delta,t_0+\delta]\subset I_+$, there exists $\epsilon_\delta>0$ such that the hypersurface $H_\epsilon (=\{\bar{r} =\bar{r}_\epsilon \})$ in $[t_0-\delta,t_0+\delta]\times \{\bar{r}<\pi/2\}$ endowed with the induced metric $\gr_\epsilon$ above is spacelike and satisfies
    \begin{equation}
        \gr_\epsilon > \dd t^2 + \frac{1}{k_\delta}\sin^2(\bar{r}_{\epsilon_\delta}) \tengS, \qquad \forall \epsilon\in (0,\epsilon_\delta) \,, \label{ineq}
    \end{equation}
    where $k_\delta$ is the maximum of $k([t_0-\delta,t_0+\delta])$.
\end{lemma}

\begin{remark}\label{r_conformalnoboundary}
    Notice that $\delta$ and, thus, $k_\delta$ are fixed here and, using \eqref{g_epsilon},  the inequality \eqref{ineq} for  $\gr_\epsilon$ can be sharpened. As a matter of fact, the coefficient which multiplies $\dd t^2$ in Eq.~\eqref{Eq:grepsilon} for each $t$ diverges when $\epsilon\searrow 0$, but the factor of $\tengS$ converges to $1/\sqrt{k(t)}$, which indicates  that the spacetime cones collapse in the direction of $\partial_t$ and remain transverse to $H_\epsilon$ in a uniform way. This is consistent with the divergence of  $C_{abcd}C^{abcd}$, which implies that  no  conformal factor transforming $H_{\epsilon=0}$ in a (non-degenerate) conformal boundary can exist. However, in a natural way the points of $H_{\epsilon=0}$ belong to the causal boundary and are non-timelike (see footnote~\ref{foojpaj} below).
\end{remark}

\begin{remark}\label{rem:a_effect}
    Notice that adding a conformal factor, in particular the $a(t)$ factor (see Eq.~\eqref{Eq:a_conformalfactor}) can change the causal character of the boundary, since it can become lightlike. However, it still behaves as a conformal boundary in the sense described above.
\end{remark}

From now on, we work with the maximal extensions of the metrics, as established in this section.

\section{Global causal properties of the \texorpdfstring{$k(t)$}{k(t)} spacetimes}
\label{Sec:globalhyp} 

As usual, $J^+(p), J^-(p)$ will denote, respectively, the causal future and past of $p$, and $J(p,q):=J^+(p)\cap  J^-(q)$, which extends to subsets $K,K'$ as $J(K,K'):= \cup_{p\in K,q\in K'} J(p,q)$. In addition, we use $p\leq q$ to indicate that the event $p$ is in the causal past of $q$ (i.e., $q\in J^+(p)$). A spacetime is called \emph{globally hyperbolic} when all the $J(p,q)$ are compact and it has no closed causal curves.

The $k(t)$ spacetimes always admit the temporal function $t$, so they are stably causal. The following well-known facts for such spacetimes will be used with no further mention (details can be found, for example, in \cite{MinSan}):
\begin{itemize}
    \item[(a)] they are strongly causal (i.e., any point has a neighborhood whose intrinsic causality matches the restriction of the causality of the spacetime),

    \item[(b)] in order to check global hyperbolicity, it is enough to show that  each $J(p,q)$ is included in a compact subset of $t^{-1}([t(p),t(q)])$ (then, it turns that $J(K, K')$ is compact whenever $K$ and $K'$ are compact),\footnote{
        This comes from an original result by Beem and Ehrlich,  see  \cite[Lem.~4.29]{Beem}, which is applicable here using item (a).
        }

    \item[(c)] the $k(t)$ spacetimes (or regions included in them)\footnote{
        Notice that the notions of global hyperbolicity and Cauchy hypersurfaces are extended to closed subsets of our spacetimes directly.
    }
    are globally hyperbolic if their cones are narrower than the cones of a globally hyperbolic spacetime; in this case, the Cauchy hypersurfaces of the latter spacetime are also Cauchy for the original one.
\end{itemize}

The following sufficient condition  will be used several times to prove that global hyperbolicity can hold under spatial curvature change.
\begin{lemma}\label{l_marc}
Let $t$ be a  temporal function on a spacetime $\M$ and $t_0\in \mathbb{R}$.
\begin{itemize}
    \item[(1)] If both subsets, $t^{-1}((-\infty,t_0))$ and $t^{-1}([t_0, \infty))$ are globally hyperbolic and the $t$-slices $t=t_* \geq t_0$ are Cauchy for  the second one, then the whole spacetime $\M$ is globally hyperbolic and the Cauchy hypersurfaces for the first subset are also Cauchy hypersurfaces for  $\M$. Moreover, when the slices $t=t_* <t_0$ are compact then they are Cauchy for the first subset (and thus for $\M$).

    \item[(2)] If $t_-\leq t_0 \leq t_+$ and the slice $t=t_0$ is a Cauchy hypersurface for  $t^{-1}((-\infty,t_+])$ and $t^{-1}([t_-,\infty))$, then this slice is  a Cauchy hypersurface for the whole spacetime $\M$ (which is therefore globally hyperbolic).
\end{itemize}
\end{lemma}
\begin{proof} ~
\begin{itemize}
    \item[(1)] It is  straightforward  that  $J(p,q)$  is included in a compact set when both $p$ and $q$ lie  either in $t\geq t_0$ or in $t<t_0$. In case $t(p)<t_0\leq t(q)$ then $K:=J^-(q)\cap \{t\geq t_0\}$ is compact because the slice $\{t= t_0\}$ is Cauchy in this region; thus,  $K_0:=K\cap \{t= t_0\}$ is compact too.  Moreover,  $K_-:=J^-(K_0)\cap \{t \geq t_0 -\epsilon \}$ is also compact for small $\epsilon >0$ (the proof is reduced to a local reasoning in convex neighborhoods by using the compactness of $K_0$ and strong causality) and, then, so is $K_{-\epsilon}:=K_-\cap \{t= t_0-\epsilon \}$.  By the global hyperbolicity of the region $t<t_0$, $J(p,K_{-\epsilon})$ is compact too and, then,  $J(p,q)$ is included in the compact set $J(p,K_{-\epsilon})\cup K_- \cup J(K_0,q)$, as required. 
    
    For the assertion on Cauchy hypersurfaces, notice that, as $\{t= t_0\}$ is Cauchy for the second region, then any causal curve starting there will arrive at the first region and, thus, it will intersect all its Cauchy hypersurfaces. In case that the slices with $t<t_0$ are compact they must be Cauchy because so is any compact acausal spacelike hypersurface in a globally hyperbolic spacetime (as its projection on a Cauchy hypersurface must be a one-leaf covering).

    \item[(2)] Any inextensible causal curve on $\M$ must contain at least a point either in $t^{-1}((-\infty,t_+])$ or in $t^{-1}([t_-,\infty))$, since they cover $\M$. Due to the fact that $t=t_0$ is Cauchy for both regions, the curve intersects it at exactly one point. Therefore, $t=t_0$ is Cauchy for the entire $\M$.
\end{itemize}
\end{proof}

\subsection{The \texorpdfstring{$k(t)$}{k(t)}-warped cosmological metric}

The properties of the $k(t)$-warped  metric are studied in detail in \cite{Sanchez2023}, where the smoothenings in Prop.~\ref{p_3metrics1} were constructed by using an additional function $\varphi(t,\theta)$ which introduces cross terms between the time and space parts. Here, these techniques will be used and eventually can be simplified.

\begin{proposition} \label{t_warped} For any smoothening of $\gw$ as in Prop.~\ref{p_3metrics1} 
\begin{itemize}
    \item[(1)] If $k(t)\leq 0$, then the $k(t)$-warped spacetime is  globally hyperbolic  and all the (non-compact)  slices $t=$constant are Cauchy.

    \item[(2)] Let $t_0\in I$ and assume that  $k(t)> 0$ for $t<t_0$ and $k(t)\leq 0$ for $t\geq t_0$ (so that there are a topological and a sign curvature change of the slices at $t=t_0$).  Then, the spacetime is globally hyperbolic with compact Cauchy hypersurfaces, being the slices  with $t<t_0$  Cauchy (but not those with $t\geq t_0$). 
\end{itemize}
\end{proposition}
\begin{proof} ~
\begin{itemize}
    \item[(1)] The manifold is   $I\times\mathbb{R}^n$, which can be regarded as  a slab of  $\mathbb{L}^{n+1}$ with narrower cones (as $S_k(r)\geq r$ for $k\leq 0$).

    \item[(2)] The property just stated shows that $t^{-1}([t_0,\infty))$ is globally hyperbolic with Cauchy $t$-slices. For the region $t^{-1}((-\infty,t_0))$, global hyperbolicity holds because each $J(p,q)$ is included in the compact subset $t^{-1}([t(p),t(q)])$ (notice that this set has the topology of $[t(p),t(q)]\times \mathbb{S}^n$). So, the result follows from Lem.~\ref{l_marc} (1).
    
    The slices $t\geq t_0$ are not compact and, thus, non-homeomorphic to the Cauchy ones.
    
\end{itemize}
\end{proof}

\begin{figure}
    \centering
    \includegraphics[scale=1.4]{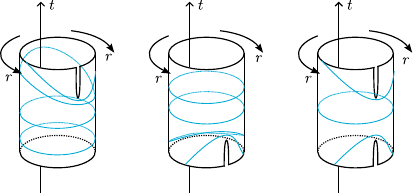}
    \caption{
    Cauchy slices for different $k(t)$-warped spacetimes with topology changes.
    }
    \label{fig:kwar_slices}
\end{figure}

\begin{remark} About the case (2) in Prop.~\ref{t_warped}: 
\begin{itemize}
    \item[(1)] Explicit slicings of the spacetime by (necessarily compact) Cauchy hypersurfaces  can be found. See Fig.~\ref{fig:kwar_slices} and also \cite[Fig. 2]{Sanchez2023}.

    \item[(2)] The curve $\big(t,r=\pi/\sqrt{k(t)}, \theta^A_0\big)$ for $t<t_0$ (and fixed angular coordinates $\theta^A_0$), corresponds to the \emph{singular comoving observer} $\gS_{\pi}$ in \cite{Sanchez2023}. This curve is timelike  therein as the metric is smoothened there. For $\gw$,  $\gamma_\pi$ is just the (smooth) curve where the expansion of the  slices $t=$ constant diverges Physically, $\gamma_\pi$  might be understood as the  blow up  of the expansion which generates the curvature  change. 

\end{itemize}
\end{remark}    

\begin{lemma}\label{lem:topo_restic_kwar}
    The $k(t)$-warped spacetime is \emph{not} globally hyperbolic  when $I_+$ is not an interval, that is, whenever there exists $t_-<t_0<t_+$ such that 
    \[
        k(t_-)>0, \qquad  k(t_0)\leq 0, \qquad k(t_+)>0.
    \]
\end{lemma}
 \begin{proof}
    If the spacetime were globally hyperbolic, then the Cauchy hypersurfaces would be spheres (as so is $\{t=t_-\}$). Then, the spacetime must be topologically $\mathbb{R}\times \mathbb{S}^n$ and a contradiction will appear as follows. 

    In the case $n=1$, the spacetime would be a topological cylinder $\mathbb{R}\times \mathbb{S}^1$,   which is homeomorphic to a sphere minus two points and has  fundamental group  equal  to $\mathbb{Z}$. However, the existence of $t_-, t_0, t_+$ implies that the spacetime is homotopic to a sphere minus at least a third point (indeed, minus an additional  point for each maximal interval where $k\leq 0$) and, thus, its fundamental group must contain at least $\mathbb{Z} \oplus \mathbb{Z}$. 
    
    If $n\geq 2$ a similar reasoning follows by using the $n$-cohomology ring  (which is also homotopy invariant) for the product $\mathbb{R}\times \mathbb{S}^n$ with some points removed. 
  \end{proof}
  
We can now characterize when the $k(t)$-warped models are globally hyperbolic (see Fig.~\ref{fig:kwar_globhyp}),

\begin{thm} \label{thm:kwar_globhyp} 
    A $k(t)$-warped cosmological spacetime is globally hyperbolic if and only if $I_+$ (see \eqref{e_Ik0}) is an interval. That is, either $k(t)\leq  0$ on  the whole $I=(a,b)$ or one can find  $a\leq t_ -< t_+ \leq b$ such that $I_+=(t_- , t_+)$ (in particular, there are at most  two  topological changes in the $t$-slicing). 

    In this case, the Cauchy hypersurfaces are compact if and only if $I_+\neq \emptyset$, that is, when $k(t_0)>0$ at some $t_0\in I$. 
\end{thm}
\begin{proof}
    The necessary condition follows from Lem.~\ref{lem:topo_restic_kwar}. The sufficient one is straightforward  from  Prop.~\ref{t_warped} when there is at most one topological change. Otherwise, there are two topological changes, and we can choose $t_0\in I$ and $ \epsilon >0$  such that $I_\epsilon:=[t_0-\epsilon, t_0+\epsilon] \subset (t_- , t_+)$. From  Prop.~\ref{t_warped},  both regions $t\geq  t_0-\epsilon$ and $t\leq  t_0+\epsilon$ are globally hyperbolic, and  the slices $t=t^*$ with  $t^*\in I_\epsilon$ are Cauchy for both,  so Lem.~\ref{l_marc}(2) applies. 

\end{proof}

\begin{figure}
    \centering
    \includegraphics[scale=1.4]{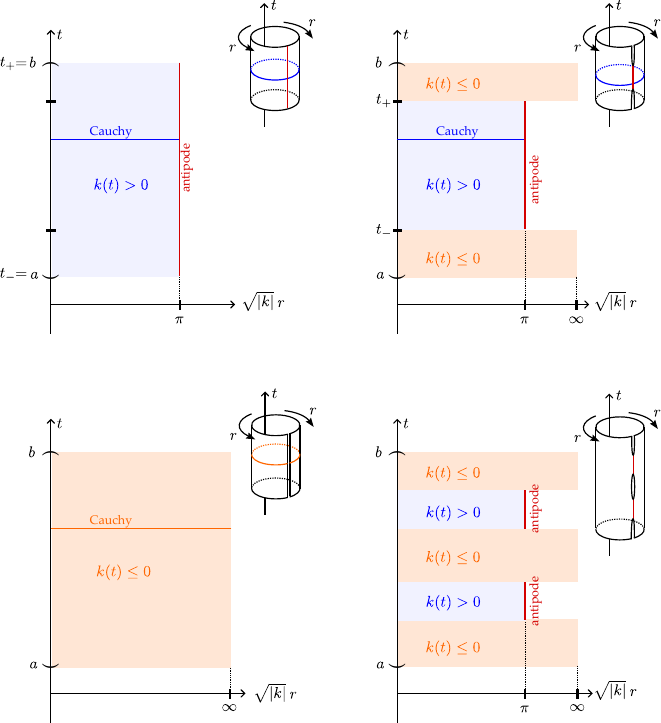}
    \caption{Top and bottom-left diagrams represent those cases for which $k(t)$-warped spacetime is globally hyperbolic. A Cauchy surface has been highlighted in each case (corresponding to a $t=t_0$ surface). Bottom-right diagram is an example of a non-globally hyperbolic case ($I_+$ is not an interval).}
    \label{fig:kwar_globhyp}
\end{figure}

\subsection{The \texorpdfstring{$k(t)$}{k(t)}-conformal cosmological spacetime}

Next, our aim is to determine exactly when the $k(t)$-conformal spacetime is globally hyperbolic.
 
We begin by noting that Lem.~\ref{lem:topo_restic_kwar} remains valid in this case as well:
\begin{lemma}\label{lem:topo_restic_kcon}
    The $k(t)$-conformal spacetime is \emph{not} globally hyperbolic  when $I_+$ is not an interval, that is, whenever there exists $t_-<t_0<t_+$ such that 
    \[
        k(t_-)>0, \qquad  k(t_0)\leq 0, \qquad k(t_+)>0.
    \]
\end{lemma}
\begin{proof}
    The reasoning used to prove Lem.~\ref{lem:topo_restic_kwar} is also valid here.
\end{proof}

The next two propositions yield representative particular cases, and the technical Lem.~\ref{l_conf} (in addition to Lem.~\ref{lem:topo_restic_kcon}) shows  the obstacle for global hyperbolicity. Then, the full result is attained at Thm.~\ref{thm:kcon_globhyp}. 

\begin{proposition}\label{t_conformal0}
  Assume that either
  \begin{equation}
      \begin{cases}
          k(t)> 0  &\ \mathrm{for} \ t<0 \\
          k(t)= 0  &\ \mathrm{for} \ t\geq 0
      \end{cases}
      \qquad \mathrm{or} \qquad
      \begin{cases}
          k(t)= 0  &\ \mathrm{for}\ t\leq 0\\
          k(t)> 0  &\ \mathrm{for}\ t>0 
      \end{cases}      
  \end{equation}
  (so that there are a topological and a sign curvature change of the slices at $t=0$).  Then, the $k(t)$-conformal spacetime is globally hyperbolic with compact Cauchy hypersurfaces, and  the slices $\{t=t_0\}$ are Cauchy if and only if $t_0<0$.  
\end{proposition}
\begin{proof}
As $k(t)\geq 0$ everywhere, there is no restriction for the open subset $\ucon$ in \eqref{e_ucon}  (see Prop.~\ref{p_3metrics2})  and the proof follows by reasoning as in Prop.~\ref{t_warped}.
\end{proof}

\begin{proposition}\label{t_conformal}
    Assume that  $k(t)\leq 0$, and either  $\dot{k}(t)\leq 0$ everywhere or $\dot{k}(t)\geq 0$ everywhere. Then the $k(t)$-conformal spacetime is  globally hyperbolic  and all the (non-compact)  slices $t=$constant are Cauchy.
\end{proposition}
\begin{proof}
    Notice first that the metric $\gc$ satisfies the stated properties if they also hold for the (\emph{double base}) 2-dimensional spacetime\footnote{
        It is not difficult to check that the converse also holds. Moreover, the result can be naturally extended to spacetimes which are warped products with a compact (or at least complete) Riemannian fiber $\mathcal{F}$ outside a compact region (recall also footnote \ref{foot_simplification_causality}), by taking into account that the projection on the base of warped product maps causal cones  onto cones. The technique  may have interest in its own right because it is easy to find examples of this situation when $\mathcal{F}$ is compact.
    }
    \begin{align}
    \teng_{\ucon_d}= -\dd t^2 + \frac{4\dd X^2}{\left(1+k(t) X^2\right)^2}, \quad 
    \ucon_d= \left\{(t,X)\in I\times \mathbb{R}: |X| < \varrho(t):= \frac{1}{\sqrt{-k(t)}}\right\}.
    \label{Eq:2Dconformal}
    \end{align}
    Indeed, if the original spacetime $\M$ were not globally hyperbolic there would exist $p=(t_p,x_p) \leq q=(t_q,x_q)$ and $z_j=(t_j,y_j)\in J(p,q) \subset \M$ such that  $r_j:=r(y_j)$ could be arbitrarily large. Then, in  $\ucon_d$: $(t_p,r(x_p)) \leq (t_q,r(x_q))$ and $(t_j,r(y_j))\in J((t_p,r(x_p)),(t_q,r(x_q)))$. Thus, the spacetime~\eqref{Eq:2Dconformal} would not be not globally hyperbolic. Analogously, if the $t$-slices were not Cauchy for the original spacetime they neither would be for the 2-dimensional one.

    As $k(t)\leq 0$, the cones of \eqref{Eq:2Dconformal} are always at least as narrow as in $\mathbb{L}^2$ (including the limit case  $r=0$) and, thus, the possible non-global hyperbolicity may appear only when $\varrho(t)<\infty$. 
   
    Let $p\leq q$ and assume $\dot{k}\leq 0$ (thus $\dot{\varrho}(t)\leq 0$); the case $\dot{k}\geq 0$ is similar interchanging the roles of $p$ and $q$.  Let us prove that $J(p,q)$  is compact.

    Consider the two past directed lightlike geodesics departing from $q$ and their parameterizations by $t$, namely, $\eta(t):=(t,r(t))$, $\hat \eta(t):=(t,\hat{r}(t))$. Choose $\eta$ so that $r(t)$ increases strictly when $t$ decreases.\footnote{
        When $r(q)=0$ both geodesics obey the same differential equation and have $\dd r/\dd t<0$ everywhere; otherwise, we choose $\dd r/\dd t<0$ at $q$ (and, then, everywhere) and $\dd \hat{r}/\dd t>0$ at $q$ (and, then, this inequality holds until $\hat{r}$ vanishes).
    }
    If these geodesics are defined in $[t(p),t(q)]$ then $K:=J^-(q)\cap \{t\geq  t(p)  \}$ is compact, the two future directed lightlike geodesics starting at $p$ cannot be imprisoned in $K$ and they will intersect $ \eta,  \hat \eta$ delimiting the compact set $J(p,q)$ (see Fig.~\ref{fig:kcon_proof2dim} left).

    Otherwise, the (outermost) geodesic $\eta$ must satisfy  $r(t_0)=\varrho(t_0)\leq \infty$  at some first point (starting from $t(q)$) $t_0\in [t(p),t(q))$. Notice: (i) the cones of $\teng_{\ucon}$ in $t\geq t_0$ are narrower than the cones of the constant curvature spacetime with $k\equiv k(t_0)$ therein because $\dot{k}(t)\leq 0$ (see Fig.~\ref{fig:kcon_proof2dim} right), and (ii) the curve  $\eta$ is then inextensible, past directed and included in the closed subset  $t\geq t_0$ of the constant curvature space $k\equiv k(t_0)$. This is a contradiction, because the latter admits $\{t=t_0\}$ as a Cauchy hypersurface and $\eta$ is past-inextensible and included in its future but does not cross it.
\end{proof}

\begin{figure}
    \centering
    \includegraphics[scale=1.4]{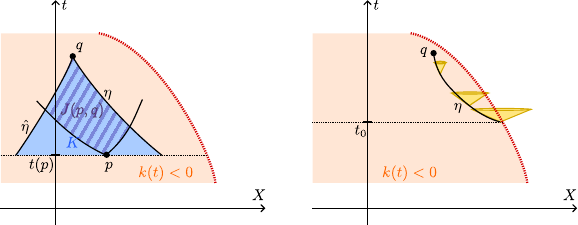}
    \caption{(Left) diagram showing a compact $J(p,q)$ for the situation described in the text. (Right) Lightcones shrinking close to the singularity as $t$ grows.}
    \label{fig:kcon_proof2dim}
\end{figure}

In Prop.~\ref{t_conformal}, whenever $k(t)\leq 0$  a sufficient condition   for global hyperbolicity is the monotonicity of $k(t)$. The following lemma will be the key  necessary condition.  
\begin{lemma}\label{l_conf}
    Assume that $k(t)$  (and thus $\varrho(t)= 1/\sqrt{-k(t)}$) admits a \emph{ strictly  minimizing critical point}, that is a point $t_0\in I$, and $\epsilon_- , \epsilon_+>0$ such that $t_0$ is an absolute  minimum of $k|_{[t_0-\epsilon_-,\ t_0+\epsilon_+]}$ and 
    \begin{equation}
        k(t_0)<0, \qquad k(t_0)<k(t_0\pm \epsilon_\pm )\,.
    \end{equation}
    Then, the $k(t)$-conformal spacetime is not globally hyperbolic.
\end{lemma}
\begin{proof}
    For sufficiently small $\{\delta_j \} \searrow 0$ the pieces of integral  curves $\gamma_j$ of $\partial_t$ in spherical coordinates 
    \begin{equation}
        \gamma_j(t):=\left(t,r=\frac{1}{\sqrt{-k(t_0)}}-\delta_j , \theta^A_0\right)\qquad \forall t\in [t_0-\epsilon_-, t_0+\epsilon_+ ],
    \end{equation}
    (for fixed coordinates in the sphere $\theta^A_0$) are well defined, as the coordinate $r$ does not reach $1/\sqrt{-k}$. The hypotheses on $k(t_0)$ imply that  
    \begin{equation}
        p_{\pm}:=\lim_j \gamma_j(t_0 \pm \epsilon_\pm)= \left(t_0\pm \epsilon_\pm , r=\frac{1}{\sqrt{-k(t_0)}}, \theta^A_0\right)
    \end{equation}
    are also well defined. Thus, if the spacetime were globally hyperbolic then the sequence of causal curves $\{\gamma_j\}_j$ would have a causal limit curve connecting $p_{-}$ and $p_+$. However, this is not possible as no subsequence of $\{\gamma_j (t_0)\}_j$ is contained in a compact set.
\end{proof}

\begin{remark}
    Notice that, for $k(t)<0$, $\varrho(t)  = 1/\sqrt{|k|} $ increases/decreases in the same way as $k(t)$. Therefore minima of $k(t)(<0)$ correspond to minima of $\varrho(t)$ (i.e., of the value of the radial coordinate corresponding to the singular spatial infinity).
\end{remark}

The previous lemma and the topological restriction in Lem.~\ref{lem:topo_restic_kcon} will yield the only restrictions to global hyperbolicity. 

\begin{figure}
    \centering
    \includegraphics[scale=1.4]{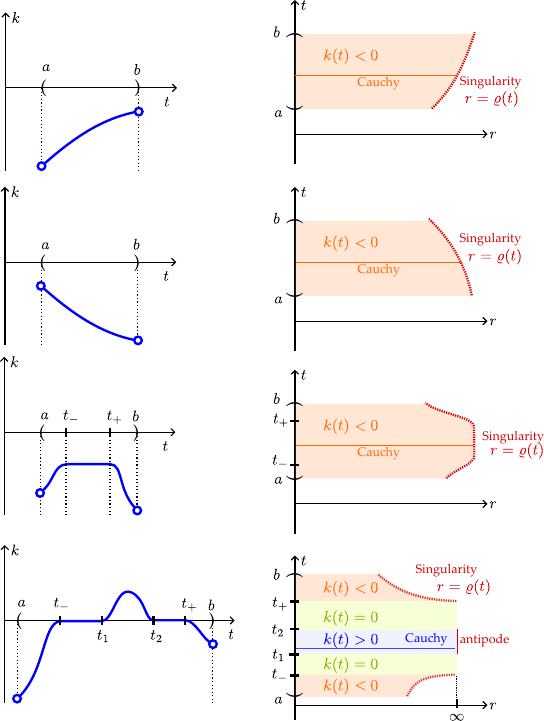}
    \caption{Schematic visualizations of different cases in Thm.~\ref{thm:kcon_globhyp}, on the left the profile of $k(t)$ and on the right its effect on the spacetime. The first two correspond to the case (1)(i), the third one to (1)(ii) and the bottom diagrams to the case (2). They cover relevant possibilities but not all possible qualitative behavior; for instance, in the third diagram $k(t)$ can be zero in $(t_-, t_+)$ or, in the last one, the function can have more critical points in $(t_1,t_2)$ as long as $k(t)$ remains positive.
    }
    \label{fig:kcon_globhyp}
\end{figure}

\begin{thm}\label{thm:kcon_globhyp}
    A $k(t)$-conformal spacetime is globally hyperbolic if and only if one of the following two exclusive cases occur (see Fig.~\ref{fig:kcon_globhyp}):
    \begin{itemize} 
        \item[(1)]  $k(t)\leq 0$ for all $t\in I$ and it does not admit strictly minimizing critical points, that is, either  (i)  the sign of $\dot{k}$ does not change between positive and negative   or  (ii)  there are $t_- \leq t_+$ such that
        \begin{equation}
            \dot{k}(t) \left\{
            \begin{array}{ll}
                \geq 0 & \hbox{if} \; t\leq t_- \\
                = 0 & \hbox{if} \;  t_-\leq t\leq t_+ \\
                \leq 0 & \hbox{if} \;  t_+\leq t
            \end{array}
            \right.  
        \end{equation}

        In any of these cases, the slices $t=$ constant are Cauchy. 

    \item[(2)] There exists $t_\star \in I$ such that $k(t_\star)>0$ and the interval $I=(a,b)$ admits $a\leq t_- \leq t_1   < t_\star <  t_2 \leq t_+ \leq b$ such that:\footnote{
        Notice that if $k$ vanishes at two points $t_-<t_+$ then it cannot be negative in $(t_-,t_+)$ as, otherwise, a strictly minimizing critical point would appear therein. It can be positive in a subinterval $(t_1,t_2)$ but it cannot be positive in a non-connected subset of $(t_-,t_+)$ because of Lem.~\ref{lem:topo_restic_kcon}.
        }
    \begin{itemize}
        \item[(a)]  $k\geq 0$ on $[t_-,t_+]\cap I$ with strict inequality exactly on $(t_1,t_2)$; 
        \item[(b)] elsewhere, $k<0$ with $\dot{k}\geq 0$ on $(a,t_-)$, and $\dot{k}\leq 0$ on $(t_+,b)$.
    \end{itemize}
    In particular, this case holds when $k(t)> 0$ for all $t\in I$ (with no restriction on $\dot{k}$).

    In any of these cases, the slice $\{t=t_0\}$ is Cauchy if and only if $t_1<t_0<t_2$. 
    \end{itemize}
\end{thm}
\begin{proof}~
\begin{itemize}
    \item[(1)] If $k(t)$ admits a strictly  minimizing critical point then it is not globally hyperbolic by Lem.~\ref{l_conf}. In case  (i), Prop.~\ref{t_conformal} applies. Otherwise case (ii) holds and this proposition (which is applicable also to closed intervals of $I$) applies to the regions $t\leq t_+$ and $t_-\leq t$; so, Lem.~\ref{l_marc}(2) does the job.  

    \item[(2)] The regions $t_-\leq t\leq t_2$ and $t_1\leq t \leq t_+$ are globally hyperbolic subsets by using Prop.~\ref{t_conformal0} and they intersect in the common Cauchy hypersurface $\{t=t_\star \}$.  Thus, Lem.~\ref{l_marc}(2) implies that the region $t_-<t<t_+$ is globally hyperbolic. 
 
    The regions $t\leq t_1$ and $t_2 \leq t$ are globally hyperbolic  subsets with Cauchy $t$ slices by using Prop.~\ref{t_conformal}. Then, Lem.~\ref{l_marc}(1) ensures that they match with $t_-<t<t_+$ in a single globally hyperbolic spacetime and the Cauchy slices are only the compact ones. 
\end{itemize}
The other possibilities are forbidden by Lem.~\ref{l_conf}, taking into account that Lem.~\ref{lem:topo_restic_kcon} only permits two topological transitions.
\end{proof}

\subsection{The \texorpdfstring{$k(t)$}{k(t)}-radial cosmological spacetime}

Let us start arguing as in the previous cases.

\begin{proposition}\label{t_radial0} 
If $k(t)\leq 0$, then the $k(t)$-radial spacetime is  globally hyperbolic and its $t$-slices are Cauchy.
\end{proposition}
\begin{proof}
Topologically, the manifold is $I\times\mathbb{R}^n$ and it is enough to check that, for each compact interval $[c,d]\times \mathbb{R}^n$, the cones are narrower than those of a specific open FLRW spacetime. Let $k^+$ be the maximum of $\{|k(t)|: t\in [c,d]\}$. Clearly, the cones are narrower than those of FLRW of curvature $k(t)\equiv -k^+$.
\end{proof}

\begin{remark}\label{r_radial}
To extend this result to the case $k(t)>0$, recall that, in any open interval where $k(t)$ is a positive constant, the metric could be extended smoothly not only to the half sphere $\bar{r} \leq \pi/2$ but also to a whole sphere (Rem.~\ref{r_3metrics3}). Thus, if we insisted in considering only the region with $\bar{r}< \pi/2$ then the hypersurface $\bar{r} = \pi/2$ would be timelike  \emph{and the spacetime would not be globally hyperbolic}\footnote{\label{fnote:kradboundary}
    However, it might be regarded as globally hyperbolic with timelike boundary (according to \cite{AFS}), eventually under a further cutoff. Indeed, if one can choose $\epsilon$ so that $r_\epsilon < \bar{r}_\star(t)$ for all $t\in I$ (see  \eqref{g_epsilon}), then $H_\epsilon$ would serve as such a timelike boundary.
    } 
while, otherwise, there would be topological transitions even in the region $k>0$, in each interval where the whole sphere is included. 
\end{remark}

From this remark, we will consider that $\dot{k}$ does not vanish
when $k>0$.
\begin{proposition}\label{p_radial}
    When $k>0$, the metric $\gr$ is globally hyperbolic if either $\dot{k}(t)>0$ or $\dot{k}(t)<0$ holds for all $t\in I$. However, no  $t$-slice is a Cauchy hypersurface. 
\end{proposition}
\begin{proof}
    Let us check that $J(p,q)$ lies in a compact subset for any points $p,q$.  Trivially, they lie in the compact subset $[t(p), t(q)]\times \{0\leq \bar{r} \leq \pi/2\}$ for the coordinates introduced in \eqref{e_g_r}, however, the  points $\bar{r} = \pi/2$ do not belong to the spacetime. Anyway, it is enough to check that no sequence $\{q_j=(t_j,x_j)\}\subset J(p,q)$ converges in these coordinates to a point $q_\infty=(t_\infty, x_\infty)$ with $\bar{r}(x_\infty)=\pi/2$. Reasoning by contradiction,  apply Lem.~\ref{l_3metrics3} with $t_0=t_\infty$ and $\delta= t(q)-t(p)$ to find a $\epsilon >0$  such that $H_\epsilon$ is spacelike for small $\epsilon$ and, so, both the future directed causal curve departing from $p$ to $q_j$ and the past directed one from $q$ to $q_j$ must cross $H_\epsilon$ transversely with increasing coordinate $\bar{r}$ (for big $j$ and small $\epsilon$). But this is impossible because only one of the two type of cones (the future or the past directed ones) on $H_\epsilon$ can lie on  the side where $\bar{r}$ increases, as $H_\epsilon$ is spacelike.\footnote{
        From a technical viewpoint, this result is related to the claim that a strongly causal spacetime is globally hyperbolic when its conformal boundary has no timelike points (see \cite[Cor. 3.4]{FHS_ATMP} for a precise formalization). As pointed out  in Rem.~\ref{r_conformalnoboundary}, the natural boundary for the radial spacetime cannot be regarded as a (standard) conformal one. However,	Lem.~\ref{l_3metrics3} provides a hypothesis which play the role of ``having no timelike points'' here. The general setting of the causal boundary (which is intrinsic and defined with independence of conformal embeddings) properly describes this situation \cite[Sec. 4]{FHS_ATMP}. Indeed, the absence of timelike points correspond to the causal boundary property that no TIP (Terminal Indecomposable Past) is S-related to a TIF (Terminal Indecomposable Future)~\cite[App. A]{AFS} and then, to the absence of naked singularities.
        \label{foojpaj}
    }

    In order to prove that no slice $\{t=t_0\}$ is Cauchy, choose any $\delta>0$ and take  $\epsilon>0$ from Lem.~\ref{l_3metrics3} as above. Assume that the future cones point out in the $\bar{r}$-increasing side along $H_\epsilon$ (an analog reasoning works if the cones are past-directed).  These cones are more tilted as $\bar{r}$ increases and, thus, $\{t=t_0\}$ will not be intersected by any future-directed causal curve starting at a point in $\{t = t_0 - \delta/2 \}$ with $\bar{r}$ close to $\pi/2$.
\end{proof}

From the proof, the projection of any Cauchy hypersurface on the $(t,\bar{r})$ part can be seen as a graph $\bar{r} \mapsto t(\bar{r})$ which diverges when $\bar{r}\rightarrow \pi/2$  (this can be seen more clearly by taking a ``double base'' 2-dimensional spacetime as in the proof of Prop.~\ref{t_conformal}).

\begin{figure}
    \centering
    \includegraphics[scale=1.4]{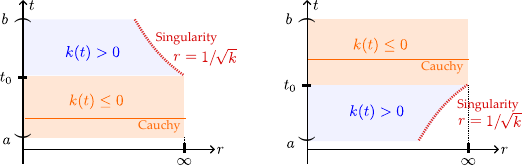}
    \caption{Schematic visualizations of the globally hyperbolic cases (i) (left) and (ii) (right) for the $k(t)$-radial metric presented in Thm.~\ref{thm:krad_globhyp} (the limit possibilities $k(t)>0$ everywhere or $k(t)\leq 0$ everywhere are permitted). In both cases only the $t$-slices in the region with $k(t)\leq 0$ are Cauchy. Observe that, contrary to the singular $k(t)$-conformal case, here whenever $k(t)>0$ and $\dot{k}(t)>0$ the singularity approaches the origin. 
    }
    \label{fig:krad_globhyp}
\end{figure}

In conclusion:
\begin{thm}\label{thm:krad_globhyp}
    A $k(t)$-radial spacetime is globally hyperbolic if there exists $a \leq t_0\leq b$ such that (see Fig.~\ref{fig:krad_globhyp})
    \begin{equation}
        \text{(i)}\ 
        \begin{cases}
            k(t)\leq 0 &\text{if } t\leq t_0 \\
            k(t)>0, \dot{k}(t)>0 &\text{if } t>t_0
        \end{cases}
        \qquad \text{or}\qquad 
        \text{(ii)}\ 
        \begin{cases}
            k(t)>0, \dot{k}(t)<0 &\text{if } t\leq t_0 \\
            k(t)\leq 0 &\text{if } t>t_0
        \end{cases}\,.
    \end{equation}
\end{thm}
\begin{proof}
    Both cases follows from Lem.~\ref{l_marc}(1) by applying Prop.~\ref{t_radial0} (to a closed interval) and Prop.~\ref{p_radial} (to a complementary open interval). 
\end{proof}
    
\begin{remark}\label{rem:krad_otherGHcases}
    This theorem generalizes Props. \ref{t_radial0} (which can be regarded as a limit case $t_0=a, b$) and \ref{p_radial}, but it does not exhaust all the possibilities for global hyperbolicity.  A trivial one is that $k(t)$ is a positive constant, as the spatial part here can be extended to the whole sphere, in striking difference with all the cases in the theorem (recall Rem.~\ref{r_radial}). 

    A less trivial possibility is to consider a (finite or infinite) sequence of points  $\ldots t_{-2} < t_{-1} <t_0 < t_1 <t_2 \ldots$ where the part with $k\leq 0$ matches with $k>0$ so that Lem.~\ref{l_marc} can be applied (as happens in Thm.~\eqref{thm:krad_globhyp}(ii) for $t_0$). This means that several transitions between $k>0$ and $k\leq 0$ can occur without spoiling global hyperbolicity. This can be done in the $k(t)$-radial metric because it is the only one among the three $k(t)$ spacetimes in which sign changes in $k(t)$ do not correspond to changes in the  topology of the constant-$t$ slices (see Rem.~\ref{rem:krad_topology}). 
\end{remark}

\section{Killing vector fields}
\label{Sec:isometries}

In this section, we work in terms of the dimension of the fiber $\dimF= n-1\geq2$ (i.e., the spheres at constant $t$ and $r$) for convenience and reintroduce the function $a(t)$.

\subsection{Necessary and sufficient conditions for FLRW metric}

\begin{lemma}  \label{FLRW_kdot}
    Let $\teng$ be any of the $k(t)$-radial, $k(t)$-warped or $k(t)$-conformal metrics. Then $\teng$ is locally isometric to a FLRW spacetime if and only if the function $k(t)$ is constant.
\end{lemma}

\begin{proof}
The case when $\teng$ is $k(t)$-warped was proved in \cite[Thm. 4.1]{Avalos2022}), so we only need to consider the $k(t)$-radial and $k(t)$-conformal metrics.
  
Sufficiency is obvious, so we only need to prove that if $\teng$ is locally FLRW then $k(t)$ is necessarily constant. Assume that at every point $p$, there exists an open neighborhood $\U_p$ of $p$ where $\teng$ is isometric to a FLRW metric. 

For the $k(t)$-radial case the component $\{t,r,t,r\}$ of the Weyl tensor is
\begin{equation}
    C_{trtr} = -\frac{\dimF-1}{4(\dimF+1)} \frac{a r^2}{(1 - k r^2)^3} 
    \left[3 a r^2 \dot{k}^2 +2 (1- k r^2) \left ( \dot{a} \dot{k} + a \ddot{k} \right ) \right].
    \label{Eq:Weylkradial}
\end{equation}
This quantity must be zero. The functions  $r^2$ and $1-k r^2$ are functionally independent (as functions of $r$), so the factors in front of them must vanish. This yields $\dot{k}=0$ on $\U_p$. Since $\M$ is connected, we conclude that $k(t)$ is a constant function.

In the case when $\teng$ is $k(t)$-conformal we cannot use a similar argument because now the Weyl tensor vanishes identically (see Prop.~\ref{prop:kconf_Weylzero}). However, in this case the shear of the congruence $\partial_t$, which is geodesic in all cases, vanishes identically (see App.~\ref{app:congruence_kconf}). Moreover, this congruence is hypersurface orthogonal and the orthogonal spaces have constant curvature $k(t)$. Therefore, \cite[Thm. 4.2]{MarsVera2024} ensures that if $\teng$ is locally FLRW then the spatial gradient and Laplacian of the expansion $\Theta$ of $\partial_t$ must vanish along any (arbitrary) timelike curve. It suffices to take the curve at the origin of the radial coordinate, that is $\mathfrak{o}(t)=\{t,r=0\}$, where we have that
\begin{align}
    \mathrm{grad}_{\teng_t} \Theta|_{\mathfrak{o}(t)}&=0\,,\\
    \Delta_{\teng_t} \Theta|_{\mathfrak{o}(t)}&=-\frac{(\dimF+1)^2}{2 a(t)^2}\dot k(t)\,,
\end{align}
where $\teng_t$ is the induced metric on the hypersurfaces of constant $t$, and $\mathrm{grad}_{\teng_t}$ and $\Delta_{\teng_t}$ are respectively the associated gradient and Laplacian. Thus $k(t)$ is constant and the result follows.
\end{proof}

\begin{proposition}\label{perfect-fluid is RW}
  Let $\teng$ be any of the $k(t)$-radial, $k(t)$-warped or $k(t)$-conformal metrics.
  Then $\teng$ is locally isometric to a FLRW spacetime if and only if the Einstein tensor has the algebraic structure of a perfect fluid.
\end{proposition}
\begin{proof}
  The ``only if'' part is direct because FLRW is of perfect fluid type.
  For the ``if'' part it suffices to show that imposing the Einstein tensor to be of perfect fluid type forces  $\dot{k}(t)=0$, and apply Lem.~\ref{FLRW_kdot}. In the three cases the 
  Einstein tensor $G^\mu{}_\nu$ of $\teng$ in spherical coordinates $\{t,r, \theta^A\}$ with $A=1,\ldots,\dimF$  decomposes in the block $\{y^i\}:=\{t,r\}$ and the angular block $\{\theta^A\}$. The angular block is already proportional to the identity, that is, denoting $\theta^1=\theta$, we can write $G^A{}_B=G^\theta{}_\theta \delta^A{}_B$. The assumption of a perfect fluid type implies the existence of a timelike eigendirection $v^i$ in the block  $\{y^i\}$ with eigenvalue $G^\theta{}_\theta$,  that is $G^i{}_j v^j=\lambda v^i$ with $\lambda=G^\theta{}_\theta$. It is thus necessary that the characteristic polynomial $\mbox{det}(G^i{}_j - G^\theta{}_\theta \delta^i{}_j)$ vanishes everywhere, explicitly
  \begin{align}\label{eq:char_poly}
    \mathrm{Pol}:= G^t{}_r G^r{}_t -(G^t{}_t-G^\theta{}_\theta)(G^r{}_r-G^\theta{}_\theta)=0.
  \end{align}
  We shall analyze this equation in the three cases and find that it implies $\dot k(t)=0$. 
  
  In the $k(t)$-conformal case we have
  \begin{equation}
    \mathrm{Pol} = -\dimF^2 r^2\dot{k}^2\frac{1}{a^2(1+k r^2)^2},
  \end{equation}
  whose vanishing clearly requires $\dot{k}=0$.

  In the $k(t)$-radial case the explicit form of $\mathrm{Pol}$ is not that simple,
  but it is of the form
  \begin{equation}
    \mathrm{Pol}=\frac{r^2}{a^3 (1-kr^2)^4}\mathrm{p}(r^2),
  \end{equation}
  where $\mathrm{p}(x)$ is a polynomial of degree 3 in $x$
  with coefficients depending on $a(t)$ and $k(t)$ up to their second
  derivatives. Equation \eqref{eq:char_poly} entails the vanishing
  of those four coefficients. It is not difficult to show
  that they vanish only when $\dot{k}=0$.

  In the $k(t)$-warped case (and for $\dimF\geq 2$) the equation  \eqref{eq:char_poly} takes the form 
  \begin{align}
    (- a^3 \dot{k}^4)r^4 &+\mathrm{c}_0 r^3 C_{k(t)}(r)S_{k(t)}(r) +(\mathrm{c}_1+\mathrm{c}_2 r^2)r^2 S_{k(t)}^2(r) \nonumber\\
    &+(\mathrm{c}_3+\mathrm{c}_4 r^2)r C_{k(t)}(r) S_{k(t)}^3(r)+(\mathrm{c}_5+\mathrm{c}_6 r^2)S_{k(t)}^4(r)=0,
  \end{align}
  where $C_{k(t)}$ is defined in \eqref{eq:Ctr}, $\{\mathrm{c}_p\}_{p=0,1\ldots 6}$ are functions of $a(t)$ and $k(t)$ (up to their second derivatives), and all vanish if $\dot k(t)=0$. Finally, we use that the combinations $r^{p_1} S_{k(t)}^{p_2}(r) C_{k(t)}^{p_3}(r)$ are linearly independent for different values of the exponents $p_1$, $p_2$ and $p_3$. Therefore, just by looking at the first term, we obtain $\dot k(t)=0$, as claimed.
\end{proof}

\subsection{Killing vector fields in spherically symmetric spaces (arbitrary dimension)}

We recall the following facts of the standard sphere $(\mathbb{S}^\dimF, \tengS)$, $\dimF \geq 2$. Viewing $\mathbb{S}^{\dimF}$ as the unit sphere in $\mathbb{R}^{\dimF+1}$ we can define $Y^{\iEi}$ ($\iEi =1, \cdots, \dimF+1$) as the pullback of the Cartesian coordinates $x^{\iEi}$ of $\mathbb{R}^{\dimF+1}$. Let us define the vectors fields on $\mathbb{S}^\dimF$
\begin{equation}
  \vecZ^{\iEi} := \gradS Y^{\iEi}, \qquad
  \vecze^{\iEi \iEj} := Y^{\iEi} \gradS Y^{\iEj} - Y^{\iEj} \gradS Y^{\iEi}.
\end{equation}
It is a well-known fact\footnote{
    The result is easy to prove using the fact that $x^{\iEi}$ has vanishing Hessian in $(\mathbb{R}^{\dimF+1}, g_E)$ and that the rotation generators $x^{\iEi}\partial_{\iEj} - x^{\iEj} \partial_{\iEj}$ are both tangent to $\mathbb{S}^\dimF$ and Killing vectors of Euclidean space.
    } 
that the $(\dimF+1)(\dimF+2)/2$  vector fields\footnote{
    We are denoting those $\vecze^{\iEi\iEj}$ with $\iEi<\iEj$ as  $\vecze^{\iEi<\iEj}$.

}
$\{ \vecZ^{\iEi} ,\vecze^{\iEi<\iEj}\}$ span the conformal Killing algebra of $(\mathbb{S}^\dimF, \tengS)$, while the Killing algebra is spanned by $\{ \vecze^{\iEi< \iEj}\}$. The proper conformal Killing vectors satisfy
\begin{equation} \label{ConfS2}
  \pounds_{\vecZ^{\iEi}} \tengS = - 2 Y^{\iEi} \tengS.
\end{equation}
Consider now a warped product  metric
\begin{equation}
    \teng = \tengB + f^2 \tengS \,,
\end{equation}
where now $\tengB$ is a generic 2-dimensional base and $f$ is an arbitrary warping factor. Later we will particularize to \eqref{eq:basespace}-\eqref{Eq:Base_warped}. Assume that this metric is defined in a product space $\M = \N \times \mathbb{S}^\dimF$ where $ \N$ is the base and $\mathbb{S}^\dimF$ is the fiber. We want to find the Killing vectors of $\teng$. At every point we can decompose any vector field $\vecxi$ as the sum of a vector $\vecxiB$ tangent to the base and a vector $\vecxiF$ tangent to the fiber
\begin{equation}
  \vecxi = \vecxiB + \vecxiF.
\end{equation}
The submanifold $\N_q := \N \times \{q\}$, $q \in \mathbb{S}^\dimF$ is totally geodesic. Hence \cite[Lem. 3.5]{Sanchez1999}  $\vecxiB$ restricted to $\N_q$  is a Killing vector of $\tengB$. The Killing equations are then equivalent to (see \cite{Sanchez1999}):
\begin{align}
  \pounds_{\vecxiF} \tengS = - \frac{2 \vecxiB (f)}{f} \tengS, \label{ABcomps}\\
  h_{ij} \partial_A \xiB^j + f^2 \gS_{AB} \partial_i \xiF^B =0, \label{Aacomps}
\end{align}
where in the second we have used coordinates $\{x^i, x^A\}$ adapted to the product structure. The first equation means that $\vecxiF$  restricted to $\mathbb{S}^m_p := \{p\} \times \mathbb{S}^\dimF$, $p \in \N$ is a conformal Killing vector of $\tengS$. Hence it is a linear combination of $\{\vecZ^{\iEi}, \vecze^{\iEi < \iEj}\}$, i.e.
\begin{equation}
  \vecxiF \Big|_{\mathbb{S}^m_p} = \beta_{\iEi} \vecZ^{\iEi} + \sum_{\iEi < \iEj} \mu_{\iEi \iEj} \vecze^{\iEi \iEj}.
\end{equation}
The ``constants'' $\{ \beta_{\iEi}, \mu_{\iEi \iEj}\}$ can depend on $p$, so they define functions $\beta_{\iEi}, \mu_{\iEi \iEj} : \N \to \mathbb{R}$.

Let $\mathfrak{n}$ be the dimension of the Killing algebra of $(\N,\tengB)$. The case $\mathfrak{n}=0$ is trivial because $\vecxiB$ is necessarily zero and then Eqs.~\eqref{ABcomps} and \eqref{Aacomps} simply state that $\vecxi = \vecxiF$ is a Killing vector of $(\mathbb{S}^\dimF, \tengS)$, so the Killing algebra of the spacetime is $(\dimF(\dimF+1)/2)$-dimensional and spanned by $\{ \vecze^{\iEi < \iEj}\}$.

Let us therefore assume in the remainder that $\mathfrak{n} >0$ and let $\{\vecX_{\mathfrak a}\}$, ($\iLa=1, \cdots, \mathfrak{n})$ be a basis of the Killing algebra of $(\N,\tengB)$. Since $\vecxiB$ is a Killing vector of $(\N,\tengB)$, there exist constants such that $\vecxiB = \alpha^{\iLa} \vecX_{\iLa}$. Note that, as before, the ``constants'' $\alpha^{\iLa}$  can change with the leaf, so in fact $\alpha^{\iLa} : \mathbb{S}^\dimF \to \mathbb{R}$. Equation \eqref{Aacomps} takes the form (cf. \cite[Prop. 3.8]{Sanchez1999})
\begin{equation}
  (\dd \alpha^{\iLa})_A h_{ij} X_{\iLa}{}^j+ f^2 \left[ (\partial_i \beta_{\iEi}) (\dd Y^{\iEi})_A  + \sum_{\iEi < \iEj} (\partial_i \mu_{\iEi \iEj}) (Y^{\iEi} \dd Y^{\iEj} )_A  \right]  =0,
  \label{secondformI}
\end{equation}
while, after using \eqref{ConfS2}, \eqref{ABcomps} is equivalent to
\begin{equation}\label{eq:betaiI}
    \alpha^{\iLa} \vecX_{\iLa}(f)  =   f \beta_{\iEi} Y^{\iEi}.
\end{equation}
The first consequence is that, taking the differential $\ds$ on $\mathbb{S}^\dimF$ of \eqref{secondformI} we obtain
  \begin{equation}\label{eq:laplacian}
    0 = f^2 \sum_{\iEi < \iEj} (\partial_A\mu_{\iEi \iEj}) \ds Y^{\iEi} \wedge \ds Y^{\iEj}  =  0.
  \end{equation}
  To prove that $\mu_{\iEi \iEj}$ are constant, we need to check that the two-forms $\{ \ds Y^{\iEi} \wedge \ds Y^{\iEj}, \iEi<\iEj\}$ are linearly independent on the sphere. To show this, let $c_{\iEi \iEj} = - c_{\iEj\iEi}$ be arbitrary constants satisfying $c_{\iEi \iEj} \ds Y^{\iEi} \wedge \ds Y^{\iEj}=0$. This is equivalent to $\ds ( c_{\iEi \iEj} Y^{\iEi} \dd Y^{\iEj})=0$ and, since $\mathbb{S}^\dimF$ ($\dimF \geq 2$) is simply connected, there exists a function $u$ such that   $c_{\iEi \iEj} Y^{\iEi} \ds Y^{\iEj} = \ds u$. Recalling the definition of $\vecze^ {\iEi \iEj}$ we have that $\sum_{\iEi < \iEj}c_{\iEi \iEj} \vecze^{\iEi\iEj} =  \gradS u$.  So, the Killing vector   $\sum_{\iEi < \iEj}c_{\iEi \iEj} \vecze^{\iEi\iEj}$ is a gradient, which readily implies that it must be identically zero because a Killing vector has, in particular, vanishing divergence  and the only harmonic functions on $\mathbb{S}^\dimF$ are the constants.  Therefore \eqref{eq:laplacian}
implies  that $\mu_{\iEi \iEj}$ are constants, so that \eqref{secondformI} reduces to
\begin{equation}
  \partial_A \left(\alpha^{\iLa} \gB_{ij} X_{\iLa}{}^j+ f^2 \partial_i \beta_{\iEi} Y^{\iEi}\right) =0.
\end{equation}
Using $\partial_i\beta_{\iEi} Y^{\iEi}=\partial_i(\beta_{\iEi} Y^{\iEi})$ and introducing \eqref{eq:betaiI}, this expression becomes
\begin{equation}
  \partial_A \Big[\alpha^{\iLa} \gB_{ij} X_{\iLa}{}^j+ f^2\alpha^{\iLa} \partial_i (\vecX_{\iLa}(\ln f ))\Big] =0,
\end{equation}
which can be rewritten as
\begin{align}
  (\partial_A \alpha^{\iLa}) \vecW_{\iLa} =0
  \label{CombWI}
\end{align}
in terms of the following vector fields on $\N$
\begin{equation}\label{def:Wa}
    \vecW_{\iLa}:=\vecX_{\iLa}+ f^2 \gradB (\vecX_{\iLa}(\ln f )).
\end{equation}

To sum up, any Killing vector $\vecxi$ of $(\M,\teng)$ reads
\begin{equation}
    \vecxi=\alpha^{\iLa} \vecX_{\iLa}+\beta_{\iEi} \vecZ^{\iEi}+ \sum_{\iEi < \iEj}\mu_{\iEi \iEj} \vecze^{\iEi \iEj}\,,
\end{equation}
where the functions $\alpha^\iLa(x^A)$ and $\beta_{\iEi}(x^i)$ satisfy Eqs. \eqref{eq:betaiI} and \eqref{CombWI} and  $\mu_{\iEi \iEj}$ are arbitrary constants. Observe that since $\vecze^{\iEi \iEj}$ span the Killing algebra of the generators of the spherical symmetry of $(\M,\teng)$, the role of $\mu_{\iEi \iEj}$ is trivial.

Now, different cases will arise depending on whether $\{\vecW_{\iLa}\}$ are linearly independent or not. We start with the following result.
\begin{lemma}\label{lemma:case_A}
    Consider $0<\mathfrak{m}\leq \mathfrak{n}$ vectors $\{\vecX_{\iLb}\}$ that form a Killing subalgebra of $(\N,\tengB)$ such that the corresponding vectors $\{\vecW_{\iLb}\}$ are linearly independent. Then $\vecxi=\alpha^{\iLb}\vecX_{\iLb}$, with $\alpha^{\iLb}$ a priori functions on the total space   $\M$, is a Killing vector of $(\M,\teng)$  if and only if $\alpha^{\iLb}$ are constants and $\vecxi(f)=0$.
\end{lemma}
\begin{proof}
    Sufficiency is obvious from \eqref{ABcomps} and \eqref{Aacomps}. For necessity, we already know that $\alpha^{\iLb}$ can at most depend on the angular variables. Since the vectors $\{\vecW_{\iLb}\}$ are linearly independent, Eq.~\eqref{CombWI} requires that $\alpha^{\iLb}$ are constants. In turn, since $\{1,Y^{\iEi}\}$ are linearly independent functions on $\mathbb{S}^\dimF$, Eq.~\eqref{eq:betaiI} holds if and only if $\beta_{\iEi}=0$ and $\vecxi(f)= \alpha^{\iLb} \vecX_{\iLb}(f)=0$.
\end{proof}
The application of the above for $\mathfrak{m}=1$ leads to the following well-known result.
\begin{corollary}\label{coro:X_killing}
    Let $\vecxi$ be a Killing vector of $(\N,\tengB)$. It is also a Killing vector   of $(\M,\teng)$  if and only if $\vecxi(f)=0$.
\end{corollary}

On the other hand, if $\{\vecW_{\iLa}\}$ are linearly dependent we find the following interesting result.
\begin{lemma}\label{lemma:case_B}
  Assume that $\{\vecW_{\iLa}\}$ are linearly dependent, satisfying $\sigma^{\iLa}\vecW_{\iLa}=0$ for $\sigma^{\iLa}\in\mathbb{R}$     not all zero. Then the vectors
    \begin{equation}
      \vecxi^{\iEi} := Y^{\iEi} \vecX + \frac{\vecX(f)}{f} \vecZ^{\iEi},
      \label{xii}
    \end{equation}
    with $\vecX := \sigma^{\iLa} \vecX_{\iLa}$, are Killing vectors of $(\M,\teng)$ and $\{\vecxi^{\iEi},\vecze^{\iEi < \iEj}\}$ are linearly independent. Therefore the Killing algebra of $(\M,\teng)$ is at least $((\dimF+1)(\dimF+2)/2)$-dimensional.\footnotemark
\end{lemma}
\footnotetext{
    Note that the maximal dimension of the Killing algebra in dimension $\dimF+2$ is $(\dimF+2)(\dimF+3)/2$ which is always greater than $(\dimF+1)(\dimF+2)/2$.
}
\begin{proof}
    First note that $\vecX$ is a non-trivial Killing vector of $(\N,\tengB)$. Let us prove that $\{\vecxi^{\iEi}, \vecze^{\iEi < \iEj}\}$ are linearly independent. Consider a vanishing linear combination $a_{\iEi} \vecxi^{\iEi} + \sum_{\iEi < \iEj} b_{\iEi \iEj} \vecze^{\iEi \iEj} =0$. The component of this vector along $\N$  is $(a_{\iEi} Y^{\iEi}) \vecX$. Since there is a point $p \in \N$ where $\vecX \neq 0$, it follows that $a_{\iEi} Y^{\iEi} =0$. This can  happen only if $a_{\iEi}=0$, so the vanishing condition reduces to $\sum_{\iEi < \iEj}b_{\iEi \iEj} \vecze^{\iEi \iEj}=0$, which implies $b_{\iEi \iEj}=0$ at once because $\{\vecze^{\iEi < \iEj}\}$ are linearly independent. So, $\{ \vecxi^{\iEi}, \vecze^{\iEi < \iEj}\}$ are linearly independent. 
  
    It remains to show that each $\vecxi^{\iEi}$ is a Killing vector, or equivalently that $\vecxi := c_{\iEi} \vecxi^{\iEi}$, $c_{\iEi} \in \mathbb{R}$ is a Killing vector. We have found above the necessary and sufficient conditions for this to happen, namely Eqs. \eqref{eq:betaiI} and \eqref{CombWI}.   In the present case we have $\alpha^{\iLa} = (c_{\iEi} Y^{\iEi}) \sigma^{\iLa}$,   $\beta_{\iEi} = \vecX(\ln f)  c_{\iEi}$ and $\mu_{\iEi \iEj}=0$. Equation  \eqref{eq:betaiI} is thus satisfied, while \eqref{CombWI} also holds trivially by virtue of $\sigma^{\iLa} \vecW_{\iLa} =0$.
  \end{proof}
  
\begin{remark} \label{ExistenceX}
   The vector $\vecX$ from Lem.~\ref{lemma:case_B} is {\it not} a Killing vector of $(\M,\teng)$ because for that to be true, we need $\vecX(f)=0$ (see Cor.~\ref{coro:X_killing}). Indeed, since $\vecX = \sigma^{\iLa} \vecX_{\iLa}$ and $\sigma^{\iLa} \vecW_{\iLa} = 0$ it follows that
   \begin{equation}\label{res:cond_on_X_case_B}
      \vecX = \sigma^{\iLa} \vecX_{\iLa} = \sigma^{\iLa} \left (\vecW_{\iLa} - f^2 \gradB (\vecX_{\iLa}(\ln f)) \right)  = - f^2 \gradB (\vecX (\ln f )).
   \end{equation}
   The left-hand side is not identically zero, so it cannot be that $\vecX(f)$ vanishes everywhere. Cor.~\ref{coro:X_killing} thus ensures that $\vecX$ is not a Killing vector of $(\M,\teng)$.
\end{remark}

\begin{remark}\label{res:one_case_B}
    Conversely, if there is a (non-trivial) Killing vector $\vecX$ of $(\N,\tengB)$ that satisfies \eqref{res:cond_on_X_case_B}, then $\vecxi^{\iEi}$ given by \eqref{xii} are Killing vectors of $(\M,\teng)$ and the Killing algebra of  $(\M,\teng)$ is at least $((\dimF+1)(\dimF+2)/2)$-dimensional.
 \end{remark}

The above results allow us to define two cases, that we introduce in the following summarizing result.
\begin{proposition}\label{res:general_spherical}
    Let $(\M,\teng)$ be a warped product metric $\M = \N \times_f \mathbb{S}^\dimF$, $\teng = \tengB + f^2 \tengS$, and let us assume its Killing algebra is of dimension at least $\frac{1}{2}m(m+1)+1$. Necessarily  the Lorentzian space $(\N,\tengB)$ admits a  Killing algebra of dimension $\mathfrak{n}\geq 1$. Let $\{\vecX_{\iLa}\}$ be a basis of the latter, and consider the vectors $\{\vecW_{\iLa}\}$ defined in \eqref{def:Wa}.   Two cases arise.
\begin{itemize}
    \item{Case A:} If $\{\vecW_{\iLa}\}$ are linearly independent, there exist constants $\alpha^{\iLa}$ such that $\vecxi=\alpha^{\iLa} \vecX_{\iLa}$ satisfies $\vecxi(f)=0$.
    \item{Case B:} If $\{\vecW_{\iLa}\}$ are linearly dependent, then Lem.~\ref{lemma:case_B} holds and the Killing algebra of $(\M,\teng)$ is at least $((\dimF+1)(\dimF+2)/2)$-dimensional. If $\vecX$ is spacelike in a domain $\U$, then $(\U,\teng)$ is locally isometric to  a FLRW spacetime.
\end{itemize}
\end{proposition}
\begin{proof}
    We have already used the fact that the parallel part to $\N$ of any Killing of $(\M,\teng)$ must be  a Killing of $(\N,\tengB)$  \cite[Lem. 3.5]{Sanchez1999}. Since by assumption the Killing algebra of $(\M,\teng)$ is at least $\dimF+2$, not all Killing vectors can be everywhere tangent to the spheres. This proves that $\mathfrak{n} \geq 1$. Now, Case A follows from Lem.~\ref{lemma:case_A} using $\mathfrak{m}=\mathfrak{n}$. Case B corresponds to Lem.~\ref{lemma:case_B}. In Case B, at points where $\vecX$ is spacelike, the Killing fields $\vecxi^{\iEi}$ are spacelike, and therefore  $(\U,\teng)$ admits $(\dimF+1)(\dimF+2)/2$ linearly independent Killing vectors acting on spacelike hypersurfaces, so it is locally FLRW \cite{Stephani2003}.
\end{proof}
The combination of this proposition with Rem.~\ref{ExistenceX} yields the following result. 
\begin{corollary}\label{coro:X two cases}
  In the setup of Prop.~\ref{res:general_spherical}, the Lorentzian space $(\N,\tengB)$ must admit a Killing vector $\vecX$ such that either $\vecX(f)=0$, and therefore $\vecX$ is Killing of $(\M,\teng)$, or $\vecX = - f^2 \gradB (\vecX (\ln f ))$, and therefore  $(\M,\teng)$ admits at least a  $((\dimF+1)(\dimF+2)/2)$-dimensional Killing algebra given by $\{\vecxi^{\iEi}, \vecze^{\iEi < \iEj}\}$.
\end{corollary}

On the other hand, since $\vecX$ of Case B is a Killing vector of $(\N,\tengB)$, it is necessary orthogonal to $\dd \RB$ (where $\RB$ is the scalar curvature of $\tengB$), so we also have the following corollary.
\begin{corollary}\label{coro:gradient}
    Consider the setup of Prop.~\ref{res:general_spherical}, and assume Case B holds. In domains where  $\dd \RB$ is timelike, the metric is locally FLRW.
\end{corollary}

\subsection{Killing vector fields of the \texorpdfstring{$k(t)$}{k(t)} geometries}

So far all the considerations in this section have been for general spherically symmetric spacetimes, i.e. for warped products of the form \eqref{eq:warpedprod} with arbitrary $\tengB$ and warping factor $f$, for $\dimF\geq 2$. We now apply these results to analyze the Killing vectors for $\tengB$ given by \eqref{eq:basespace} and, in particular, for the $k(t)$-radial, $k(t)$-warped and $k(t)$-conformal metrics.

We start by obtaining necessary conditions that the three $k(t)$ metrics must satisfy in order to admit more isometries than just spherical symmetry. The case when the $k(t)$ metric is in fact an FLRW is trivial, so throughout this section we assume that the $k(t)$ metrics are not locally isometric to FLRW. By Prop.~\ref{res:general_spherical} we have two possible cases, A or B.

A necessary condition for $\teng$ to belong to Case A is that the two-form
\begin{equation}\label{eq:cond_cas_A}
    \dfom_1 := \dd \RB \wedge \dd f
\end{equation}
vanishes identically. This is because $\dfom_1 (\vecxi,\cdot)=0$ (as $\vecxi(  \RB)=0$ for any Killing vector of $(\N,\tengB)$ and $\vecxi(f)=0$ since we are in Case A) and a two-form in two dimensions is proportional to the volume form $\dfeta_{\tengB}$ of $(\N,\tengB)$, so at any point it is either zero or has trivial kernel. Define the function $F_1$ by means of $\dfom_1= F_1 \dfeta_{\tengB}$.\footnote{
    We use the orientation in $\{t,r\}$ so that     $\dfeta_{\tengB} (\partial_t,\partial_r)$ is positive.
    } 
We thus need to analyze the consequences of imposing $F_1 =0$. For each $k(t)$ metric evaluating $F_1=0$ at $r=0$ yields
\begin{equation}
    \frac{\dd}{\dd t} \left( \frac{\ddot{a}}{a} \right)=0,
\end{equation}
so $\ddot{a} = c a$, with $c \in \mathbb{R}$.

As for Case B, we note that the squared norm of $\dd \RB$ at  $r=0$ in all three $k(t)$ metrics takes the form
\begin{equation}
    \gB^{ij}\partial_i\RB \partial_j\RB \Big|_{r=0} = - 4 \left ( \frac{\dd}{\dd t} \left( \frac{\ddot{a}}{a} \right) \right)^2 \leq 0.
\end{equation}
Now, if $a^{-1} \ddot{a}$ is not a constant function, then Cor.~\ref{coro:gradient} implies that $(\M,\teng)$ is locally isometric to FLRW in some non-empty domain. But then by Lem.~\ref{FLRW_kdot} we have that $\dot{k}=0$, so in fact the metric is locally FLRW everywhere. Therefore,  $\ddot{a} = c a$, with $c \in \mathbb{R}$ also in Case B. Moreover,  the scalar curvature $\R$ of $(\M,\teng)$ is constant on the $\mathbb{S}^\dimF$ fibers, so it descends to a function on $\N$. The two-form on $(\N,\tengB)$
\begin{equation}
    \dfom_2 :=  \dd   \RB \wedge \dd \R
\end{equation}
must vanish identically whenever there exists $\vecX$ in  Case B. The argument is as before because  $\vecX (  \RB) =0$ (since $\vecX$ is a Killing vector of $(\N,\tengB)$) and also $\vecX (\R) =0$, because $\vecZ^{\iEi} (\R)=0$ due to spherical symmetry and therefore
\begin{equation}
    0 =\Lie_{\vecxi^{\iEi}}\R= \left(Y^{\iEi} \vecX + \frac{\vecX(f)}{f} \vecZ^{\iEi}\right) (\R) = Y^{\iEi} \vecX (\R)
\end{equation}
and $Y^{\iEi}$ are obviously not identically zero. For each $k(t)$ metric we thus need to analyze the consequences of imposing $F_2 =0$ defined by $\dfom_2 = F_2 \dfeta_{\tengB}$. Observe that to compute $F_2$ we have to use the following identity
\begin{equation}
    \R= \RB + \dimF(\dimF-1)\frac{1}{f^2}\left(1-h^{ij}\partial_if \partial_j f\right)-2\dimF\frac{1}{f}\Delta_{\tengB} f.
\end{equation}

We have now enough ingredients to analyze the possible Killing algebras that the $k(t)$ metrics, apart from FLRW, admit.
In the $k(t)$-radial and $k(t)$-conformal cases the vanishing of $F_2$ will suffice to show that, leaving aside FLRW  metrics, Case B never happens. However, in the  $k(t)$-warped case the condition  $\ddot{a} = c a$ is equivalent to $\RB = 2c$, and therefore  $F_2 =0$ provides no information. We will then have to use a different approach to deal with Case B in the $k(t)$-warped case.
Observe that, by the considerations above,the existence of more isometries than spherical symmetry implies $\ddot{a} = c a$. Let us state the main result, which we next prove for each $k(t)$ metric in three different subsections.

\begin{thm}
\label{resultKV}
  Consider the $k(t)$-warped, the $k(t)$-radial and the $k(t)$-conformal metrics.  If the metrics are not locally FLRW, then the Killing algebra is three dimensional, generating the spherical symmetry, except in the particular case $a(t)=a_0e^{b t}$, $k(t)=k_0 e^{2bt}$ for constants $a_0$, $k_0$ and $b$, in which the three metrics admit an additional Killing vector
  \begin{equation}\label{eq:kill_final}
    \vecX=\partial_t- b r \partial_r\,.
  \end{equation}
\end{thm}

\begin{remark}
  Observe that the character of the Killing vector \eqref{eq:kill_final} may differ in different regions.
  The norm of $\vecX$ is 
  \begin{align}
    k(t)\mbox{-warped}:&\quad \gB_{ij}X^i X^j=-1+a_0^2b^2r^2e^{2bt},\\
    k(t)\mbox{-conformal}:&\quad \gB_{ij}X^i X^j=\frac{4e^{-2bt}a_0^2 b^2 r^2-\left(e^{-2bt} + k_0 r^2\right)^2}{e^{-2bt} + k_0 r^2},\\
    k(t)\mbox{-radial}:&\quad \gB_{ij}X^i X^j=-\frac{e^{-2bt}-(k_0+a_0^2b^2)r^2}{e^{-2bt}-k_0r^2}.
\end{align}
Note that the denominators do not vanish in the corresponding charts as defined in Subsec.~\ref{ssec:def_kt_metrics}.

\end{remark}

\subsubsection{Proof for the \texorpdfstring{$k(t)$}{k(t)}-warped metric}

We already know that we need $\ddot{a} = c a$, $c \in \mathbb{R}$ for the existence of an additional isometry to the spherical symmetry. In the $k(t)$-warped case this is equivalent to $(\N,\tengB)$ being of constant curvature. Define the constant $a_0$ by the first integral $\dot{a}^2 = c a^2 - a_0$. Decomposing the Killing vector of Cor.~\ref{coro:X two cases} as $\vecX =  X^t \partial_t + X^{r} \partial_{r}$, the Killing equations for $\vecX$ are
\begin{equation}
    \partial_t X^t =0, \qquad \quad
    \partial_{r} X^{t} - a^2 \partial_t X^{r} =0, \qquad \quad
    \dot{a} X^t + a \partial_{r} X^{r} =0.
\end{equation}
The first equation gives $X^t(r)$, so the first and third equations can be integrated in terms of a free function $F(r)$ and a free function $B(t)$ as
\begin{equation}
    X^{r} = - \frac{\dot{a}}{a} F(r) + B(t), \qquad \qquad  X^t = F'\,,
\end{equation}
where prime denotes derivative with respect to $r$. Inserting into the second equation gives
\begin{equation}
    F'' + a_0 F = a^2 \dot{B}.
\end{equation}
By separation of variables, there is a constant $c_0$ such that
\begin{equation}
    \dot{B} = \frac{c_0}{a^2}, \qquad \quad F'' + a_0 F = c_0.
    \label{eqsBF}
\end{equation}
Note that the functions $F(r)$ and $B(t)$ are not univocally defined since the shift $F = \hat{F} + f_0$, $B = \hat{B} + f_0 \frac{\dot{a}}{a}$, with $f_0$ constant, leaves $\vecX$ invariant. Without loss of generality we can exploit this shift to set $F(0)=0$, which we assume from now on.

We now analyze the first possibility arising from Cor.~\ref{coro:X two cases}, $\vecX(f)=0$, which reads explicitly
\begin{equation}\label{eq:Xfzero}
    S_{k(t)} F' \left (\dot{a}  - a \frac{\dot{k}}{2 k} \right ) + a C_{k(t)} \left ( B - \frac{\dot{a}}{a} F + r \frac{\dot{k}}{2k} F' \right )=0.
\end{equation}
Evaluating at $r=0$ and using $S_{k(t)} (0)=0$, $C_{k(t)} (0) =1$ gives $B(t)=0$.
Imposing the first equation in \eqref{eqsBF} gives $c_0=0$,
so we have found that
\begin{align}
    \vecX= F' \partial_t - \frac{\dot{a}}{a} F \partial_{r}, \qquad \quad F'' + a_0 F =0, \quad F(0)=0. \label{ODE}
\end{align}
The general solution for $F$ is $F = f_1 S_{a_0}(r)$ with $f_1 \in  \mathbb{R}$. We assume $f_1 \neq 0$ as otherwise the vector $\vecX$ is identically zero. We still need to impose $\vecX(f)=0$ everywhere. It is simplest to analyze the Taylor series of Eq.~\eqref{eq:Xfzero} at $r=0$.
It suffices to compute the first two non-trivial terms. The result is
\begin{equation}
    0 = \frac{6}{f_1} \vecX(f) = \left ( 2 (k -a_0) \dot{a} - a \dot{k} \right ) r^3  + \frac{1}{10} \left (  2 (a_0^2 -k^2) \dot{a} +5 a_0 a \dot{k} + a k \dot{k} \right )  r^5 + O(r^7).
\end{equation}
The first term gives $\dot{k} = 2 ( k- a_0) \frac{\dot{a}}{a}$, which inserted into the second one gives $a_0 \dot{k} =0$. Using the hypothesis that $\dot{k} \neq 0$ we conclude $a_0=0$. Hence $\dot{a}^2 = c a^2$. This requires that $c \geq 0$ and there exists $b \in \mathbb{R}$ such that $\dot{a} = b a$. Moreover, $\dot{k} = 2 b k$ and $F = f_1 r$. The Killing vector is (dropping the irrelevant multiplicative constant $f_1$)
\begin{equation}\label{eq:kill_warped}
    \vecX= \partial_t - b r \partial_{r}.
\end{equation}
It is immediate to check that $\vecX(f)=0$.

We now consider the second possibility of Cor.~\ref{coro:X two cases}, namely that
\begin{equation}
    V^t\partial_t+V^r\partial_r :=\vecX+ f^2 \gradB (\vecX(\ln f ))=0\,.
\end{equation}
A straightforward calculation shows that $V^t$ evaluated at $r=0$ is $V^t|_{r=0}=F'(0)$. Thus $F'(0)=0$, and therefore the solution of \eqref{eqsBF} is unique, depends only on $c_0$, and we have $F=0$ iff $c_0=0$. The Taylor expansions of $V^t$ and $V^r$ around $r=0$ take the form
\begin{align}
    &V^r =  - c_0 \left ( V_4(t)r^4 + V_6(t)r^6 \right ) + O(r^7),\\
    &V^t=-\frac{1}{3}\left(- B \dot k a^2 +2 c_0(a_0-k)\right)r^3 + V_5(t)r^5 + O(r^6),
\end{align}
where the functions $V_4(t), V_5(t), V_6(t)$ depend on $a(t)$, $B(t)$, $k(t)$ and their derivatives. It turns out that the pair of equations $V_4(t)=V_6(t)=0$ lead to FLRW since
\begin{equation}
    \left(\frac92 a_0 + 12 k \right) V_4 + 45 V_6 = \dot{k}(k-4 a_0)\,.
\end{equation} 
From $V^r=0$ we therefore conclude that there will be no further Killings unless $c_0=0$, or, equivalently, $F=0$.  Moreover, from the first term in $V^t$, since $\dot k$ is not identically zero, we obtain $B(t)=0$ and hence  $\vecX=0$. Summarizing, the only solution of $\vecX+ f^2 \gradB (\vecX(\ln f ))=0$ for $\dot k\neq 0$ is $\vecX=0$, and the result is proved.
  
\subsubsection{Proof for the \texorpdfstring{$k(t)$}{k(t)}-conformal metric}
We start considering Case A of Prop.~\ref{res:general_spherical}. We recall that because of \eqref{eq:cond_cas_A} $F_1$ must vanish. In the $k(t)$-conformal metric the function $F_1$ takes the form
\begin{equation}
    F_1 = W^3 r^2 \sum_{s=0}^2 C_s(t) r^{2s}. 
\end{equation}
The combination $\mathcal{C} := C_2 - k C_1 + k^2 C_0$ is
\begin{equation}
    \mathcal{C}  = \frac{ \dot{k}}{a^4} \left ( -2 a k \ddot{k} + 2 \dot{a} \dot{k} k + a \dot{k}^2 \right ).
\end{equation}
We again leave out the case $\dot{k}=0$, so $\mathcal{C} =0$  provides an expression for $\ddot{k}$ which inserted back into $C_1$ gives
\begin{equation}
    C_1 = \frac{3 \dot{k}}{8 k a^5} \left ( 2 \dot{a} k - a \dot{k} \right )^2,
\end{equation}
so
$k=k_1 a^2$, $k_1 \in \mathbb{R}$. This is compatible with $\mathcal{C} =0$ only if $\dot{a} = b a$, $c = b^2$. When this holds, $F_1$ vanishes identically. Choosing the integration constants so that   $a=a_0e^{bt}$ and $k_1=k_0/a^2_0$, $k_0 \in \mathbb{R}$, the metric in this case is 
\begin{equation}
    \teng = - \dd t^2  + \left ( \frac{2 a_0 e^{-bt}}{e^{-2bt} + k_0 r^2} \right )^2   \left ( \dd r^2 + r^2 \tengS \right ),
    \label{kconfCaseA}
\end{equation}
and one checks easily that  $(\N,\tengB)$ admits (up to a constant factor) a single Killing vector given by
\begin{equation}
    \vecX  = \partial_t - b r \partial_r.
\end{equation}
This vector field is also a Killing vector of $(\M,\teng)$ since $\vecX(f)=0$. We conclude that the only  $k(t)$-conformal and not FLRW metric that fits into Case A is \eqref{kconfCaseA}. 

As for the Case B, in the $k(t)$-conformal case the function $F_2$ takes the form
\begin{equation}
    F_2 = \frac{m(m+1)}{2a^6} W^3 r \sum_{s=0}^2 G_s(t) r^{2s}.
\end{equation}
The functions $G_s$ are polynomials in $\{\dot{k}, \ddot{k},\dddot{k}\}$ with $G_0$ independent of $\dddot{k}$ and $G_1, G_2$ linear in this variable. One can then eliminate $\dddot{k}$ by computing the resultant $\mbox{res}_1$  of  $G_1$ and $G_2$ with respect to $\dddot k$. Then we eliminate $\ddot{k}$ by computing the resultant of $\mbox{res}_1$ and $G_0$ with respect to $\ddot k$.
  
This final resultant reads:
\begin{equation}
    \left(2 \dot{a} k -  a \dot{k}\right)  \left( 2 \dot{a}^3 - 2c a^2 \dot{a} +2 \dot{a} k - a\dot{k} \right) \dot{k}^{4} =0.
\end{equation}
Hence, two cases arise depending in which of the first two factors vanishes.  In either case we can solve for $\dot{k}$. Substituting back into the equations $G_s =0$ it follows easily that the only possibility with non-constant $k$ is, again, $k = k_0 e^{2bt}, a = a_0 e^{bt}$, that belongs to Case A. The result for the $k(t)$-conformal metric thus follows.

\subsubsection{Proof for the \texorpdfstring{$k(t)$}{k(t)}-radial metric}

We now move on to the $k(t)$-radial metric and start with Case A.
The function $F_1$ introduced after \eqref{eq:cond_cas_A} now takes the form
\begin{equation}
    F_1 = W^{5} r^2 \sum_{s=0}^2 D_s(t) r^{2s}. 
\end{equation}
The condition $F_1=0$ then requires $D_0$, $D_1$ and $D_2$ to vanish, and hence the combination 
\begin{equation}
    \mathcal{D} := \sum_{s=0}^2 D_s k^{2-s} = \frac{3}{a^6} \dot{k}^2 \left ( \dot{k} a - 2 k \dot{a} \right ).
\end{equation}
As usual, we leave out the case with $\dot{k} =0$. Thus $\dot{k} = 2 k a^{-1} \dot{a}$. The expression of $D_0$ becomes
\begin{equation}
    D_0 = - \frac{12 k \dot{a}}{a^8} \left (\dot{a}^2 - c a^2  \right )\,,
\end{equation}
so it must be that $\dot{a} = b a$, with $b \in \mathbb{R}$ and $c= b^2 \neq 0$. The function $F_1$ vanishes identically in this case. Inserting $a(t) = a_0 e^{bt}$ and $k(t) = k_0 e^{2bt}$ in the $k(t)$-radial metric leads to:
\begin{align}
    \teng = -\dd t^2 + \frac{a_0^2}{e^{-2bt} - k_0r^2} \dd r^2  + a_0^2  e^{2bt} r^2 \tengS\,.\label{kradialcaseB}
\end{align}
It is easy to find that $(\N,\tengB)$ admits a unique (up to a constant factor) Killing vector given by
\begin{equation}
    \vecX = \partial_t - b r \partial_r\,.
\end{equation}
This vector field is also a Killing vector of $(\M,\teng)$ because $\vecX(f)=0$. So, we conclude that the only situation where a $k(t)$-radial metric is not a FLRW and belongs to Case A is \eqref{kradialcaseB}. 

To deal with Case B we analyze the vanishing of the function $F_2$, introduced also at the beginning of this subsection.
A direct computation gives that $F_2$ takes the form
\begin{equation}
    F_2 =r W^7 \sum_{s=0}^2 Q_s(t) r^{2s}.
\end{equation}
Thus, if $F_2=0$, each  $Q_s$  must vanish identically. The combination ${\mathcal Q} := \sum_{s=0} Q_s k^{2-s}$ just reads
\begin{equation}
    \mathcal{Q} = \frac{6\dimF}{a^{10}} \dot{k}^3 \left ( \dot{a}^2 - c a^2 \right ),
\end{equation}
so $\dot{a} =  b a$, $c = b^2$. Then $Q_0$ becomes
\begin{equation}
    Q_0 = - \frac{2\dimF(\dimF+1)}{a^{10}}   \left ( \ddot{k} + 2 b \dot{k} \right ) \left ( \dot{k} - 2 b k  \right ),
\end{equation}
and therefore we conclude that either $\dot{k} = 2 b k$ or $\dot{k} = -2  b k + b_0$ with $b_0 \in \mathbb{R}$. In the latter case one gets
\begin{equation}
    Q_1 = \frac{6\dimF(\dimF+1)}{a^{10}} (4 b k - b_0) (2b k - b_0)^2 \,,    
\end{equation}
which implies that $k$ is constant, against hypothesis. The case $\dot{k} = 2 b k$ has already been considered when we studied Case A. In that case the metric $\tengB$ admits a single (up to scale) Killing field $\vecX$, but it satisfies $\vecX(f)=0$, so we fall out of Case B, and the result follows.

\section{Comparing with other geometries}
\label{sec:StephaniLTB} 

In this last section, we study whether the $k(t)$ metrics are related with each other and whether they can be classified within some known families of inhomogeneous cosmological spacetimes.

\subsection{Inequivalence among the three \texorpdfstring{$k(t)$}{k(t)} metrics}

First, we show that for generic choices of $k(t)$ the three spacetimes are not equivalent among themselves. We begin proving the inequivalence between the $k(t)$-conformal and the rest. 

\begin{proposition}\label{prop:ineqkt1}
    The $k(t)$-conformal metric is not locally isometric to either the $k(t)$-radial nor the $k(t)$-warped when $\dot{k}(t) \neq 0$. 
\end{proposition}

\begin{proof} 
    The $k(t)$-conformal is locally conformally flat (see Prop.~\ref{prop:kconf_Weylzero}), whereas the $k(t)$-radial (see Eq.~\eqref{Eq:Weylkradial}) and the $k(t)$-warped (see \cite[Thm.~4.1]{Avalos2022}) are not.
\end{proof}

\noindent Now it remains to show that the $k(t)$-warped and the $k(t)$-radial are not isometric. 

\begin{proposition}\label{prop:ineqkt2}
    The $k(t)$-warped spacetime is not locally isometric to the $k(t)$-radial when $\dot{k}(t) \neq 0$. 
\end{proposition}

\begin{proof}
    We separate the case when the Killing algebra of the spacetime is just $\mathrm{so}(n)$ from the case when it is $\mathrm{so}(n) \oplus \mathbb{R}$. 
        
    Assume first that both  metrics admit the additional Killing vector $\vecX=\partial_t- b r \partial_r$. We use the expressions for $a(t), k(t)$ given in Result \ref{resultKV} and we assume that $b\neq 0$ (otherwise $\dot{k}=0$ against hypothesis). The orbits of the Killing algebra are the hypersurfaces $\HH_{u_0}$ defined by $r e^{bt} = u_0$, with $u_0$ a positive constant. A first necessary condition for the existence of a local isometry between $\gw$ and $\gr$ is that, to each value of $u_0>0$, there exists a value $v_0>0$ such that $(\HH_{u_0},\qw_{u_0})  \subset (\uwar, \gw)$ is locally isometric to $(\HH_{v_0},\qr_{v_0})  \subset (\urad, \gr)$, where $\qw_{u_0}, \qr_{v_0}$ are the respective first fundamental forms. A simple computation gives
    \begin{align}
          \qw_{u_0}& =( -1 + a_0^2 u_0^2 b^2 ) \dd t^2 + a_0^2 S^2_{k_0} (u_0) \tengS\, ,  \label{Eq:qw} \\
          \qr_{v_0} & = \left ( -1 + \frac{a_0^2 v_0^2 b^2}{1- k_0 v_0^2} \right )  \dd t^2 + a_0^2 v_0^2 \tengS \label{Eq:qr} \,.
    \end{align}   
    For generic values of $u_0, v_0$, these tensors are non-degenerate, and represent metrics of cylinders $\mathbb{R}\times \mathbb{S}^{n-1}$. The Ricci tensor of both  metrics~\eqref{Eq:qw} and ~\eqref{Eq:qr} has a single eigenvector with vanishing eigenvalue, namely the direction spanned by $\partial_t$. The local isometry must respect these eigenspaces. Consequently, it must respect also its orthogonal planes. Thus, the isometry between the cylinders $\qw_{u_0}$ and $\qr_{v_0}$ will necessarily have to also map isometrically $(n-1)$-spheres of constant $t$ to $(n-1)$-spheres, that is, to map the orbits of the Killing algebra $\mathrm{so}(n)$ to each other. As a result, the assumed local isometry between $\gw$ and $\gr$ must also map the 2-surfaces orthogonal to the $(n-1)$-spheres to each other. The scalar curvatures of the metrics of those 2-surfaces, namely the base space metrics $\tengB^{\kwar}$ and $\tengB^{\krad}$, are $\R^{(\tengB^{\kwar})} =2b^2$ and $\R^{(\tengB^{\krad})} = 2b^2(2k_0 r^2e^{2bt}+1)/(k_0 r^2e^{2bt}-1)^2$. The first is constant while the second  is not. Hence the metrics $\gw$ and $\gr$ are not locally isometric in the case when both admit the additional Killing vector $\vecX$.

    When there is no additional Killing vector, the Killing algebra is $\mathrm{so}(n)$ and the orbits are, for both metrics, the spheres of constant $t$ and $r$. A necessary condition for the existence of a local isometry is that the orbits of the Killing algebra are mapped locally isometrically to each other, and that the same happens for their orthogonal $2$-surfaces. In other words, the base spaces  of both spacetimes must be locally isometric and, in addition, this local isometry must map the warping function $f^{\kwar}(t,r) = a(t)r$ of $\gr$ to the warping function $f^{\krad}(t,r) := a(t) S_{k(t)}(r)$ of  $\gw$. The metric of the base space of $\gw$ is $\tengB^{\kwar} = -\dd t^2 + a(r)^2 \dd r^2$. This metric admits a Killing vector $\vecxi= \partial_r$ which satisfies the following three properties ($\vecxi^{\perp}$ is any vector field everywhere perpendicular to $\vecxi$, $\Lie$ denotes Lie derivative and $\R^{(\teng)}$ denotes the scalar curvature of a metric $\teng$):
    \begin{itemize}
      \item[(i)] $\vecxi$ is spacelike everywhere.
      \item[(ii)] $\Lie_{\vecxi} \R^{(\tengB^{\kwar})} =0$.
      \item[(iii)] $\Lie_{\vecxi^{\perp}} \left ( \frac{1}{\sqrt{\tengB^{\kwar}(\vecxi,\vecxi)}} \Lie_{\vecxi} f^{\kwar} \right ) =0$.
    \end{itemize}
    A necessary condition for existence of a local isometry between  $\gw$ and $\gr$ is the existence of a Killing vector $\hat{\vecxi}$ of the base space metric $\tengB^{\krad} = - \dd t^2 + a(t)^2/(1 - k(t) r^2) \dd r^2$ satisfying the same properties (i)-(iii) above with all quantities referred to the $k(t)$-radial spacetime. We decompose the Killing as $\hat{\vecxi} = A(t,r) \partial_t + B(r,t) \partial_r$ and note that $B$ cannot vanish anywhere because of (i). The Killing equations for $\hat{\vecxi}$ are equivalent to
    \begin{align}
        \dot{A} =0, \qquad \dot{B} = \frac{(1 - k r^2)}{a^2} A',
        \qquad B' =
        - A \left ( \frac{\dot{a}}{a} + \frac{\dot{k} r^2}{2 (1-k r^2)} \right )
        - \frac{k r B}{1- k r^2}. \label{killeqs}
    \end{align} 
    Conditions (ii) and (iii) are equations for $A,B$ which, upon using equations \eqref{killeqs} in the second one, become polynomial in $A$ and $B$ (the first one is homogeneous of degree one,  while the second one is homogeneous of degree two). They can be easily combined so as to make $A$ disappear. The result is an equation of the form
    \begin{equation}
        r B^2 \sum_{i=0}^6 {\mathcal P}_i(t) r^{2i} = 0 ,
    \end{equation}
    where ${\mathcal P}_i(t)$ depend only on $a(t),k(t)$ and their derivatives. Since $B$ is nowhere zero, the six equations ${\mathcal P}_i(t)=0$ must hold true.  It is useful to compute the combination
    \begin{equation}
        \sum_{i=0}^6 {\mathcal P}_i(t) k^{6-i} = 9 \dot{k}^6 a^6
        \left ( \dot{k} - 2 k \frac{\dot{a}}{a} \right ) =0 .
    \end{equation}
    Thus, $\dot{k} = 2 k a^{-1} \dot{a}$. Note that $\dot{a}$ cannot be identically zero because we are assuming $\dot{k} \neq 0$. After inserting  this into ${\mathcal P}_i$, it is advantageous to compute the combination
    \begin{equation}
        \sum_{i=1}^6 i  {\mathcal P}_i(t) k^{6-i} = 1728 \, \dot{a}^5
        k^7 \left ( -\ddot{a} + \frac{\dot{a}^2}{a} \right ) = 0.
    \end{equation}
    Hence $\ddot{a} = a^{-1} \dot{a}^2$ which can be integrated once to give $\dot{a} = b a$, with $b$ constant. But then $a(t)= a_0 e^{bt}$ and $k(t) = k_0 e^{2bt}$ and the spacetime $\gr$ admits an additional Killing vector, against hypothesis in the present case.
\end{proof}

\subsection{Comparison of the \texorpdfstring{$k(t)$}{k(t)} metrics with the Stephani universe}

Let us take the extension of the Stephani universe to $n+1$  dimensions\footnote{
    In principle, this form of the metric is slightly more general than the one used in the book~\cite[Sec.~19.7]{Ellis2012}, see~\cite{Clarkson1999}.
}
\begin{align}\label{eq:Stephani}
    \teng_{\text{Steph}} = - \left(\frac{3}{\Theta(t)}\frac{\dot{U}(t,\boldsymbol{x})}{U(t,\boldsymbol{x})}\right)^2 \dd t^2 +  \frac{1}{U(t,\boldsymbol{x})^2} (\dd r^2 + r^2 \tengS) ,
\end{align}
with
\begin{align}
    U(t,\boldsymbol{x}) := a(t) + b(t) r^2 - 2\boldsymbol{c}(t) \cdot \boldsymbol{x}\,,
\end{align}
where, in principle, $\Theta(t),a(t),b(t)$ and $\boldsymbol{c}(t) = \left(c_1(t),c_2(t),...,c_n(t)\right)$ are arbitrary functions of time, $\boldsymbol{x} = (x^1,x^2,...,x^n)$ is the position vector in Euclidean space and $r:=\sqrt{\boldsymbol{x}\cdot\boldsymbol{x}}$. Our expression is based on that presented in \cite{Barnes1998}. The two properties that invariantly characterize the geometry are:
\begin{enumerate}
    \item The metric is locally conformally flat for $n\geq3$, as the Weyl tensor identically vanishes for \eqref{eq:Stephani}.
    
    \item The matter supporting this geometry corresponds to a perfect fluid. This means that the components of the Einstein tensor $G^\mu{}_\nu$ (in the holonomic basis $\{t,r,\theta^A\}$, where it is a diagonal matrix) exhibit the same eigenvalue along the spatial directions (i.e., the fluid is isotropic).
\end{enumerate}

We can use the second property to prove the following result: 

\begin{proposition}\label{prop:ineqStephani}
    None of the $k(t)$  metrics with non-constant $k(t)$ are locally isometric to the Stephani class \eqref{eq:Stephani}.
\end{proposition}
\begin{proof}
    By Prop.~\ref{perfect-fluid is RW}, if a $k(t)$ metric is of perfect fluid type then it is FLRW, and
    thus, by Lem.~\ref{FLRW_kdot}, $k(t)$ must be constant, against hypothesis.
\end{proof}

As a complementary comment, we notice that not all the Stephani geometries are spherically symmetric, contrary to the $k(t)$ metrics, which always exhibit an $\mathrm{SO}(n)$ symmetry group. In fact, for the specific case $n=3$, the isometries of the Stephani spacetimes where fully characterized in \cite{Barnes1998}. 

In that case, depending on the rank $\mathfrak{r}$ of the vector space spanned by the functions $a(t),b(t), \boldsymbol{c}(t)$, the metric displays a different amount of symmetries. Following~\cite{Barnes1998}, we have that for $\mathfrak{r} \geq 4$, there are no Killing vector fields, for $\mathfrak{r} = 3$, there is a one-dimensional symmetry group, if $\mathfrak{r} = 2$ the spacetime admits a three-dimensional group of isometries with two-dimensional orbits and in the case in which $\mathfrak{r} = 1$, the spacetime is FLRW. Thus, whenever $\mathfrak{r} \neq 2$, the spacetime is automatically not isometric to any of the $k(t)$ spacetimes. Furthermore, the exceptional case for which an additional Killing vector fields exists can never be isometric to the Stephani universes, as they never exhibit a four dimensional group of symmetries. 

Even though we are able to show that our $k(t)$ metrics are never isometric to the Stephani universes for $\dot{k} \neq 0$, a more general classification of the symmetries of the Stephani universes for $n > 3$ is lacking in the literature and it would be very interesting to perform. 

%
\subsection{Comparison of the \texorpdfstring{$k(t)$}{k(t)} metrics with the Lema\^itre-Tolman-Bondi metric}

Another family of inhomogeneous cosmological spacetimes is given by the LTB metric (see \cite[Chap. 15]{Ellis2012}):
\begin{equation}
    \label{eq:LTBmetric}
    \teng_{\text{LTB}} = - \dd t^2 + \frac{R'(t,r)^2}{1 + 2 E(r) } \dd r^2 + R(t,r)^2 \tengS\,,
\end{equation}
where we have generalized the geometry to arbitrary dimensions by extending the 2-dimensional spherical fibers to $(n-1)$-spheres. We notice that the spacetime is spherically symmetric, i.e., it exhibits a
$\mathrm{SO}(n)$ group of symmetry as all the $k(t)$ metrics. Here, $E(r)$ is an arbitrary function and $R(t,r)$ is constrained to satisfy
\begin{equation}
    \label{eq:LTBcondition}
    \dot{R}^2 = 2 E + \frac{2M}{R^{n-2}} + \frac{2\Lambda}{n(n-1)} R^2 ,
\end{equation}
with $\Lambda$ the cosmological constant and $M(r)$ another arbitrary function. The condition \eqref{eq:LTBcondition} ensures that the only non-vanishing component of the Einstein tensor is $G^t{}_t \propto M'(r)$, i.e., the matter supporting the geometry is pure dust (pressureless matter), and thus perfect fluid in particular. As above, direct application of Prop.~\ref{perfect-fluid is RW} leads to the following result:

\begin{proposition}\label{prop:ineqLTB}
    None of the $k(t)$ metrics with non-constant $k(t)$ are locally isometric to the LTB geometries~\eqref{eq:LTBmetric}.
\end{proposition}

\section{Summary of results}
\label{Sec:results} 

In the following, we summarize the relevant results of this article.

\paragraph{Curvature and singularities.}

The $k(t)$-conformal metric for any $k(t)$ is locally conformally flat (see Prop.~\ref{prop:kconf_Weylzero}). The other $k(t)$ models are not (for non constant $k(t)$).

For non-constant $k(t)$ there is a curvature singularity at the upper limit of $r$ for $k(t)>0$ in the $k(t)$-warped and the $k(t)$-radial cases, and for $k(t)<0$ in the $k(t)$-conformal case (see, respectively, Prop.~\ref{prop:kwar_sing}, \ref{prop:krad_sing} and \ref{prop:kconf_sing}).

\paragraph{Extendibility.}

The domains of the three $k(t)$ metrics initially defined admit smooth extensions to $r=0$ independently of the sign of $k(t)$. For $k(t)\leq 0$, the three metrics are inextendible at the upper limit of $r$. However, in the regions where $k(t)>0$, i.e., when $t\in I_+$, we have that:
\begin{itemize}
    \item For $\gw$ the points $r = \pi/ \sqrt{k(t)}$ can be included in the manifold, so that each $t$-slice is a topological sphere and the region is diffeomorphic to $I_+\times \mathbb{S}^{n}$. The metric is singular therein, but we will not be concerned with this, taking into account the possibility of smoothening on a region $\mathcal{U}_\pi$  (Prop.~\ref{p_3metrics1}).
    
    \item For $\gc$, the points $r=\infty$ are  also  included so that the region is diffeomorphic to $I_+\times \mathbb{S}^{n}$ and the metric becomes smooth therein.

    \item For  $\gr$, the points $r=1/\sqrt{k(t)}$ are not included in the manifold so that each $t$-slice is a topological disk and  the region is diffeomorphic to $I_+\times \mathbb{R}^{n}$ (as so will be the whole spacetime). However, the value $r=1/\sqrt{k(t)}$ yields a boundary which is spacelike in the sense of Lem.~\ref{l_3metrics3}.
\end{itemize}

\paragraph{Global hyperbolicity.}

We have fully characterized the necessary and sufficient conditions for global hyperbolicity for the $k(t)$-warped and $k(t)$-conformal cases:
\begin{itemize}
    \item The $k(t)$-warped spacetime is globally hyperbolic if and only if it belongs to one of the cases in Thm.~\ref{thm:kwar_globhyp}.        
    \item The $k(t)$-conformal spacetime is globally hyperbolic if and only if it belongs to one of the cases in Thm.~\ref{thm:kcon_globhyp}.
\end{itemize}
However, for the $k(t)$-radial case, we identified some sufficient conditions that were provided in Thm.~\ref{thm:krad_globhyp}. The existence of more possibilities is explained in Rem.~\ref{rem:krad_otherGHcases}, but an exhaustive mathematical enumeration does not seem simple.

\paragraph{Isometries.}

\begin{itemize}

    \item The three $k(t)$ metrics are isometric to FLRW only when $k(t)$ is constant (Lem.~\ref{FLRW_kdot}). Moreover, this occurs if and only if the Einstein tensor of the $k(t)$ metric is that of a perfect fluid (Prop.~\ref{perfect-fluid is RW}).
    
    \item For generic functions $k(t)$ and $a(t)$, the Killing algebra of the geometries is $\mathrm{so}(n)$. Only when $a(t)=a_0 e^{bt}$ and $k(t)=k_0 e^{2bt}$, with $a_0$, $k_0$, and $b$ constant, do they admit a single additional Killing vector (see Thm.~\ref{resultKV}). 
    
    \item As an intermediate step, we exhaustively characterized the Killing vectors of a general spherically symmetric warped product in Prop.~\ref{res:general_spherical}. 

\end{itemize}

\paragraph{Comparison with other geometries.}

The different $k(t)$ metrics are not locally isometric to each other (see Props.~\ref{prop:ineqkt1} and \ref{prop:ineqkt2}), nor to any element lying in the Stephani class \eqref{eq:Stephani} (see Prop.~\ref{prop:ineqStephani}), nor to any element in the LTB class \eqref{eq:LTBmetric} (see Prop.~\ref{prop:ineqLTB}).

\section{Discussion and conclusions}
\label{Sec:conclusions} 

In this article, we have presented three geometries that might serve as cosmological models, each foliated by spatial hypersurfaces of constant but time-dependent curvature $k(t)$, where the sign of $k(t)$ may vary with time. We have analyzed and characterized these geometries both locally and globally.

From the local point of view, we have computed  their curvature properties and analyzed under which conditions curvature singularities appear. Furthermore, we have examined their  Killing vectors and shown that, aside from an exceptional case, all the geometries are only invariant under rotations. We have also shown under which conditions the geometries can be extended beyond the patch in which they are originally introduced. Furthermore, we have studied their global properties, in particular their global hyperbolicity.  

We have given necessary and sufficient conditions for the spacetime to be globally hyperbolic in the $k(t)$-warped and conformal  metrics, and give some sufficient conditions for the $k(t)$-radial, as a full classification of all the possible cases seems too complicated.  We also characterize for the $k(t)$-warped and conformal metrics all the cases when the topological transition occurs maintaining global hyperbolicity (Thm.~\ref{thm:kwar_globhyp}, Thm.~\ref{thm:kcon_globhyp}).  Then, for simple choices of $k(t)$, the following interpretation appears: 
 
\begin{quote} 
    The topological change of the $t$-slices  models an expansion starting at a finite boundaryless region with finite matter and  energy. After the topological change, the space would seem infinite by looking at  comoving observers. 
 
    However, predictability from the finite initial region (or any region, necessarily spatially finite, evolved from it) would be always preserved. 
\end{quote}
 
One can think of a huge expansion which might happen at very early $t$-times even though, for any Cauchy temporal function $\tau$, this ``early $t$-time'' along the expanding region lasts for arbitrarily big values of $\tau$ (thus, the expansion has not finished yet!). So, in the $k(t)$- warped and conformal cases:  
\begin{quote}
    One might arrive  to a present-day ``Euclidean'' $t$-space coming from a rather classic Big Bang.  Indeed, this would happen when $k(t) \simeq 0$ in the region $t\geq t_0$ for some $t_0$, meaning that the density of matter there (necessarily far from the asymptotic direction of expansion) is negligible. 
\end{quote} 

For the $k(t)$-radial metric one has  transitions similar to the other metrics when $k(t)\leq 0$ (Prop.~\ref{t_radial0}). However, when $k(t)>0$ somewhere, the situation is subtler, because the aforementioned possibility of extending the metric from the  half sphere  to a complete sphere depends crucially on the vanishing of $\dot{k}(t)$, yielding a dramatic change in the spacetime (Rem.~\ref{r_radial}).
Assuming additionally that $\dot{k}(t)$ does not vanish when $k(t)>0$, one has a {\em globally hyperbolic} spacetime foliated by half spheres (of  radius  increasing or decreasing with $t$, respectively)
which are not Cauchy hypersurfaces (Prop.~\ref{p_radial}, Thm.~\ref{thm:krad_globhyp}).

To finish, we mention some lines of work that might be interesting to pursue in the future. First of all, we stress that the phenomenon of sign changing spatial curvature can be obtained by using different parameterizations of constant curvature spaces and, thus other possibilities (see for example~\cite{Hestenes}) might be also worth of exploring. We also note that there is a specific choice of functions $a(t)$ and $k(t)$, identified in Sec.~\ref{Sec:isometries}, in which the metric exhibits an additional Killing field that would be interesting to further study. Finally, we note that it would be interesting to characterize the energy-momentum tensor associated with these geometries and we leave it for future work. 

\acknowledgments{
The authors would like to thank Jose Beltrán Jiménez, José M. M. Senovilla, Carlos Barceló, Luis J. Garay, Andrzej Krasinski, Robin Croft and Thomas Van Riet for useful discussions and feedback.
G. Garc\'ia-Moreno is supported by the MUR FIS2 Advanced Grant ET-NOW (CUP:~B53C25001080001), by the INFN TEONGRAV initiative and by the project Grant No. PID2023-149018NB-C43 funded by MCIN/AEI/10.13039/501100011033.
The work of B. Janssen and A. Jim\'enez Cano has been supported by the grant PID2022-140831NB-I00 funded by MCIN/AEI/10.13039/501100011033.
M. Mars acknowledges financial support under projects PID2024-158938NB-I00 (Spanish Ministerio de Ciencia e Innovación and FEDER “A way of making Europe”), SA097P24 (JCyL) and RED2022-134301-T funded by MCIN/AEI/10.13039/501100011033. 
R. Vera was supported by grant IT1628-22 from the Basque Government, and PID2021-123226NB-I00 funded by ``ERDF A way of making Europe'' and MCIN/AEI/10.13039/501100011033.
M. Sánchez was supported by PID2024-156031NB-I00 and the framework IMAG-María de Maeztu grant CEX2020-001105-M, both funded by MCIN/AEI/10.13039/50110001103.

The computations have been checked with xAct \cite{xAct}, a Mathematica package for tensorial symbolic calculus; the corresponding notebook is available upon request. 
Some algebraic computations have been performed with the help of the free PSL version of REDUCE.
}

\appendix 
\section{Properties of the function \texorpdfstring{$S_{k(t)}(r)$}{S}}
\label{app:appSk}

In this appendix we collect some expressions that are useful for calculations involving the curvature of the $k(t)$-warped case. They depend on the function $S_{k(t)}(r)$ introduced in Eq.~\eqref{eq:defSk} and its derivatives. First, we introduce the convenient notation:
\begin{equation}\label{eq:Str}
    S(t,r) := S_{k(t)}(r)\,.
\end{equation}
Let us now define the function:
\begin{equation}\label{eq:Ctr}
    C(t,r):=C_{{k(t)}}(r)=
        \begin{cases}
         \cos(\sqrt{k(t)}\,r) & \text{if}\ k(t)>0\\
          1 & \text{if}\ k(t)=0\\
          \cosh(\sqrt{-k(t)}\,r) & \text{if}\ k(t)<0   
        \end{cases} \quad .
\end{equation}
Then one can prove (for $k(t)\neq 0$):\footnote{\label{foot:k0case}
    The case $k(t)=0$ can also be recovered from these expressions by simply dropping all the terms with $\dot{k}$ and $\ddot{k}$.
    }
\begingroup
\begin{align}
    1 &= kS^2+C^2\,, \\   
    S' &= C\,, \\    
    S'' &= C' = -kS\,, \\  
    \dot{S} &= \frac{1}{2}\frac{\dot{k}}{k}(rC- S)\,, \label{eq:dotS} \\ 
    \dot{S}' &= \dot{C}= -\frac{1}{2}r\dot{k} S \,, \\ 
    \ddot{S} &= \frac{1}{2}S\left[\left(\frac{\ddot{k}}{k}-\frac{3\dot{k}^2}{2k^2}\right)\frac{rC-S}{S}-\frac{r^2 \dot{k}^2}{2k} \right]\,.
\end{align}

In App.~\ref{app:curvatures_kwar} and App.~\ref{app:congruences_kwar}, the function $C(t,r)$ always appears through the following combination:
\begin{equation}\label{eq:Ntr}
   N(t,r):=\frac{rC(t,r)}{S(t,r)}-1 =
        \begin{cases}
         \sqrt{k(t)}\,r\cot(\sqrt{k(t)}\,r)-1 & \text{if}\ k(t)>0\\
          0 & \text{if}\ k(t)=0\\
          \sqrt{-k(t)}\,r\coth(\sqrt{-k(t)}\,r)-1 & \text{if}\ k(t)<0   
        \end{cases} \quad .
\end{equation}
Observe that, by virtue of \eqref{eq:dotS}, this function fulfills
\begin{equation}
    \frac{\partial \ln S}{\partial t} = N \frac{\dot{k}}{2k}\,.
\end{equation}

\section{Curvature tensors for the \texorpdfstring{$k(t)$}{k(t)} metrics}
\label{app:curvatures}

The expressions \eqref{eq:R_WF1}-\eqref{eq:R_WF6} can be used to compute the explicit components of the three $k(t)$ metrics in the given chart $\{t,r,\theta^A\}$. We omit the components $\R_{tA}$ and $\R_{rA}$ since they are identically vanishing.

\subsection{Curvature tensors for the \texorpdfstring{$k(t)$}{k(t)}-warped metric}
\label{app:curvatures_kwar}

For the $k(t)$-warped metric $\gw$ we use the notation \eqref{eq:Str} and find for $k(t)\neq 0$:\footnote{
    The case $k(t)=0$ can be obtained, as mentioned in footnote~\ref{foot:k0case}, by simply dropping all the terms with $\dot{k}$ and $\ddot{k}$; in particular $\Xi_1 \equiv 0$ and $\Xi_2\equiv 0$.
    }
\begingroup
\allowdisplaybreaks
\begin{align}
    \R_{tt} &= -\frac{n \ddot{a}}{a} + \frac{n-1}{4}(\Xi_1(0,3,1)-\Xi_2(2,1)) \,, \\
    \R_{tr} &= \frac{n-1}{2} r \dot{k} \,,\\
    \R_{rr} &= (n-1)( k + \dot{a}^2 ) + a \ddot{a} + \frac{n-1}{4}a^2 \Xi_2(1,0) \,, \\
    \R_{AB} &=  \Bigg[ (n-1)(k + \dot{a}^2) + a \ddot{a}  - \frac{a^2}{4} \Big(\Xi_1(n-2,2n-1,n-1)-\Xi_2(2n-1,1)\Big)\Bigg] S^2 \gS_{AB}\,, \\
    \R &= n(n-1) \left(\frac{k}{a^2}+\frac{\dot{a}^2}{a^2}\right) +2n\frac{\ddot{a}}{a} - \frac{n-1}{4} \Big( \Xi_1(n-2,2(n+1),n)-2\Xi_2(n+1,1)\Big) \,, \label{eq:RicciScalarkwar}\\
    G_{tt} &= \frac{n-1}{2}\left[n\left(\frac{k}{a^2}+\frac{\dot{a}^2}{a^2}\right) - \frac{1}{4} \Big((n-2)\Xi_1(1,2,1)-2\Xi_2(n-1,0)\Big)\right]\,,\\
    G_{tr} &= \R_{tr}\,,\\
    G_{rr} &= -\frac{n-1}{2}\left[ (n-2)(k +  \dot{a}^2) + 2 a \ddot{a} - \frac{a^2}{4}\Big(\Xi_1(n-2,2(n+1),n)-2\Xi_2(n,1)\Big)
    \right]\,,\\
    G_{AB}&= \Bigg[- \frac{(n-1)(n-2)}{2}(k + \dot{a}^2)- (n-1) a \ddot{a} \nonumber\\
    &\qquad\qquad + \frac{a^2(n-2)}{8}\Big(\Xi_1(n-3, 2n,n-1) -2\Xi_2(n,1)\Big)\Bigg] S^2 \gS_{AB} \,,
\end{align}
\endgroup
where we are using the abbreviations:
\begin{align}
    \Xi_1(p_1,p_2,p_3)&:= \frac{\dot{k}^2}{k^2}\left[p_1 \left(1 - \frac{r^2}{S^2}\right)+p_2 N + p_3 k r^2\right] \,,\\
    \Xi_2(p_1,p_2) &:= 2 N  \left(p_1\frac{\dot{a}\dot{k}}{ak}+p_2 \frac{\ddot{k}}{k}\right)\,,
\end{align}
and $N$ is defined in \eqref{eq:Ntr}.

\subsection{Curvature tensors for the \texorpdfstring{$k(t)$}{k(t)}-conformal metric}

For the $k(t)$-conformal metric $\gc$ we find: 
\begingroup
\allowdisplaybreaks
\begin{align}
    \R_{tt} &= -n \left(\frac{\ddot{a}}{a} - \frac{r^2 (2 \dot{a} \dot{k} + a \ddot{k})}{a (1 + k r^2)} + \frac{2 r^4 \dot{k}^2}{(1 + k r^2)^2} \right) \,, \\
    \R_{tr} &=  \frac{2(n-1) r \dot{k}}{(1 + k r^2)^2} \,,\\
    \R_{rr} &=   \frac{4\big( (n-1)( k + \dot{a}^2) + a \ddot{a}\big)}{(1 + k r^2)^2} - \frac{4r^2 a (2n \dot{a} \dot{k} + a \ddot{k})}{(1 + k r^2)^3} + \frac{4(n+1) r^4 a^2 \dot{k}^2}{(1 + k r^2)^4}  \,, \\
    \R_{AB} &=   \R_{rr}\, r^2\gS_{AB} \,, \\
    \R &= n(n-1) \left(\frac{k}{a^2}+\frac{\dot{a}^2}{a^2}\right)+2n\frac{\ddot{a}}{a} -  \frac{2n r^2 \big((n+1) \dot{a} \dot{k} + a \ddot{k}\big)}{a (1 + k r^2)} +\frac{n(n+3) r^4 \dot{k}^2}{(1 + k r^2)^2} \,, \label{Eq:Ricci_kconformal}\\
    G_{tt} &= \frac{n-1}{2} \left( \frac{k}{a^2}+\frac{\dot{a}^2}{a^2} -  \frac{2 r^2 \dot{a} \dot{k}}{a(1 + k r^2)} + \frac{ r^4 \dot{k}^2}{(1 + k r^2)^2}\right)
    \,,\label{eq:Gtt_kconformal}\\
    G_{tr} &= \R_{tr}\,,\\
    G_{rr} &= 2(n-1)\left[
    -  \frac{(n-2) (k + \dot{a}^2) + 2 a \ddot{a}}{(1 + k r^2)^2} + \frac{2 a r^2 (n \dot{a} \dot{k} + a \ddot{k})}{(1 + k r^2)^3} - \frac{(n+2) a^2 r^4 \dot{k}^2}{(1 + k r^2)^4}
    \right] \,,\\
    G_{AB}&= G_{rr}\, r^2\gS_{AB} \,. 
\end{align}
\endgroup

\subsection{Curvature tensors for the \texorpdfstring{$k(t)$}{k(t)}-radial metric}

For the $k(t)$-radial metric $\gr$ we find:
\begingroup
\allowdisplaybreaks
\begin{align}
    \R_{tt} &= -\frac{n \ddot{a}}{a} - \frac{r^2 (2 \dot{a} \dot{k}+a \ddot{k})}{2 a (1 -  k r^2)}  - \frac{3 r^4 \dot{k}^2}{4 (1 - k r^2)^2} \,, \\
    \R_{tr} &= \frac{n-1}{2}\frac{r \dot{k}}{1 - k r^2}\,,\\
    \R_{rr} &=  \frac{(n-1)(k + \dot{a}^2) + a \ddot{a}}{1 - k r^2} +  \frac{ r^2 a \big((n+1) \dot{a} \dot{k} + a \ddot{k}\big)}{2 (1 - k r^2)^2}+\frac{3 r^4 a^2  \dot{k}^2}{4 (1 - k r^2)^3} \,, \\
    \R_{AB} &= \left[(n-1) (k + \dot{a}^2) + a \ddot{a} + \frac{ r^2 a \dot{a} \dot{k}}{2(1-  k r^2)} \right]r^2\gS_{AB} \,, \\
    \R &= n(n-1) \left(\frac{k}{a^2}+\frac{\dot{a}^2}{a^2}\right) + 2n \frac{\ddot{a}}{a} + \frac{r^2 \big((n+1) \dot{a} \dot{k} + a \ddot{k}\big)}{a (1 -  k r^2)} +\frac{3 r^4 \dot{k}^2}{2 (1 - k r^2)^2}  \,, \label{Eq:RicciScal_krad}\\
    G_{tt} &= \frac{n-1}{2} \left[n\left(\frac{k}{a^2}+\frac{\dot{a}^2}{a^2}\right) + \frac{r^2 \dot{a} \dot{k}}{a(1 -  k r^2)}\right]\,,\\
    G_{tr} &= \R_{tr}\,,\\
    G_{rr} &= -\frac{n-1}{2}\, \frac{(n-2)(k + \dot{a}^2) + 2 a \ddot{a}}{1 - k r^2}\,,\\
    G_{AB}&= - \left[\frac{n-1}{2}\big((n-2)(k + \dot{a}^2) + 2 a \ddot{a}\big) + \frac{ r^2 a(n \dot{a} \dot{k} + a \ddot{k})}{2(1- k r^2)}+ \frac{3 r^4 a^2 \dot{k}^2}{4 (1 - k r^2)^2}\right] r^2 \gS_{AB} \,.
\end{align}
\endgroup

\section{Expansion, shear and vorticity}
\label{app:congruence}

Here we report the values of the expansion and the shear tensor (we recall that the vorticity is zero) for the three different metrics. It can be easily checked that for $\dot{k} = 0$, we recover the FLRW result ($\Theta = n \dot{a}/a$, $\sigma_{\mu\nu} = 0$) in the three cases, as it should. 
\subsection{\texorpdfstring{$k(t)$}{k(t)}-warped metric}
\label{app:congruences_kwar}
For the $k(t)$-warped metric $\gw$ we find for $k\neq 0$:
\begin{align}
    \Theta  &=  n \frac{\dot{a}}{a}+ \frac{n-1}{2} N \frac{\dot{k}}{k} \, , \\
    \sigma_{rr}&= - \frac{(n-1)a^2}{2n}  N \frac{\dot{k}}{k}, \qquad 
    \sigma_{AB} = \frac{a^2S^2N}{2n} \frac{\dot{k}}{k}\gS_{AB} , \qquad 
    \sigma^2   = \frac{n-1}{4 n} N^2 \frac{\dot{k}^2}{k^2} \,,
\end{align}
where $S=S(t,r)$ and $N=N(t,r)$ are defined in \eqref{eq:Str} and \eqref{eq:Ntr}, respectively.

\subsection{\texorpdfstring{$k(t)$}{k(t)}-conformal metric}
\label{app:congruence_kconf}
For the $k(t)$-conformal metric $\gc$ we find:
\begin{align}
  \Theta & = n \frac{\dot{a}}{a} - \frac{n r^2 \dot{k}}{1+ kr^2}, \label{Eq:kconf_expansion} \\
  \sigma_{rr} & = 0 , \qquad     
  \sigma_{AB}  = 0, \qquad   
  \sigma^2 =  0 \,.
\end{align}
Notice that although the congruence is shear-free it does not violate Ellis theorem~\cite{Ellis1971} since the congruence is not that of the proper time of the perfect fluid supporting the configuration. 
\subsection{\texorpdfstring{$k(t)$}{k(t)}-radial metric}
For the $k(t)$-radial metric $\gr$ we find:
\begin{align}
    \Theta & = n \frac{\dot{a}}{a} + \frac{ r^2 \dot{k}}{2(1 -  k r^2)}, \label{Eq:Expansion_krad} \\
    \sigma_{rr} &  = \frac{(n-1) r^2 a^2 \dot{k}}{2n (1- k r^2 )^2}, \qquad 
    \sigma_{AB}  = - \frac{r^4 a^2 \dot{k}}{2n (1-kr^2)} \gS_{AB}, \qquad
    \sigma^2 = \frac{(n-1) \dot{k}^2 r^4}{4 n (1 - k r^2)^2} . \label{Eq:Shear}  
\end{align}
%

\bibliographystyle{JHEP}
\bibliography{kt_biblio}

\providecommand{\href}[2]{#2}\begingroup\raggedright\begin{thebibliography}{10}

\bibitem{ONeill1983}
B.~O'Neill, \emph{Semi-Riemannian Geometry With Applications to Relativity,
  103, Volume 103 (Pure and Applied Mathematics)}.
\newblock Academic Press, 1983.

\bibitem{Sanchez2023}
M.~S\'anchez, \emph{{A class of cosmological models with spatially constant
  sign-changing curvature}},
  \href{http://dx.doi.org/10.4171/pm/2099}{\emph{Portug. Math.} {\bfseries 80}
  (2023) 291--313}, [\href{https://arxiv.org/abs/2209.11184}{{\ttfamily
  2209.11184}}].

\bibitem{Avalos2022}
R.~Avalos, \emph{{On the rigidity of cosmological space-times}},
  \href{http://dx.doi.org/10.1007/s11005-023-01720-9}{\emph{Lett. Math. Phys.}
  {\bfseries 113} (2023) 98},
  [\href{https://arxiv.org/abs/2211.07013}{{\ttfamily 2211.07013}}].

\bibitem{Bergmann1981}
O.~Bergmann, \emph{A cosmological solution of the einstein equations with heat
  flow},
  \href{http://dx.doi.org/https://doi.org/10.1016/0375-9601(81)90782-9}{\emph{Physics
  Letters A} {\bfseries 82} (1981) 383}.

\bibitem{Mersini2007}
L.~Mersini-Houghton, Y.~Wang, P.~Mukherjee and E.~Kafexhiu, \emph{{Nontrivial
  Geometries: Bounds on the Curvature of the Universe}},
  \href{http://dx.doi.org/10.1016/j.astropartphys.2007.12.006}{\emph{Astropart.
  Phys.} {\bfseries 29} (2008) 167--173},
  [\href{https://arxiv.org/abs/0705.0332}{{\ttfamily 0705.0332}}].

\bibitem{Stichel2018}
P.~C. Stichel, \emph{{Analytical solutions for two inhomogeneous cosmological
  models with energy flow and dynamical curvature}},
  \href{http://dx.doi.org/10.1103/PhysRevD.98.104022}{\emph{Phys. Rev. D}
  {\bfseries 98} (2018) 104022},
  [\href{https://arxiv.org/abs/1805.08459}{{\ttfamily 1805.08459}}].

\bibitem{Wang2025}
D.~Wang, O.~Mena, S.~Capozziello and D.~Mota, \emph{{Do low-redshift
  observations open the doors to an open universe?}},
  \href{https://arxiv.org/abs/2512.19565}{{\ttfamily 2512.19565}}.

\bibitem{Larena2008}
J.~Larena, J.-M. Alimi, T.~Buchert, M.~Kunz and P.-S. Corasaniti,
  \emph{{Testing backreaction effects with observations}},
  \href{http://dx.doi.org/10.1103/PhysRevD.79.083011}{\emph{Phys. Rev. D}
  {\bfseries 79} (2009) 083011},
  [\href{https://arxiv.org/abs/0808.1161}{{\ttfamily 0808.1161}}].

\bibitem{Clifton2024}
T.~Clifton and N.~Hyatt, \emph{{A radical solution to the Hubble tension
  problem}}, \href{http://dx.doi.org/10.1088/1475-7516/2024/08/052}{\emph{JCAP}
  {\bfseries 08} (2024) 052},
  [\href{https://arxiv.org/abs/2404.08586}{{\ttfamily 2404.08586}}].

\bibitem{Raffai2025}
P.~Raffai, D.~E.~R. Kis, D.~A. K{\"o}dm{\"o}n, A.~Pataki, R.~L. B{\"o}ttger and
  G.~D{\'a}lya, \emph{{A Case for an Inhomogeneous Einstein-de Sitter
  Universe}},  \href{https://arxiv.org/abs/2511.03288}{{\ttfamily 2511.03288}}.

\bibitem{Wald1984}
R.~M. Wald, \emph{{General Relativity}}.
\newblock Chicago Univ. Pr., Chicago, USA, 1984,
  \href{http://dx.doi.org/10.7208/chicago/9780226870373.001.0001}{10.7208/chicago/9780226870373.001.0001}.

\bibitem{Ellis2012}
G.~F.~R. Ellis, R.~Maartens and M.~A.~H. MacCallum, \emph{Relativistic
  Cosmology}.
\newblock Cambridge University Press, 2012.

\bibitem{Choquet-Bruhat2009}
Y.~Choquet-Bruhat, \emph{{General Relativity and the Einstein Equations}}.
\newblock Oxford Mathematical Monographs. Oxford University Press, United
  Kingdom, 2009.

\bibitem{HawkingEllis}
S.~W. Hawking and G.~F.~R. Ellis, \emph{{The Large Scale Structure of
  Space-Time}}.
\newblock Cambridge Monographs on Mathematical Physics. Cambridge University
  Press, 2, 2023,
  \href{http://dx.doi.org/10.1017/9781009253161}{10.1017/9781009253161}.

\bibitem{Sgeroch}
M.~S\'anchez, \emph{Recent progress on the notion of global hyperbolicity},  in
  \emph{Advances in Lorentzian Geometry} (M.~Plaue, A.~D. Rendall and
  M.~Scherfner, eds.), vol.~49 of \emph{AMS/IP Studies in Advanced
  Mathematics}, p.~105.
\newblock American Mathematical Society, 2011.

\bibitem{Spenrose}
M.~S{\'a}nchez, \emph{Globally hyperbolic spacetimes: slicings, boundaries and
  counterexamples}, {\emph{General Relativity and Gravitation} {\bfseries 54}
  (2022) 124}.

\bibitem{AFS}
L.~Ak{\'e}, J.~L. Flores and M.~S{\'a}nchez, \emph{Structure of globally
  hyperbolic spacetimes with timelike boundary}, {\emph{Rev. Mat. Iberoam}
  {\bfseries 37} (2021) 45--94}.

\bibitem{Solis}
D.~A. Solis, \emph{Global properties of asymptotically de Sitter and Anti de
  Sitter spacetimes}.
\newblock University of Miami, 2006.

\bibitem{Ge}
R.~Geroch, \emph{Domain of dependence}, {\emph{Journal of Mathematical Physics}
  {\bfseries 11} (1970) 437--449}.

\bibitem{BS03}
A.~N. Bernal and M.~S{\'a}nchez, \emph{On smooth cauchy hypersurfaces and
  geroch's splitting theorem}, {\emph{Comm. Math. Phys.} {\bfseries 343} (2003)
  461--470}.

\bibitem{BS05}
A.~N. Bernal and M.~S{\'a}nchez, \emph{Smoothness of time functions and the
  metric splitting of globally hyperbolic spacetimes}, {\emph{Communications in
  mathematical physics} {\bfseries 257} (2005) 43--50}.

\bibitem{BS06}
A.~N. Bernal and M.~S{\'a}nchez, \emph{Further results on the smoothability of
  cauchy hypersurfaces and cauchy time functions}, {\emph{Letters in
  Mathematical Physics} {\bfseries 77} (2006) 183--197}.

\bibitem{Aetal}
L.~Andersson, G.~J. Galloway and R.~Howard, \emph{The cosmological time
  function}, {\emph{Classical and quantum gravity} {\bfseries 15} (1998)
  309--322}.

\bibitem{RZ}
S.~E. Rugh and H.~Zinkernagel, \emph{On the physical basis of cosmic time},
  {\emph{Studies in History and Philosophy of Science Part B: Studies in
  History and Philosophy of Modern Physics} {\bfseries 40} (2009) 1--19}.

\bibitem{Krasinski1981}
A.~Krasinski, \emph{{Space-times with spherically symmetric hypersurfaces}},
  \href{http://dx.doi.org/10.1007/BF00756363}{\emph{Gen. Rel. Grav.} {\bfseries
  13} (1981) 1021--1035}.

\bibitem{Krasinski1983}
A.~Krasinski, \emph{{On the Global Geometry of the Stephani Universe}},
  \href{http://dx.doi.org/10.1007/BF00759044}{\emph{Gen. Rel. Grav.} {\bfseries
  15} (1983) 673--689}.

\bibitem{Krasinski1997}
A.~Krasinski, \emph{{Inhomogeneous cosmological models}}.
\newblock Cambridge Univ. Press, Cambridge, UK, 3, 2011.

\bibitem{Dabrowski1991}
M.~P. {Dabrowski}, \emph{{Isometric embedding of the spherically symmetric
  Stephani universe: Some explicit examples.}},
  \href{http://dx.doi.org/10.1063/1.530166}{\emph{Journal of Mathematical
  Physics} {\bfseries 34} (Apr., 1993) 1447--1479}.

\bibitem{Cook1975}
M.~W. {Cook}, \emph{{On a class of exact spherically symmetric solution to the
  Einstein gravitational field equations}},
  \href{http://dx.doi.org/10.1071/PH750413}{\emph{Australian Journal of
  Physics} {\bfseries 28} (Aug., 1975) 413--422}.

\bibitem{Sussman1988}
R.~A. Sussman, \emph{On spherically symmetric shear‐free perfect fluid
  configurations (neutral and charged). iii. global view},
  \href{http://dx.doi.org/10.1063/1.527962}{\emph{Journal of Mathematical
  Physics} {\bfseries 29} (05, 1988) 1177--1211},
  [\href{https://arxiv.org/abs/https://pubs.aip.org/aip/jmp/article-pdf/29/5/1177/19096939/1177\_1\_online.pdf}{{\ttfamily
  https://pubs.aip.org/aip/jmp/article-pdf/29/5/1177/19096939/1177\_1\_online.pdf}}].

\bibitem{Lachieze-Rey1995}
M.~Lachieze-Rey and J.-P. Luminet, \emph{{Cosmic topology}},
  \href{http://dx.doi.org/10.1016/0370-1573(94)00085-H}{\emph{Phys. Rept.}
  {\bfseries 254} (1995) 135--214},
  [\href{https://arxiv.org/abs/gr-qc/9605010}{{\ttfamily gr-qc/9605010}}].

\bibitem{MarsVera2024}
M.~Mars and R.~Vera, \emph{New characterization of {R}obertson--{W}alker
  geometries involving a single timelike curve}, {\emph{Journal of Physics A:
  Mathematical and Theoretical} {\bfseries 57} (2024) 355402}.

\bibitem{MinSan}
E.~Minguzzi, M.~S{\'a}nchez et~al., \emph{The causal hierarchy of spacetimes},
  {\emph{Recent developments in pseudo-Riemannian geometry} {\bfseries 4}
  (2008) 299--358}.

\bibitem{Beem}
J.~K. Beem, \emph{Global Lorentzian geometry}.
\newblock Routledge, 2017.

\bibitem{FHS_ATMP}
J.~L. Flores, J.~Herrera and M.~S{\'a}nchez, \emph{On the final definition of
  the causal boundary and its relation with the conformal boundary},
  {\emph{Adv. Theor. Math. Phys.} {\bfseries 15} (2011) 991--1058}.

\bibitem{Sanchez1999}
M.~Sánchez, \emph{{On the geometry of generalized Robertson-Walker spacetimes:
  curvature and Killing fields}},
  \href{http://dx.doi.org/https://doi.org/10.1016/S0393-0440(98)00061-8}{\emph{Journal
  of Geometry and Physics} {\bfseries 31} (1999) 1--15}.

\bibitem{Stephani2003}
H.~Stephani, D.~Kramer, M.~A.~H. MacCallum, C.~Hoenselaers and E.~Herlt,
  \emph{{Exact solutions of Einstein's field equations}}.
\newblock Cambridge Monographs on Mathematical Physics. Cambridge Univ. Press,
  Cambridge, 2003,
  \href{http://dx.doi.org/10.1017/CBO9780511535185}{10.1017/CBO9780511535185}.

\bibitem{Clarkson1999}
C.~A. Clarkson and R.~Barrett, \emph{{Does the isotropy of the CMB imply a
  homogeneous universe? Some generalized EGS theorems}},
  \href{http://dx.doi.org/10.1088/0264-9381/16/12/302}{\emph{Class. Quant.
  Grav.} {\bfseries 16} (1999) 3781--3794},
  [\href{https://arxiv.org/abs/gr-qc/9906097}{{\ttfamily gr-qc/9906097}}].

\bibitem{Barnes1998}
A.~Barnes, \emph{Symmetries of the {S}tephani universes},
  \href{http://dx.doi.org/10.1088/0264-9381/15/10/012}{\emph{Classical and
  Quantum Gravity} {\bfseries 15} (oct, 1998) 3061}.

\bibitem{Hestenes}
H.~Li, D.~Hestenes and A.~Rockwood, \emph{A universal model for conformal
  geometries of euclidean, spherical and double-hyperbolic spaces},  in
  \emph{Geometric Computing with Clifford Algebras: Theoretical Foundations and
  Applications in Computer Vision and Robotics}, pp.~77--104.
\newblock Springer, 2001.

\bibitem{xAct}
J.~M. Martin-Garcia, A.~Garc{\'\i}a-Parrado, A.~Stecchina, B.~Wardell,
  C.~Pitrou, D.~Brizuela et~al., \emph{{xAct: Efficient tensor computer algebra
  for the Wolfram Language}}, {\emph{{\tt \url{http://www.xact.es}}} (latest
  version Oct. 2021) }.

\bibitem{Ellis1971}
G.~F.~R. Ellis, \emph{{Relativistic cosmology}},
  \href{http://dx.doi.org/10.1007/s10714-009-0760-7}{\emph{Proc. Int. Sch.
  Phys. Fermi} {\bfseries 47} (1971) 104--182}.

\end{thebibliography}\endgroup

\end{document}